\ifdefined\compileAAAI
\renewcommand{\input}[1]{}
\else

\documentclass[a4paper, 11pt]{article}

\usepackage{comment}
\usepackage{fullpage} 
\usepackage{amsmath}
\usepackage{amssymb}
\usepackage{bbm}
\usepackage{mathtools}
\usepackage{setspace}
\usepackage{graphicx}
\usepackage{hyperref}
\hypersetup{colorlinks,allcolors=black}

\usepackage{algorithm}
\usepackage{algpseudocode}
\usepackage{amsthm}
\usepackage{thm-restate}
\usepackage{enumerate}
\usepackage[normalem]{ulem}
\usepackage{float}
\usepackage{afterpage}

\allowdisplaybreaks
\usepackage{pgfplots,tikz}
\pgfplotsset{compat=1.17}
\usetikzlibrary{patterns,pgfplots.fillbetween}

\DeclareMathOperator*{\argmax}{arg\,max}

\def\EQ#1{\begin{equation*}\begin{split}#1\end{split}\end{equation*}}

\def\colon{\,:\,}
\def\bar{\,|\,}

\def\1{\mathbbm{1}}

\newcommand{\reals}{\mathbb{R}}
\usepackage{color-edits}

\addauthor{TP}{red}

\addauthor{MF}{blue}

\addauthor{AF}{purple}

\addauthor{SM}{green!50!black}

\addauthor{TODO}{brown}

\usepackage[numbers,square]{natbib}

\newtheorem{theorem}{Theorem}
\newtheorem{lemma}{Lemma}[section]

\theoremstyle{definition}
\newtheorem{definition}[lemma]{Definition}
\newtheorem{example}[definition]{Example}

\newcommand{\agents}{[n]}
\newcommand{\nw}{\mathrm{NW}}
\newcommand{\sw}{\mathrm{SW}}

\usepackage{authblk}
\author[1,2]{Michal Feldman}
\author[1]{Simon Mauras}
\author[1]{Tomasz Ponitka}
\affil[1]{Tel Aviv University}
\affil[2]{Microsoft ILDC}

\title{On Optimal Tradeoffs between EFX and Nash Welfare\thanks{This work was supported by the European Research Council (ERC) under the European Union's Horizon 2020 research and innovation program (grant agreement No. 866132), by the Israel Science Foundation (grant number 317/17), by an Amazon Research Award, by the NSF-BSF (grant number 2020788), and by a grant from the Tel Aviv University Center for AI and Data Science (TAD).

We also thank Amos Fiat for invaluable discussions, and Vishnu V. Narayan for his helpful feedback.}
}

\begin{document}

\maketitle

\begin{abstract}
A major problem in fair division is how to allocate a set of indivisible resources among agents fairly and efficiently. {The goal of this work is to characterize the} tradeoffs between 
{two} well-studied measures of fairness and efficiency --- {\em envy freeness up to any item} (EFX) for fairness, and {{\em Nash welfare}} for efficiency
--- {by saying, for given constants $\alpha$ and $\beta$, whether there exists an $\alpha$-EFX allocation that guarantees a $\beta$-fraction of the maximum Nash welfare ($\beta$-MNW).}
For additive valuations, we show that for any $\alpha \in [0,1]$, there exists a partial allocation that is $\alpha$-EFX and $\frac{1}{\alpha+1}$-MNW{, and this tradeoff is tight (for any $\alpha$)}. 
We also show that for $\alpha\in[0,\varphi-1 \approx 0.618]$ these partial allocations can be turned into complete allocations where all items are assigned. Furthermore, for any $\alpha \in [0, 1/2]$, we show that the tight tradeoff of $\alpha$-EFX and $\frac{1}{\alpha+1}$-MNW with complete allocations holds for the more general setting of {\em subadditive} valuations.
{Our results improve upon the current state of the art, for both additive and subadditive valuations, and match the}
best-known approximations of EFX under complete allocations, regardless of Nash welfare guarantees.
Notably, our {constructions} for additive valuations also provide EF1 and constant approximations for maximin share guarantees. 
\end{abstract}

\fi

\section{Introduction}

A common resource allocation setting has $m$ goods and $n$ agents with (possibly different) preferences over the set of goods. 
One of the biggest questions in such scenarios, dating back to \citet*{S1948}, is how to allocate goods {\em fairly} among agents, given their preferences.
This question has attracted extensive attention in many communities, {which led} to {multiple} {different} notions of fairness, capturing various philosophies of what {it} means {to be fair}. 
Another central goal is to divide the items {\em efficiently} so that the collective welfare of the agents, represented by some efficiency measure, is maximized. The aim of this paper is to analyze the extent to which fairness and efficiency can be achieved simultaneously for two well-studied notions of fairness and efficiency.
We study this question in a resource allocation setting with indivisible goods, described below.

\vspace{0.05in}
{\bf Resource allocation problems.}
A resource allocation problem is given by a set of $m$ indivisible goods, a set of $n$ agents{, and} a valuation function $v_i:2^{[m]}\rightarrow \reals^{\geq 0}$ {for every agent $i$}, which assigns a real value $v_i(S)$ to every bundle of goods $S\subseteq [m]$.
{The most common type of valuation functions are {\em additive} valuations, where} there exist individual values $v_{ig}$ such that $v_i(S)=\sum_{g \in S} v_{ig}$ for every bundle $S$ {and agent $i$}. We also consider the widely studied class of {\em subadditive} valuations, where $v_i(S \cup T) \leq v_i(S) + v_i(T)$ for every $S,T \subseteq [m]$. Subadditive valuations constitute the frontier of complement-free valuations \cite{LLN06}.

An allocation of $m$ goods amongst $n$ agents is given by a vector $X=(X_1, \ldots, X_n)$ where the sets $X_i$ and $X_j$ are disjoint for every pair of {distinct} agents $i$ {and} $j$. { The set $X_i$ represents the set of goods allocated to agent $i$.}
An allocation is said to be {\em complete} if $\bigcup_i X_i = [m]$, and {\em partial} if some items {might be} unallocated.
{In resource allocation problems,} one is {typically} interested in complete allocations{,} but in some contexts{,} including in the context of this paper{,} {it also makes sense to} consider partial allocations. 

\vspace{0.05in}
{\bf Fairness notions.}
{The problem of finding allocations that are fair is called
fair division.}
{Many {different definitions of fairness} have been {proposed} in the fair division domain, see \cite{ABFV22} for a survey.}
{One of the most compelling {of those} notions} is that of {\em envy-freeness}, introduced by Foley \cite{F66}. 
An allocation $X=(X_1, \ldots, X_n)$ is envy-free (EF) if for every pair of agents $i$ {and} $j$, it holds that $v_i(X_i) \geq v_i(X_j)$, namely, every agent (weakly) prefers her own allocation over that of any other agent.
In the context of envy-freeness, it is usually required that the allocations are complete. 
Indeed, without this constraint, every instance vacuously admits an EF allocation, namely one where no item is allocated. 
 
A major weakness of EF is that even the simplest settings may not admit complete EF allocations. For example, 
{if there is a single good and two agents with value $1$ for this good, }
then in any complete allocation, only one of the agents gets the {single} good {and} the other {agent is} envious. 

Consequently, various relaxations of EF have been {introduced}. 
{First, consider envy-freeness up to one item (EF1), defined by \citet*{B11}.}
An allocation $X$ is EF1 if for every pair of agents $i$ {and} $j$, there exists an item $g \in X_j$, such that $v_i(X_i) \geq v_i(X_j - g)${, n}amely, agent $i$ does not envy agent $j$ after removing some item from $j$'s bundle.
Using a {well-known} technique {of eliminating} envy cycles, it is not too difficult to show that a complete EF1 allocation always exists \cite{LMMS04}.
For concreteness, consider the following example. 

\begin{example}
\label{ex:simple-ex}
Consider a setting with three items, $\{a,b,c\}$, and two agents, $\{1,2\}$, and suppose that both agents have identical additive valuations, with values $v(a)=v(b)=1$ and $v(c)=2$.
\end{example}

\noindent In Example~\ref{ex:simple-ex}, the allocation $X_1=\{a\}$ and $X_2=\{b,c\}$ is EF1. Indeed, after removing $c$ from $X_2$, agent 1 has no envy towards agent 2.

A stronger notion than EF1 is {
envy freeness up to {\em any} item (EFX), introduced by \citet*{CKMPSW16}.}
An allocation $X$ is EFX if for every pair of agents $i$ {and} $j$, and {\em every} item $g \in X_j$, it holds that $v_i(X_i) \geq v_i(X_j - g)$.
{This} is a stronger {condition} than EF1 since the requirement for no envy applies to the removal of {\em any} item from $j$'s bundle.
For instance, in Example~\ref{ex:simple-ex}, the allocation $X_1=\{a\}$ and $X_2=\{b,c\}$ is not EFX, since after removing $b$ from $X_2$, agent $1$ still envies agent $2$.

Unfortunately, complete EFX allocations are not known to exist even for additive valuations, except in several special cases such as identical valuations \cite{PR18}, identical items \cite{LMMS04}, or the case of three agents \cite{CGM20}. 
Arguably, the existence of EFX allocations is the most enigmatic problem in the fair division domain~\cite{P20}.

A natural step to approach the existence of EFX allocations is to consider the approximate notion of $\alpha$-EFX,  defined by \citet*{PR18}.
An allocation $X$ is $\alpha$-EFX, for {some} $\alpha \in [0,1]$, if for every pair of agents $i$ {and} $j$, and {every} item $g \in X_j$, it holds that $v_i(X_i) \geq \alpha \cdot v_i(X_j - g)$.
The existence of $\alpha$-EFX allocations has been studied for several classes of valuations.
In particular, previous work has established (i) the existence of $(\varphi-1)$-EFX allocations for additive valuations, where $\varphi \approx 1.618$ is the golden ratio~\cite{AMN20}, and (ii) the existence of $1/2$-EFX allocations for subadditive valuations \cite{PR18}. 

\vspace{0.05in}
{\bf Efficiency notions.}
Common measures {that capture the efficiency of an allocation} are: (i) social welfare --- the sum of agent values, $\sw(X)=\sum_{i \in \agents}v_i(X_i)$, and (ii) Nash welfare --- the geometric mean of agent values, $\nw(X)=\prod_{i\in \agents}v_i(X_i)^{1/n}$.
In this paper, as in many previous studies, we consider the Nash welfare notion.

{In the literature, Nash welfare is often described as a measure interpolating between social welfare, which captures efficiency, and egalitarian welfare, defined as $\mathrm{EW}(X) = \min_{i \in \agents} v_i(X_i)$, which captures fairness \cite{CKMPSW16}. From that perspective, Nash welfare simultaneously serves as a measure of both fairness and efficiency
}

{The most important property of the Nash welfare is that it} encourages more balanced allocations relative to social welfare.
For example, consider a setting with two agents and two items, where every agent values every item at $1$. The unique maximum Nash welfare (MNW) allocation is the one that allocates one item per agent. In terms of social welfare, however, every allocation is equally good, including one that gives two items to one of the agents and none to the other, as all complete allocations have social welfare of $2$.

{Remarkably}, under additive valuations, every allocation that maximizes Nash welfare is EF1~\cite{CKMPSW16}. {In fact, maximizing Nash welfare is the only welfarist rule satisfying EF1~\cite{YS23}.} 
{Similar results have been shown for other valuation classes: every allocation that maximizes Nash welfare is $1/4$-EF1 when the valuations are subadditive \cite{WLG21}, and it is EFX when the valuations are additive and bi-valued \cite{ABFHV20}, or when they are submodular and dichotomous \cite{BEF21}.}

An allocation $X$ is said to be $\beta$-max Nash welfare ($\beta$-MNW) if the Nash welfare of $X$ is at least a $\beta$ fraction of the maximum Nash welfare.

\vspace{0.05in}
{\bf Fairness vs efficiency trade-off.}
A natural question is whether fairness and efficiency can be achieved simultaneously. \citet*{BLMS19} studied this question with respect to {the efficiency measure of the social welfare} for {instances with} additive valuations. Specifically, the provided bounds for the ``price of fairness" with respect to several fairness notions; namely, the  fraction of the maximum social welfare that can be achieved, when constrained by the corresponding fairness {condition}. 

They showed that the price of fairness of EFX is $\Omega(\sqrt{n})$, i.e., there are instances with additive valuations in which no EFX allocation can achieve more than a $O(1/\sqrt{n})$ fraction of the optimal social welfare. 
On the other hand, it is possible to find (partial) EFX allocations that obtain a constant fraction of the maximum Nash welfare \cite{CGH19}.
This motivates the study of optimal tradeoffs between $\alpha$-EFX and $\beta$-MNW,  {which is the focus of this paper.}

Notably, the original motivation for considering $\alpha$-EFX has been the embarrassment around the EFX existence problem, and the original motivation for considering 
$\beta$-MNW has been the NP-hardness of maximizing Nash welfare \cite{RE09,NRR12a}. 
As it turns out, an equally important motivation for studying approximate notions of EFX and MNW is the fact that one may come at the expense of another.
Thus, understanding the tradeoffs between these fairness and efficiency measures is crucial when designing a resource allocation scheme. 

Prior to our work, the following tradeoffs (demonstrated in Figure~\ref{fig:results}) have been known:
\vspace{-0.02in}
\begin{itemize}
    \item  instances with additive valuations admit a partial EFX allocation that is $1/2$-MNW \cite{CGH19}.
    \item  instances with additive  valuations admit a complete allocation that is $(\varphi-1)$-EFX (with no Nash welfare guarantees) \cite{AMN20}.
    \item  instances with subadditive valuations admit a complete allocation that is $1/2$-EFX and $1/2$-MNW \cite{GHLVV22}.
\end{itemize}

As described in the next section, we extend these results {to} give a more complete picture of the optimal tradeoffs between $\alpha$-EFX and $\beta$-MNW.

\subsection{Our Results}
In this paper we provide results on the optimal trade-offs between approximate EFX and approximate maximum Nash welfare (MNW), for both additive valuations and subadditive valuations. 
Our results are demonstrated in Figure~\ref{fig:results}, where the left and right figures correspond to additive and subadditive valuations, respectively.

\begin{figure}[h!]
    \begin{minipage}{.49\textwidth}
\hspace{2.5cm} Additive valuations\\
\begin{tikzpicture}
\begin{axis}[
  width=\textwidth, height=\textwidth,
  clip = false,
  xtick={0,0.618,1},
  ytick={0,0.5,0.618,1},
  xticklabels = {$0$,$1/2$,$\varphi-1$,$1$},
  yticklabels = {$0$,$\varphi-1$,$1$},
  xlabel={$\alpha$-EFX},
  ylabel={$\beta$-MNW},
  xlabel style={at={(axis description cs:0.25,0)}},
  ylabel style={at={(axis description cs:0,0.25)}},
  xmin=0,  xmax=1, ymin=0,  ymax=1]
\addplot [name path=f,domain=0:1] {1/(1+x)};
\path[name path=topaxis] (axis cs:0,1) -- (axis cs:1,1);
\path[name path=botaxis] (axis cs:0,0) -- (axis cs:1,0);
\tikzfillbetween[of=f and topaxis,on layer=main] {pattern=north east lines};
\tikzfillbetween[of=botaxis and f,on layer=background,soft clip={domain=0:.618}] {black!5!white};
\draw[dotted] (0,0.618) -- (0.618,0.618);
\draw[dotted] (0,0.5) -- (1,0.5);
\draw (0.618,0.618) -- (0.618,0);
\node[fill=white] at (0.6,0.9)
{\footnotesize no partial allocation};
\node at (0.15,0.75)
{\footnotesize $\beta = \frac{1}{1+\alpha}$};
\node at (0.3,0.25) {complete};
\node at (0.3,0.18) {allocations};
\node at (0.8,0.25) {partial};
\node at (0.8,0.18) {allocations};
\begin{scope}[blue]
    \node[anchor=north east] (CGH19) at (0.9,0.45) {\cite{CGH19}};
    \draw[->] (CGH19.5) to (0.99,0.49);
    \node[anchor=south] (AMN20) at (0.5,0.05) {\cite{AMN20}};
    \draw[->] (AMN20) to (0.608,0.01);
    \fill (.618,0) circle (0.1cm); 
    \draw[very thick] (1,0.5) circle (0.1cm); 
    \draw[very thick] (1,0.51) -- (1,1);
\end{scope}
\draw[->] (0.14,0.79) -- (0.17,0.82);
\end{axis}
\end{tikzpicture}
\end{minipage}
\hfill
\begin{minipage}{.49\textwidth}
\hspace{1.8cm} Subadditive valuations\\
\begin{tikzpicture}
\begin{axis}[
  width=\textwidth, height=\textwidth,
  clip = false,
  xtick={0,0.5,1},
  ytick={0,0.5,0.67,1},
  xlabel={$\alpha$-EFX},
  ylabel={$\beta$-MNW},
  xticklabels = {$0$,$1/2$,$1$},
  yticklabels = {$0$,$1/2$,$2/3$,$1$},
  xlabel style={at={(axis description cs:0.25,0)}},
  ylabel style={at={(axis description cs:0,0.25)}},
  xmin=0,  xmax=1, ymin=0,  ymax=1]
\addplot [name path=f,domain=0:1] {1/(1+x)};
\path[name path=topaxis] (axis cs:0,1) -- (axis cs:1,1);
\path[name path=botaxis] (axis cs:0,0) -- (axis cs:1,0);
\tikzfillbetween[of=f and topaxis,on layer=main] {pattern=north east lines};
\tikzfillbetween[of=botaxis and f,on layer=background,soft clip={domain=0:.5}] {black!5!white};
\draw[dotted] (0,0.667) -- (0.5,0.677);
\draw[dotted] (0,0.5) -- (1,0.5);
\draw (0.5,0.67) -- (0.5,0);
\node at (0.25,0.25) {complete};
\node at (0.25,0.18) {allocations};
\node at (0.75,0.25) {?};
\node[fill=white] at (0.6,0.9)
{\footnotesize no partial allocation};
\node at (0.15,0.75)
{\footnotesize $\beta = \frac{1}{1+\alpha}$};
\begin{scope}[blue]
    \node[anchor=north east] (GHLVV22) at (0.42,0.45) {\cite{GHLVV22}};
    \draw[->] (GHLVV22.5) to (0.49,0.49);
    \fill (.5,0.5) circle (0.1cm); 
\end{scope}
\draw[->] (0.14,0.79) -- (0.17,0.82);
\end{axis}
\end{tikzpicture}
\end{minipage}
    \caption{Trade-off between the existence of  $\alpha$-EFX and $\beta$-MNW allocations, for additive (left) and subadditive (right) valuations. 
    Previous positive results are represented by dots (complete dots for complete allocations, and hollow dots for partial allocations). 
    Dots that lie on the curves belong to the regions of positive results. All positive results for additive valuations also guarantee~EF1.}
    \label{fig:results}
\end{figure}
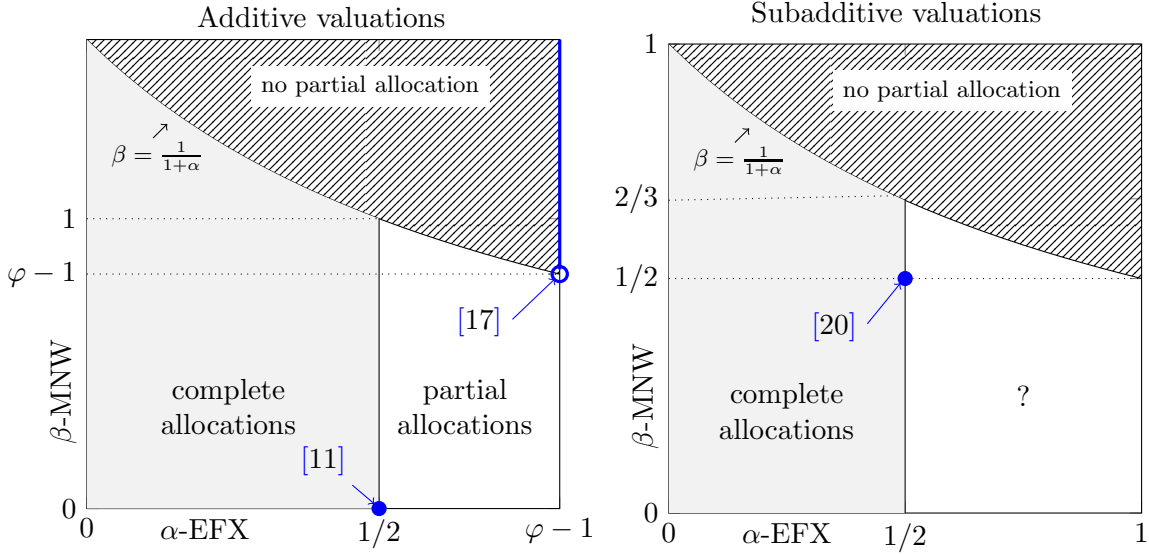

Our first result gives existence guarantees on  partial allocations with approximate EFX and approximate MNW {(and EF1)}, for additive valuations. 

\begin{restatable}
{theorem}{thmadditivepartial}
    Every instance with additive valuations admits a partial allocation that is $\alpha$-EFX{, EF1,} and $\frac{1}{\alpha+1}$-MNW, for every $0 \leq \alpha \leq 1$.
    \label{thm:additive-partial}
\end{restatable}

For $0 \leq \alpha \leq \varphi - 1 \approx 0.618$, we show that our partial allocations can be turned into complete ones without any loss. 
This is cast in the following theorem.

\begin{restatable}
{theorem}{thmadditivecomplete}
    Every instance with additive valuations admits
    a complete allocation that is $\alpha$-EFX{, EF1,} and $\frac{1}{\alpha+1}$-MNW, for every $0 \leq \alpha \leq \varphi - 1 \approx 0.618$.\label{thm:complete_additive}
\end{restatable}

In particular, Theorem~\ref{thm:complete_additive}
extends the existence of $(\varphi-1)$-EFX complete allocation \cite{AMN20} to $(\varphi-1)$-EFX complete allocation that is also $(\varphi-1)$-MNW.
Note that, by Theorem~\ref{thm:additive-upper} below, $(\varphi-1)$ is the highest possible MNW approximation of a (complete or partial) $(\varphi-1)$-EFX allocation.

Our final positive result shows that for any $0 \leq \alpha \leq 1/2$, the same trade-off between EFX and MNW approximation extends to subadditive valuations.

\begin{restatable}
{theorem}{thmsubadditivecomplete}
    Every instance with subadditive valuations admits a complete allocation that is $\alpha$-EFX and $\frac{1}{\alpha+1}$-MNW, for every $0 \leq \alpha \leq 1/2$.
    \label{thm:complete_subadditive}
\end{restatable}

In particular, Theorem~\ref{thm:complete_subadditive} extends the existence of {a} $1/2$-EFX {and} $1/2$-MNW complete allocation by \cite{GHLVV22} to {the existence of a} $1/2$-EFX complete allocation that is also  $2/3$-MNW. 
Note that, by Theorem~\ref{thm:additive-upper} below, $2/3$ is the highest possible MNW approximation of a (complete or partial) $1/2$-EFX allocation, even for additive valuations.

More generally, our tradeoffs are tight, as the following theorem shows.

\begin{restatable}[Impossibility]
{theorem}{thmlowerbound}
\label{thm:additive-upper}
For every $0 < \alpha \leq 1$ and $\beta > \frac{1}{\alpha+1}$, there exists an instance with additive (and hence {also} subadditive) valuations that admits no allocation (even partial) that is $\alpha$-EFX and $\beta$-MNW.
\end{restatable}

\textbf{Computational remarks.}
While our positive results 
(Theorems~\ref{thm:additive-partial},~\ref{thm:complete_additive}, and~\ref{thm:complete_subadditive}) are stated as existence results, to prove the existence of allocations with the stated guarantees, we construct polynomial-time algorithms that find such allocations, given a max Nash welfare (MNW) allocation as input. 

Moreover, the algorithms used to prove Theorems~\ref{thm:additive-partial} and~\ref{thm:complete_subadditive} apply also when given an {\em arbitrary} allocation as input. 
In particular, these are poly-time algorithms which, given an arbitrary allocation $X$ as input, produce an $\alpha$-EFX allocation that gives at least a $1/(\alpha+1)$ fraction of the Nash welfare of $X$, under the corresponding conditions. 

Thus, when given black-box access to an algorithm that computes a $\beta$-MNW allocation, they provide an $\alpha$-EFX allocation that is also $\beta/(\alpha+1)$-MNW. 
This extension is important in light of the fact that finding a MNW allocation is NP-hard, even for additive valuations \cite{RE09}, while constant approximation algorithms exist, even for subadditive valuations. 

{In particular, {using} this extension, combined with the 
$(e^{-1/e}-\varepsilon)$-MNW approximation for additive {valuations} (\citet*{BKV18}) and the {$(1/4-\varepsilon)$}-MNW approximation for submodular {valuations} (\citet*{GHLVV22}) {and the constant-factor-MNW approximation for subadditive valuations (\citet*{DLRV23})}, our algorithms find in polynomial time {a partial $\alpha$-EFX allocation with  constant-factor-MNW, for any $0 \leq \alpha \leq 1/2$  when valuations are {subadditive,}
and any $0 \leq \alpha \leq 1$ when valuations are additive.}
See Appendix~\ref{sec:comp} for more details.
}\\

{
\textbf{Maximin share guarantees.} Even though the allocations that we construct in the proof of Theorem~\ref{thm:complete_additive} are designed with EFX and MNW in mind, they also satisfy other desirable fairness notions related to the maximin share guarantee. More specifically, on top of $\alpha$-EFX, EF1, and $\frac{1}{\alpha+1}$-MNW, these allocations are also $\frac{\alpha}{\alpha^2+1}$-GMMS (and hence $\frac{\alpha}{\alpha^2+1}$-MMS) and $(\varphi-1)$-PMMS. See Appendix~\ref{sec:mms} for definitions and  
more details.

It is worth noting that while the result of \citet*{AMN20} guarantees for every instance with additive valuations the existence of a complete allocation that is $(\varphi-1)$-EFX, EF1, $(\frac{2}{2+\varphi}\approx 0.55)$-GMMS, $(2/3)$-PMMS with no efficiency guarantees, our result gives a complete allocation that is $(\varphi-1)$-EFX, EF1, $(\frac{\varphi-1}{\varphi^2-2\varphi} \approx 0.44)$-GMMS, $(\varphi-1)$-PMMS, and $(\varphi-1)$-MNW.

}

\subsection{Our Techniques}

All of {the proofs of} our existence results are constructive; namely, we {give} algorithms that produce allocations with the desired properties.
In this section, we explain our approach for producing partial and complete allocations. 

\vspace{0.05in}
{\bf Partial allocations.}
For additive valuations, we use Algorithm~\ref{alg:alloc}, which generalizes the algorithm used in \citet*{CGH19} to produce a partial allocation that is (fully) EFX and $1/2$-MNW. 
For the special case of $\alpha = 1$, our algorithm is essentially the same as the original one. 
A very similar modification of the original algorithm was also used by \citet*{GHLVV22} {for the special case of $\alpha = 1/2$}{; however, their modification differs from Algorithm~\ref{alg:alloc} in that it does not guarantee EF1 for additive valuation.}
Both of the algorithms used in \citet*{CGH19} and in \citet*{GHLVV22} were described using a certain notion of EFX feasibility graphs. Here, we present Algorithm~\ref{alg:alloc} in a different way without referring to that notion, which also gives a more direct description of the previous algorithms.

In Algorithm~\ref{alg:alloc}, we start with a maximum Nash welfare allocation and we iteratively drop items that cause envy, possibly reordering the bundles between agents at the same time. We continue to do so until we reach an $\alpha$-EFX {and EF1} allocation.
For example, in the instance described in Example~\ref{ex:simple-ex}, with allocation $X_1=\{a\}, X_2=\{b,c\}$, we first remove item $b$ from $X_2$; 
this eliminates the envy of agent $1$ for agent $2$.

The crucial part of the analysis is to show that the items are removed conservatively so that {for} each agent her final bundle {is worth} at least a $1/(\alpha+1)$ fraction of the bundle she started with, which yields a $1/(\alpha+1)$ approximation to the maximum Nash welfare. To prove this,
{we assume that this condition is violated at some point, and we use this assumption to} construct an allocation with higher Nash welfare than the initial one, which contradicts the assumption that the algorithm {is} given the Nash welfare maximizing allocation as input.

{Consider next subadditive valuations. 
{For $\alpha = 1/2$, the algorithm of}
\citet*{GHLVV22} yields a suboptimal $1/2$ fraction of the maximum Nash welfare.}
We show that {by modifying Algorithm~\ref{alg:alloc}} it is possible to achieve an optimal $2/3$ fraction of the maximum Nash welfare.
Specifically, we modify {the algorithm} to reallocate some of the items that were removed {during the procedure} back to other agents {who value them}.
This is on a very high level what Algorithm~\ref{alg:new_subadditive_alg} is designed to do.
{The technical details are given in the corresponding section.}

\vspace{0.05in}
{\bf Complete allocations.}
The main tool that we use to turn partial allocations into complete ones with the same fairness and Nash welfare guarantees is the envy-cycles procedure \cite{LMMS04}. In this procedure (Algorithm~\ref{alg:envy_cycles}), as long as the allocation is not complete, we take one of the unallocated items and give it to an agent that no other agent envies, or if there is no such agent, then we find a cycle of agents with the property that each agent prefers the bundle of the following agent, and then we improve the allocation by moving the bundles along the cycle. 
It can be shown that if the value of any agent for any of the unallocated items is bounded, then this procedure preserves the initial EFX guarantees, see Lemma~\ref{lem:make_complete}.
The same method was used by \citet*{GHLVV22} to show that there exists a complete $1/2$-EFX and $1/2$-MNW allocation for subadditive valuations, by \citet*{AMN20} to show that there exists a complete $(\varphi -1)$-EFX allocation (with no guarantees on Nash welfare) for additive valuations{, and by~\citet*{FHLSY21} to show that there exists a complete $0.73$-EFR (envy-free up to a random good) allocation for additive valuations.}

Here, the crucial part of the analysis is to provide the appropriate bounds on the value of the unallocated items. In the additive case, we use a similar argument as for showing that the allocation returned by Algorithm~\ref{alg:alloc} is $\frac{1}{\alpha+1}$-MNW. That is, we show that if an item of high value for some agent was removed, then it is possible to construct an allocation with higher Nash welfare than the initial allocation, which again leads to a contradiction. 

In the subadditive case, we use the same argument as \citet*{GHLVV22}. Specifically, we consider a procedure (Algorithm~\ref{alg:sing_swaps}) where if an agent finds one of the unallocated items more valuable than the bundle that was allocated to her, then we swap her bundle with the chosen unallocated item. Eventually, no agent prefers any of the unallocated items to her own bundle, and we can apply the envy-cycles procedure to the resulting partial allocation.

\subsection{Related Literature}

In this section, we review some of the most related literature to our paper. 

\vspace{0.05in}
{\bf Notions of fairness.} The focus of our paper is on the EFX fairness notion, but many other fairness notions have been proposed in the context of resource allocation problems. 
Up until recently, most research in fair division focused on allocating divisible goods, where, unlike in our setting, a single item can be split between agents. \citet*{S1948} introduced the notion of proportionality, then \citet*{F66} and \citet*{V74} initiated the study of envy-freeness by analyzing competitive equilibria from equal incomes (CEEI). These results have been further generalized by \citet*{W85} to cake-cutting settings with one divisible, heterogeneous good.

In the context of indivisible goods, \citet*{B11} introduced the notion of maximin share (MMS) fairness. For additive valuations, \citet*{KPW18} showed that a $2/3$-MMS allocation always exists and can be constructed in polynomial time. Further improvements have been provided by \citet*{GHSSY18} who proved the existence of $3/4$-MMS allocations and by \citet*{GT20} who designed a poly-time algorithm that finds such allocations. {The notion of MMS was further generalized to pairwise maximin share (PMMS)  \cite{CKMPSW16} and groupwise maximin share (GMMS)  \cite{BBMN18}.}

The first relaxation of envy-freeness for indivisible goods, namely EF1, was formally introduced by \citet*{B11}, but an algorithm that produces allocations satisfying EF1 for general monotone valuations was developed even earlier by \citet*{LMMS04}. 

\vspace{0.05in}
{\bf Additional literature on EFX.} \citet*{CKMPSW16} first defined the notion of EFX without determining whether such allocations exist. \citet*{PR18} showed the first existence results for EFX for the settings where there are two agents and where the agents have identical valuations. This was later extended by \citet*{CGM20} to the setting with three agents with additive valuations. 

For approximately EFX allocations, \citet*{PR18} showed that there always exists a $1/2$-EFX allocation even for subadditive valuations, and \citet*{AMN20} improved this to the existence of $(\varphi-1)$-EFX allocations{, albeit only} for additive valuations.

Another line of work
{studies the construction of}
EFX allocations that are almost complete. \citet*{CKMS20} showed that there always exists a partial EFX allocation 
{with at most $n-1$ unallocated elements.} \citet*{CGMMM21} improved this result to allocations with a sublinear number of unallocated items. In settings with $4$ agents, which remains the smallest unresolved case, \citet*{BCFF22} constructed EFX allocations with a single unallocated item.

\vspace{0.05in}
{\bf Additional literature on Nash welfare.} 
The maximum Nash welfare allocation was proposed by \citet*{N50} as a solution to the bargaining problem. Then, Nash welfare was first treated as an efficiency measure by \citet*{KN79}. 

Let us {first} review the results on computing allocations with high Nash welfare in {the} setting with additive valuations. 
When goods are divisible, a MNW allocation can be computed in polynomial time using convex programming, as shown by \citet*{EG59}. However, when goods are indivisible, this problem is NP-hard, even with additive valuations \cite{RE09}. Furthermore, \citet*{GHM18} extended this result by showing that it is NP-hard to find an allocation that is $0.94$-MNW. 
On the positive side, \citet*{BKV18} devised an algorithm that {finds} a $0.69$-MNW {allocation}. 

Beyond the setting of additive valuations, \citet*{GHLVV22} designed an algorithm that finds a $0.25$-MNW allocation for submodular valuations (a strict subclass of subadditive). {Furthermore, \citet*{DLRV23} recently introduced an algorithm returning a constant-factor-MNW allocation for subadditive valuations.}

\vspace{0.05in}
{\bf Trade-offs in fair division.} The trade-offs between social welfare and various fairness constraints are captured by the {\em price of fairness} which was introduced independently by \citet*{BFT11} and \citet*{CKKK09}, and {further investigated by
\cite{BLMS19, BarmanB020, BuLLST22, LiLLTT24}}.
The trade-offs between EF1, EFX, MMS, and PMMS were first studied by \citet*{ABM18}. Moreover, \citet*{AMN20} designed an algorithm that produces an allocation that is simultaneously $(\varphi-1)$-EFX, EF1,  $0.55$-GMMS, and $2/3$-PMMS when the valuations are additive.

\section{Model and Preliminaries}

\subsection{Our Model}
We consider settings with a set  $\agents = \{1, \ldots, n \}$ of $n$ agents, and a set $[m] = \{1, \ldots, m\}$ of $m$ items. 
Every agent $i \in \agents$ has a valuation function denoted by $v_i : 2^{[m]} \to \mathbb{R}^{\geq 0}$, which assigns a real value $v_i (S)$ to every set of items $S\subseteq [m]$.
We consider the following valuation classes:
\begin{itemize}
    \item additive: $v_i(S \cup T) = v_i(S) + v_i(T)$ for any disjoint $S, T \subseteq [m]$.
    \item subadditive: $v_i(S \cup T) \leq v_i(S) + v_i(T)$ for any (not necessarily disjoint) $S, T \subseteq [m]$.
\end{itemize}
An instance of a resource allocation problem is given by a collection of valuation functions $v_1, \ldots, v_n$ over the set of $m$ items.  
Throughout this paper, we use the standard notation $v_i(g) = v_i(\{g\})$ for $g \in [m]$ and $Z-g = Z \setminus \{g \}$ and $Z + g = Z \cup \{g\}$ for $\{g\}, Z \subseteq [m]$.

An allocation $X = (X_1, \ldots, X_n)$ is a collection of disjoint subsets of items, i.e., $X_i \cap X_j = \emptyset$ for $i \neq j$ and $X_i \subseteq [m]$ for every $i \in \agents$. We say that $X$ is complete if $\bigcup_{i \in \agents} X_i = [m]$, and that it is partial otherwise.

The Nash welfare of an allocation $X$ is denoted by $\nw(X) = \prod_{i \in \agents} v_i(X_i)^{1/n}$. For every instance, a maximum Nash welfare (MNW) allocation $X$ is any allocation that maximizes $\nw(X)$ among all possible allocations. We say that an allocation $Z$ is $\beta$-MNW for some $\beta \in [0,1]$ if {it holds that} $\nw(Z) \geq \beta \cdot \nw(X)$.

An allocation $X$ is $\alpha$-EFX if for every $i, j \in \agents$ and $g \in X_j$, it holds that $v_i(X_i) \geq \alpha \cdot v_i(X_j -g )$. Whenever this condition is violated, i.e., $v_i(X_i) < \alpha \cdot v_i(X_j-g) $ {for some $g \in Z_j$, then} we say that the agent $i$ envies agent $j$ in the $\alpha$-EFX sense.
We say that an allocation is EFX if it is $1$-EFX.

\subsection{From Partial Allocations to Complete Allocations.}\label{sec:par_to_com}

{In this section we introduce the notion of $\gamma$-separated allocations, and show how to transform a partial allocation satisfying $\gamma$-separation and $\alpha$-EFX into a complete allocation that gives good EFX guarantees without any loss in Nash welfare.   
This result is key in turning our partial allocations to complete allocations, for both additive and subadditive valuations.}

{We first give the definition of a $\gamma$-separated allocation.}

\begin{definition}[$\gamma$-separation]
    Let $Z=(Z_1, \ldots, Z_n)$ be a partial allocation, and let $U$ be the set of unallocated items in $Z$. We say that $Z$ satisfies \emph{$\gamma$-separation}, for some $\gamma \in [0,1]$, if for every agent $i${, it holds that} $\gamma \cdot v_i(Z_i) \geq v_i(x)$ for all $x \in U$, { i.e.}, agent $i$ prefers $Z_i$ significantly more (by a factor of $1/\gamma$) than any single unallocated item. 
\end{definition}

{The following key lemma shows that any partial allocation that is $\alpha$-EFX and $\gamma$-separated can be turned into a complete $\min(\alpha, 1/(1+\gamma))$-EFX allocation with weakly higher Nash welfare. The proof is based on the envy-cycles procedure introduced by \cite{LMMS04}, see Appendix~\ref{sec:sep_proofs}.}

\begin{restatable}
{lemma}{lemmakecomplete}
    Let $Z=(Z_1, \ldots, Z_n)$ be a partial allocation that is $\alpha$-EFX and satisfies $\gamma$-separation. {Then, there exists}
    a complete allocation $Y$ that is $\min(\alpha, 1/(1+\gamma))$-EFX such that $\nw(Y) \geq \nw(Z)$.
    {Moreover, if $Z$ is EF1, then $Y$ is also EF1.}
    \label{lem:make_complete}
\end{restatable}

{
With this lemma in hand, in order to establish the existence of complete allocations with good EFX and MNW guarantees, it suffices to produce partial allocations with good Nash welfare and separation guarantees. 
}

\section{Additive Valuations}\label{sec:additive}

In this section, we study instances with additive valuations. In Section~\ref{sec:additive-partial}, we prove Theorem~\ref{thm:additive-partial} by constructing a partial allocation with the desired EFX and MNW guarantees. In Section~\ref{sec:additive-complete}, we prove Theorem~\ref{thm:complete_additive} by using 
the techniques from Section~\ref{sec:par_to_com}
to {turn the partial allocation into a complete allocation with the desired guarantees.}

\subsection{Partial Allocations}
\label{sec:additive-partial}

{Our main result in this section is the following:}

\thmadditivepartial*
\begin{algorithm}
\caption{{Additive valuations.}}
\label{alg:alloc}
\begin{flushleft}
\hspace*{\algorithmicindent} \textbf{Input} $(X_1, \ldots, X_n)$ is a complete MNW allocation. \\
\hspace*{\algorithmicindent} \textbf{Output} $(M_1, \ldots, M_n)$ is a partial $\alpha$-EFX and $\frac{1}{\alpha+1}$-MNW allocation.
\end{flushleft}
\vspace{-0.15in}
\begin{algorithmic}
[1]\State match $i$ to $J$ means $M_i \gets J$
\State unmatch $Z_i$ means $M_u \gets \bot$ if there is $u$ matched to $Z_i$
\Procedure{alg}{} 
\State $Z\gets(X_1, \ldots, X_n)$
\State $M\gets (\bot, \ldots, \bot)$
\While{there is an agent $i$ with $M_i = \bot$}
\State $i^\star \gets $ any agent with $M_i = \bot$
\If{{($Z_{i^\star}=X_{i^\star}$ and }$v_{i^\star}(Z_{i^\star}) \geq \alpha \cdot v_{i^\star}(Z_j -g)$ for all $j$ and $g \in Z_j$) {or \label{line:first_if}\\ \;\;\;\;\;\;\;\;\;\;\;\;\;\, ($Z_{i^\star} \neq X_{i^\star}$ and $v_{i^\star}(Z_{i^\star}) \geq v_{i^\star}(Z_j -g)$ for all $j$ and $g \in Z_j$)}}\label{line:if_add}
    \State unmatch $Z_{i^\star}$
    \State match $i^\star$ to $Z_{i^\star}$
\Else  
    \State $j^\star,g \gets $ any $j^\star$ and $g \in Z_{j^\star}$ that maximize $v_{i^\star}(Z_{j^\star}-g)$\label{line:first_else}
    \State unmatch $Z_{j^\star}$
    \State change $Z_{j^\star}$ to $Z_{j^\star} - g$\label{line:change_z_add}
    \State match $i^\star$ to $Z_{j^\star}$\label{line:last_else}
\EndIf
\EndWhile\label{while}
\State \textbf{return} $(M_1, \ldots, M_n)$
\EndProcedure
\end{algorithmic}
\end{algorithm}

{Algorithm~\ref{alg:alloc} receives as input a Nash welfare maximizing allocation $X=(X_1, \ldots, X_n)$, and produces an allocation $M=(M_1, \ldots, M_n)$ that is both $\alpha$-EFX and $\frac{1}{\alpha+1}$-MNW. We refer to $M$ as a matching between agents and bundles {of items}.
We say that an agent $i$ is \emph{matched} to $M_i$ if $M_i \neq \bot$, and that $i$ is \emph{unmatched} otherwise.
The algorithm maintains a set of bundles $(Z_1, \ldots, Z_n)$ which are initially set to $Z_i = X_i$. 
Throughout the algorithm, every matched agent $i$ is matched to some bundle $Z_j$. 
We say that $Z_j$ is \emph{matched} to $i$ if $M_i = Z_j$.
The matching $M$ never assigns the same bundle to two agents, and at the end of the algorithm, every bundle $Z_j$ is matched to some agent $i$. This is formally stated in the following claim whose proof is deferred to Appendix~\ref{sec:proofs-additive}.
}

\begin{samepage}
\begin{restatable}{claim}{claaddbasic}
\label{cla:add_basic}
    At the end of each iteration of the algorithm, 
    \begin{enumerate}[(i)]
        \item for any agent $i$, it holds that $Z_i \subseteq X_i$,
        \item for any agent $i$, it holds that $M_i \in \{\bot\} \cup \{Z_j \bar j \in \agents\}$, and \label{list:mi_structure}
        \item for any distinct agents $i_1$ and $i_2$ with $M_{i_1}, M_{i_2} \neq \bot$, it holds that $M_{i_1} \neq M_{i_2}$.
    \end{enumerate}
\end{restatable}
\end{samepage}

{
The high-level idea of the algorithm is as follows: 
The algorithm starts off by setting $Z_i=X_i$ for every agent $i${;} recall that $X$ is {an} MNW allocation. {At first,} no agent is matched. 
{Then, the algorithm} proceeds by shrinking the bundles $Z_j$'s and {matching} them to agents in a way that eliminates envy. 
More precisely, as long as there is an unmatched agent $i^\star$, one of the following two operations takes place{.}
(i) If $Z_{i^\star}$ is ``good enough'' for $i^\star$, then $i^\star$ is matched to $Z_{i^\star}$. 
{The condition for $Z_{i^\star}$ to be good enough for $i^\star$ in the case where $Z_{i^\star} = X_{i^\star}$ is that if $i^\star$ gets $Z_{i^\star}$, then she does not envy any other bundle $Z_j$ for any agent $j$, in the $\alpha$-EFX sense. 
If $Z_{i^\star} \subsetneq X_{i^\star}$, then the condition is that if $i^\star$ gets $Z_{i^\star}$, then she does not envy any other bundle $Z_j$ for any agent $j$, in the EFX sense, which is a stronger requirement than in the $Z_{i^\star} = X_{i^\star}$ case.}
(ii) Otherwise, {if $Z_{i^\star}$ is not good enough for $i^\star$,} the algorithm picks the most valuable (from the perspective of $i^\star$) {\em strict} subset of $Z_{j^\star}$ for some $j^\star$, shrinks $Z_{j^\star}$ to the chosen strict subset, and matches $i^\star$ to the new $Z_{j^\star}$ (leaving the agent previously matched to $Z_{j^\star}$, if any, unmatched). {It can be shown} that in this case, $i^\star$ does not envy any other bundle $Z_j$ in the stronger sense of EFX (rather than $\alpha$-EFX).
}

This procedure terminates because every time the second operation takes place, $Z_{j^\star}$ shrinks. This is shown in the following lemma.

\begin{restatable}{lemma}{runtimelemma}\label{lem:addruntimelemma}
Algorithm~\ref{alg:alloc} terminates in a polynomial (in $n$ and $m$) number of steps.
\end{restatable}

\begin{proof} 
    For any $t$, consider the state of the algorithm after exactly $t$ iterations of the while loop{, and} let $\mathsf{Removed}_t = \bigcup_{i\in \agents} X_i \setminus Z_i$ and $\mathsf{SelfMatched}_t = \{ i \in \agents \bar M_i = Z_i\}$. Consider the tuple ${\phi}_t = (|\mathsf{Removed}_t|, |\mathsf{SelfMatched}_t|)$. Note that by the design of the algorithm, the tuple ${\phi}_t$ strictly increases lexicographically with each iteration of the algorithm. Since there are at most $(m+1)n$ values that ${\phi}_t$ can take, the algorithm terminates after at most $(m+1)n$ iterations.
\end{proof}

{Let us first make a few simple observations that are crucial to the analysis of the algorithm. First, t}he matching of the algorithm ensures
that $M_{i^\star}$ is ``good enough'' for $i^\star$ whenever $i^\star$ is matched. 
This in fact implies that the final allocation $M$ is $\alpha$-EFX {and EF1}. 
{Second, whenever an agent $i$ {with an untouched bundle} is matched by the second operation (lines~\ref{line:first_else}-\ref{line:last_else}), she is matched to a bundle that she prefers significantly more (by a factor of $1/\alpha$) to $Z_i$.
Third, any touched bundle, i.e., one from which the algorithm removed an element, is matched. These observations are formally stated in the following claim, whose} proof is deferred to Appendix~\ref{sec:proofs-additive}.

\begin{restatable}{claim}{claaddineq}
\label{cla:add_ineq}
    At the end of each iteration, for any matched agent $i$, it holds that
    \begin{enumerate}[(i)]
        \item {if $Z_i = X_i$, then} $v_i(M_i) \geq \alpha \cdot v_i(Z_j-g)$ for all $j$ and all $g \in Z_j$, \label{list:efx} 
        \item {if $Z_i \neq X_i$, then $v_i(M_i) \geq v_i(Z_j-g)$ for all $j$ and all $g \in Z_j$,} \label{list:ef1} 
        \item {if $M_i \neq Z_i$, then $v_i(M_i) > v_i(Z_i)$,} \label{list:1_inequality}
        \item if $M_i \neq Z_i$ {and $Z_i = X_i$}, then $v_i(M_i) > (1/\alpha) \cdot v_i(Z_i)$, and \label{list:a_inequality}
         \item for any agent $i$ with $Z_i \subsetneq X_i$, there is some agent $u$ with $M_u = Z_i$.\label{list:unmatched}
    \end{enumerate}
\end{restatable}

We now present the key lemma in our analysis.

\begin{lemma}\label{lem:additive_bound}
At the end of the run of the algorithm, we have $v_i(Z_i) \geq \frac{1}{\alpha+1} \cdot v_i(X_i)$.
\end{lemma}

{
The proof of the lemma actually shows that the invariant condition $v_i(Z_i) \geq  \frac{1}{\alpha+1} \cdot v_i(X_i)$ for all agents $i$ is preserved throughout the entire execution of the algorithm. 
This condition ensures that the final allocation is $\frac{1}{\alpha+1}$-MNW.
The proof of is based on the idea that whenever the condition is violated, we can construct an alternative allocation $\widehat{X}$ with a higher Nash welfare than $X$, contradicting the optimality of $X$.

Before presenting the proof of Lemma~\ref{lem:additive_bound}, we need the following definition.
}

\begin{definition}[Improving sequence] \label{def:improving_seq}
At the end of each iteration where the algorithm executes lines~\ref{line:first_else}-\ref{line:last_else}, we construct an improving sequence $j_1, \ldots, j_\ell$ in the following way. Set $j_1 = j^\star$. Next, for $s \geq 1$, until either (i) $Z_{j_{s}}$ is matched to $j^\star$, or (ii) $Z_{j_{s}}$ is unmatched, inductively define $j_{s+1}$ to be the unique agent so that $M_{j_{s+1}} = Z_{j_{s}}$. Let $\ell = s$ for $s$ where the process ends.
\end{definition}

Note that the improving sequence has  $j_2 = i^\star$ by the design of the algorithm.

{
With this definition in hand, we are ready to present the proof of Lemma~\ref{lem:additive_bound}.
}

\begin{proof}[Proof of Lemma~\ref{lem:additive_bound}]
Suppose for the purpose of contradiction that the algorithm gets to a point where $v_{j^\star}(Z_{j^\star}) < \frac{1}{\alpha+1} \cdot v_{j^\star}(X_{j^\star})$, and consider the first time it happens. Note that it must be the case that the algorithm executed lines \ref{line:first_else}-\ref{line:last_else} during that iteration {because otherwise no bundle is modified}. Let $j_1, \ldots, j_\ell$ be the improving sequence constructed at the end of that iteration (Definition~\ref{def:improving_seq}). See Figure~\ref{fig:additive_bound} for {the} pictures accompanying this proof.

\underline{Case 1: $j$ ends with condition (i).} Consider first the case where $Z_{j_{\ell}}$ is matched to $j^\star$. Let us define an allocation
\begin{align*}    
\widehat{X}_{j_s} = (X_{j_s} \setminus Z_{j_s}) \cup Z_{j_{s-1}} && \textrm{for } 1 \leq s \leq \ell
\end{align*}
where $j_{0}=j_\ell$, and $\widehat{X}_i = X_i$ for $i \notin \{j_1, \ldots, j_\ell\}$. Observe that for all $s$, since $M_{j_s} = Z_{j_{s-1}}$, it holds that $v_{j_s}(Z_{j_{s-1}}) > v_{j_s}(Z_{j_s})$ by {Claim~\ref{cla:add_ineq}(\ref{list:1_inequality})}, and hence
\EQ{
v_{j_s}(\widehat{X}_{j_s}) = v_{j_s}(X_{j_s}) - v_{j_s}(Z_{j_s}) + v_{j_s}(Z_{j_{s-1}}) > v_{j_s}(X_{j_s})
}
where the first equality holds by additivity. It follows that $\widehat{X}$ is a Pareto improvement of $X$, contradicting the NW-maximiality of $X$. {In particular, note that this implies that at the end of the execution of the algorithm, where all the agents are matched, it must hold that $M_i = Z_i$ for any agent $i$.}

\underline{Case 2: $j$ ends with condition (ii).} 
 Consider next the case where $Z_{j_\ell}$ is unmatched. By Claim~\ref{cla:add_ineq}(\ref{list:unmatched}), it holds that $Z_{j_\ell} = X_{j_\ell}$. Let $\widehat{X}$ be the allocation constructed in the following way
 \begin{align*}
   &\widehat{X}_{j^\star} = X_{j^\star} \setminus Z_{j^\star} \\*
&\widehat{X}_{j_s} = (X_{j_s}\setminus Z_{j_s}) \cup Z_{j_{s-1}}  && \textrm{for } 2 \leq s < \ell \\*
&\widehat{X}_{j_\ell} = X_{j_\ell} \cup Z_{j_{\ell-1}}  
 \end{align*}
and $\widehat{X}_i = X_i$ for $i \notin \{j_1, \ldots, j_\ell\}$. We have the following inequalities (the second equality in each line is by additivity). The contradiction assumption that $v_{j^\star}(Z_{j^\star}) < \frac{1}{\alpha+1} \cdot v_{j^\star}(X_{j^\star})$ implies
\EQ{
v_{j^\star}(\widehat{X}_{j^\star}) =
v_{j^\star}(X_{j^\star} \setminus Z_{j^\star}) &= v_{j^\star}(X_{j^\star}) - v_{j^\star}(Z_{j^\star}) > \frac{\alpha}{\alpha+1}  \cdot v_{j^\star}(X_{j^\star}).
}
By {Claim~\ref{cla:add_ineq}(\ref{list:1_inequality})}, it holds that
\EQ{
v_{j_s}(\widehat{X}_{j_s}) = v_{j_s}((X_{j_s}\setminus Z_{j_s}) \cup Z_{j_{s-1}}) = v_{j_s}(X_{j_s}) + v_{j_s}(Z_{j_{s-1}})- v_{j_s}(Z_{j_s}) > v_{j_s}(X_{j_s}).} 
It follows from {Claim~\ref{cla:add_ineq}(\ref{list:a_inequality})} and the fact that $X_{j_\ell}=Z_{j_\ell}$ that
\EQ{
v_{j_\ell}(\widehat{X}_{j_\ell}) = v_{j_\ell}(X_{j_\ell} \cup Z_{j_{\ell-1}}) = v_{j_\ell}(X_{j_\ell}) + v_{j_\ell}(Z_{j_{\ell-1}}) > v_{j_\ell}(X_{j_\ell}) + \frac{1}{\alpha} \cdot v_{j_\ell}(Z_{j_\ell}) = \frac{\alpha+1}{\alpha} \cdot v_{j_\ell}(X_{j_\ell}).}
Finally, combining the inequalities above gives
\EQ{ 
v_{j^\star}(\widehat{X}_{j^\star})& \cdot v_{i^\star}(\widehat{X}_{i^\star}) \cdot v_{j_{3}}(\widehat{X}_{j_{3}}) \cdots v_{j_{\ell-1}}(\widehat{X}_{j_{\ell-1}}) \cdot v_{j_\ell}(\widehat{X}_{j_\ell}) \\
&> \, \frac{\alpha}{\alpha+1}v_{j^\star}(X_{j^\star}) \cdot v_{i^\star}(X_{i^\star}) \cdot v_{j_{3}}(X_{j_{3}} ) \cdots v_{j_{\ell-1}}(X_{j_{\ell-1}})  \cdot \frac{\alpha+1}{\alpha}v_{j_\ell}(X_{j_\ell}).
}
This implies that $\nw({\widehat{X}}) > \nw({X})$, contradicting the NW-maximality of $X$.
\end{proof}

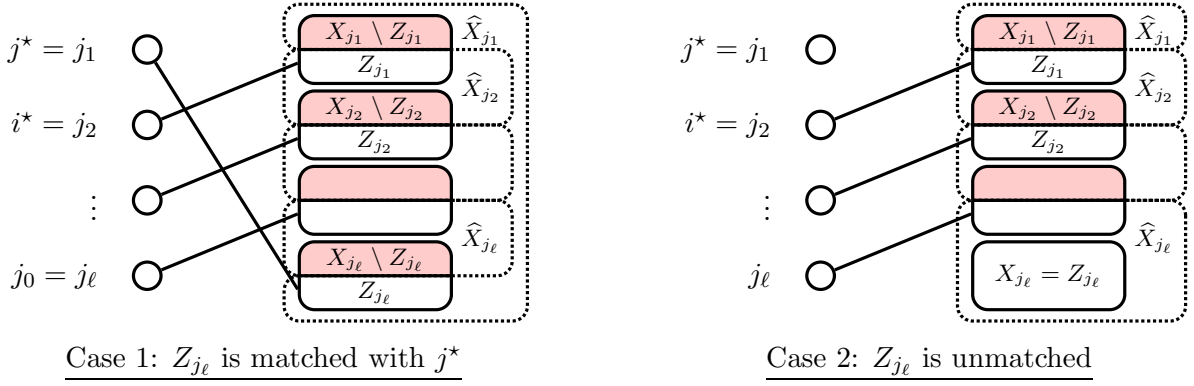
\begin{figure}[t]
\begin{center}
\def\yscale{0.5}
\begin{minipage}{.45\textwidth}
\centering
\begin{tikzpicture}[very thick,yscale=\yscale]
\foreach \i in {2,...,5}
{
    \begin{scope}
        \clip (2,10-2*\i) rectangle (4,10-2*\i+1);
        \node[fill, red!20!white, rounded corners=.2cm, minimum width=2cm, minimum height=1.8*\yscale cm] at (3,10-2*\i){};
    \end{scope}
    \draw (2,10-2*\i) -- (4,10-2*\i);
}
\foreach \i in {2,...,5}
{
    \node[draw,circle] (j\i) at (0,10-2*\i) {};
    \node[draw, rounded corners=.2cm, minimum width=2cm, minimum height=1.8*\yscale cm] (X\i) at (3,10-2*\i){};
}
\node[anchor=east] at (-.5,6) {$j^\star = j_1$};
\node[anchor=east] at (-.5,4) {$i^\star = j_2$};
\node[anchor=east] at (-.5,2) {$\vdots$};
\node[anchor=east] at (-.5,0) {$j_0 = j_{\ell}$};
{\footnotesize
\node at (3,6.45) {$X_{j_1}\setminus Z_{j_1}$};
\node at (3,4.45) {$X_{j_2}\setminus Z_{j_2}$};
\node at (3,0.45) {$X_{j_{\ell}}\setminus Z_{j_{\ell}}$};
\node at (3,5.55) {$Z_{j_1}$};
\node at (3,3.55) {$Z_{j_2}$};
\node at (3,-.45) {$Z_{j_{\ell}}$};
\node[anchor=west] at (4,6.5) {$\widehat{X}_{j_1}$};
\node[anchor=west] at (4,5) {$\widehat{X}_{j_2}$};
\node[anchor=west] at (4,1) {$\widehat{X}_{j_\ell}$};
}
\draw (j5) -- (X4.190);
\draw (j4) -- (X3.190);
\draw (j3) -- (X2.190);
\draw (j2) -- (X5.190);
\begin{scope}[densely dotted, rounded corners=.2cm]
    \draw (2,0) -- (1.8,0) -- (1.8,-1.2) -- (5,-1.2) -- (5,7.2) -- (1.8,7.2) -- (1.8,6) -- (2,6);
    \draw (4.8,6) -- (4.8,4) -- (1.8,4) -- (1.8,6) -- cycle;
    \draw (4.8,4) -- (4.8,2) -- (1.8,2) -- (1.8,4) -- cycle;
    \draw (4.8,2) -- (4.8,0) -- (1.8,0) -- (1.8,2) -- cycle;
\end{scope}
\end{tikzpicture}
\smallbreak\underline{Case 1: $Z_{j_\ell}$ is matched with $j^\star$}
\end{minipage}
\hfill
\begin{minipage}{.45\textwidth}
\centering
\begin{tikzpicture}[very thick,yscale=\yscale]
\foreach \i in {2,...,4}
{
    \begin{scope}
        \clip (2,10-2*\i) rectangle (4,10-2*\i+1);
        \node[fill, red!20!white, rounded corners=.2cm, minimum width=2cm, minimum height=1.8*\yscale cm] at (3,10-2*\i){};
    \end{scope}
    \draw (2,10-2*\i) -- (4,10-2*\i);
}
\foreach \i in {2,...,5}
{
    \node[draw,circle] (j\i) at (0,10-2*\i) {};
    \node[draw, rounded corners=.2cm, minimum width=2cm, minimum height=1.8*\yscale cm] (X\i) at (3,10-2*\i){};
}
\node[anchor=east] at (-.5,6) {$j^\star = j_1$};
\node[anchor=east] at (-.5,4) {$i^\star = j_2$};
\node[anchor=east] at (-.5,2) {$\vdots$};
\node[anchor=east] at (-.5,0) {$j_{\ell}$};
{\footnotesize
\node at (3,6.45) {$X_{j_1}\setminus Z_{j_1}$};
\node at (3,4.45) {$X_{j_2}\setminus Z_{j_2}$};
\node at (3,5.55) {$Z_{j_1}$};
\node at (3,3.55) {$Z_{j_2}$};
\node at (3,0) {$X_{j_\ell} = Z_{j_\ell}$};
\node[anchor=west] at (4,6.5) {$\widehat{X}_{j_1}$};
\node[anchor=west] at (4,5) {$\widehat{X}_{j_2}$};
\node[anchor=west] at (4,1) {$\widehat{X}_{j_\ell}$};
}
\draw (j5) -- (X4.190);
\draw (j4) -- (X3.190);
\draw (j3) -- (X2.190);
\begin{scope}[densely dotted, rounded corners=.2cm]
    \draw (4.8,7.2) -- (4.8,6) -- (1.8,6) -- (1.8,7.2) -- cycle;
    \draw (4.8,6) -- (4.8,4) -- (1.8,4) -- (1.8,6) -- cycle;
    \draw (4.8,4) -- (4.8,2) -- (1.8,2) -- (1.8,4) -- cycle;
    \draw (4.8,2) -- (4.8,-1.2) -- (1.8,-1.2) -- (1.8,2) -- cycle;
\end{scope}
\end{tikzpicture}
\smallbreak\underline{Case 2: $Z_{j_\ell}$ is unmatched}
\end{minipage}
\end{center}
\vspace{-.2cm}
\caption{The structure of the new allocation $\widehat{X}$ in Case 1 and Case 2. The edges represent the matching $M$.}
\label{fig:additive_bound}
\end{figure}

{
We now have all the ingredients needed to prove Theorem~\ref{thm:additive-partial}.
}

\begin{proof}[Proof of Theorem \ref{thm:additive-partial}]
The allocation $(M_1, \ldots, M_n)$ is $\alpha$-EFX by Claim~\ref{cla:add_ineq}(\ref{list:efx}, \ref{list:ef1}) and Claim~\ref{cla:add_basic}(\ref{list:mi_structure}). {Moreover, $M$ is EF1 since either $Z_i \subsetneq X_i$ and the property follows from Claim~\ref{cla:add_ineq}(\ref{list:ef1}), or $Z_i = X_i$ and then, since $X$ is EF1 by the result of \citet*{CKMPSW16}, for all agents $j$, it holds that $v_i(M_i) \geq v_i(Z_i) = v_i(X_i) \geq v_i(X_j - g) \geq v_i(Z_j - g)$ for some $g \in X_j$.}

Fix any agent $i$. Let $j$ be the unique agent matched to $Z_i$. If $j \neq i$, then by {Claim~\ref{cla:add_ineq}(\ref{list:1_inequality})}, it holds that $v_i(Z_j) \geq v_i(Z_i)$. If $j=i$, then clearly $v_i(Z_j) \geq v_i(Z_i)$. 
Therefore, by Lemma~\ref{lem:additive_bound},
    \EQ{
    \prod_{i \in \agents} v_i(M_i)^{1/n} &\geq \prod_{i \in \agents} v_i(Z_i)^{1/n} 
    \geq \prod_{i \in \agents} \left( \frac{1}{\alpha+1} \cdot v_i(X_i)\right)^{1/n} 
    =  \frac{1}{\alpha+1} \cdot \prod_{i \in \agents} v_i(X_i)^{1/n}
    }
    and so the result follows.
\end{proof}

\subsection{Complete Allocations}\label{sec:additive-complete}

{
In this section, we provide the following result.
}
\thmadditivecomplete*

{Let $(Z_1, \ldots, Z_n)$ be the bundles at the end of the run of Algorithm~\ref{alg:alloc}. }

{The following lemma is the key component in the proof of  Theorem~\ref{thm:complete_additive}. 
It offers additional analysis of Algorithm~\ref{alg:alloc}, showing that the allocation returned by this algorithm is $\alpha$-separated.
}

\begin{restatable}{lemma}{lemadditivebettersep}
\label{lem:additive_better_sep} The partial allocation
$(Z_1, \ldots, Z_n)$ satisfies $\alpha$-separation.
\end{restatable}

The remainder of this section is dedicated to proving this lemma. We prove the $\alpha$-separation property by constructing an allocation $\widehat{X}$ {which is built from $X$} and using the optimality of $X$ to infer that the Nash welfare of $\widehat{X}$ is at most the Nash welfare of $X$.
The construction of $\widehat{X}$ is based on a non-trivial redistribution of the items among agents, which requires additional definitions, as follows.

\begin{definition}[Touching]\label{def:touching}
For any agent $i$ with $Z_i \subsetneq X_i$, we say that agent $k$ was the \emph{last one to touch} $i$ if in the last iteration in Algorithm \ref{alg:alloc} with $j^\star = i$, it was the case that $i^\star = k$. If, on the other hand, $Z_i = X_i$, then we say that agent $i$ is \emph{untouched}.    
\end{definition}

The following inequalities follow {from Claim~\ref{cla:add_ineq}}.

\begin{restatable}{claim}{clalasttouchedsimple}
\label{cla:last_touched_simple} 
Let $i$ be any agent with $Z_i \subsetneq X_i$ and let $k$ be the last agent to touch $i$. Then, 
\begin{enumerate}[(i)]
    \item {$v_k(Z_i) > v_k(Z_k)$}, and\label{list:touching_1_inequality}
    \item {if $Z_k = X_k$, then $v_k(Z_i) > (1/\alpha) \cdot v_k(Z_k)$.} \label{list:touching_a_inequality}
\end{enumerate}
\end{restatable}

\begin{proof}
      Denote by $(Z_1', \ldots, Z_n')$ the {bundles} at the end of the last iteration where $j^\star = i$. Note that $Z_r \subseteq Z_r' \subseteq X_r$ for any agent $r$ by the design of the algorithm. By the assumption that $k$ is the last agent that touched $i$, during {the chosen} iteration, it holds that $i^\star = k$. Moreover, by the assumption that this is the last time the algorithm removes an element from $Z_i$, it holds that $Z_i = Z_i'$. It follows that 
      {
\begin{align*}
    v_k(Z_k) &\leq v_k(Z_k') && (\text{since } Z_k \subseteq Z_k')\\
    &<   v_k(Z_i') && (\text{by Claim \ref{cla:add_ineq}(\ref{list:1_inequality})}) \\
    &= v_k(Z_i) && (\text{since } Z_i = Z_i')
    \intertext{which gives property (i). {Furthermore,} if we assume $Z_k=X_k$, then it holds that}
    v_k(Z_k) &\leq v_k(Z_k') && (\text{since } Z_k \subseteq Z_k')\\
    &<  \alpha \cdot v_k(Z_i') && (\text{by Claim \ref{cla:add_ineq}(\ref{list:a_inequality})}) \\
    &= \alpha \cdot v_k(Z_i) && (\text{since } Z_i = Z_i')
\end{align*}}
which gives property (ii).
\end{proof}

We now define a touching sequence, cf. Definition~\ref{def:improving_seq}.

\begin{definition}[Touching sequence] \label{def:touching_seq}
Given any {agent} $i$, we define a \emph{touching sequence} $k_1, \ldots, k_\ell$ in the following way. Let {$k_1 = i$}. For every {$s \geq 1$}, until either
(i) $k_s$ is untouched, or (ii) the last agent to touch $k_s$ is already in the sequence $k_1, \ldots, k_{s-1}$,
    define $k_{s+1}$ to be the last agent to touch $k_s$. Let $\ell = s$ for $s$ where the process ends.
\end{definition}

The proof of Lemma~\ref{lem:additive_better_sep}  uses the following technical lemma.

\begin{lemma} \label{lem:trick}
    Let $i$ and $j$ be any distinct agents. Let $\widehat{X}_i = (X_i \setminus Z_i) \cup (X_j \setminus Z_j)$.
    Suppose that $v_i(\widehat{X}_i) \leq \alpha \cdot v_i(X_i)$. Then, $v_i(x) \leq \alpha \cdot v_i(Z_i)$ for all $x \in X_j \setminus Z_j$.
\end{lemma}
\begin{proof}
    Observe that for any $x \in X_j \setminus Z_j$, {it holds that}
\begin{align*}
    \frac{v_i(x)}{v_i(Z_i)} &\leq \frac{v_i(X_j \setminus Z_j)}{v_i(Z_i)} && (\text{by monotonicity of }v_i) \\ 
    &= \frac{v_i((X_i \setminus Z_i) \cup (X_j \setminus Z_j)) - v_i(X_i\setminus Z_i)}{v_i(X_i)- v_i(X_i\setminus Z_i)} && (\text{by additivity of } v_i)\\
    &\leq \frac{v_i((X_i \setminus Z_i) \cup (X_j \setminus Z_j))}{v_i(X_i)} && (\star) \\
    &\leq \alpha && (\text{by the assumption})
\end{align*}
where ($\star$) follows because if $0 \leq c<a\leq b$, then $(a-c)/(b-c) \leq a/b$.
\end{proof}

{We are now ready to prove Lemma~\ref{lem:additive_better_sep}. 
}

\begin{proof}[Proof of Lemma~\ref{lem:additive_better_sep}]
Fix any agents $i$ and $j$ {with $Z_j \neq X_j$} and any item $x \in X_j \setminus Z_j$. The goal is to show that $v_i(x) \leq \alpha \cdot v_i(Z_i)$.  See Figure~\ref{fig:additive_better_sep} for the pictures accompanying Cases 4 and 5 of this proof.

\underline{Case 1: $i=j$.} First, assume that $i=j$. Then, it holds that 
\begin{align*}
v_i(x) &\leq v_i(X_i \setminus Z_i) && (\text{by monotonicity of }v_i) \\
&= v_i(X_i) - v_i(Z_i)  && (\text{by additivity of }v_i) \\
&\leq (\alpha+1) \cdot v_i(Z_i) - v_i(Z_i)  && (\text{by Lemma \ref{lem:additive_bound}}) \\
&= \alpha \cdot v_i(Z_i)
\end{align*}
which is the desired bound.

    Next, assume that $i \neq j$. Let $k_1, \ldots, k_\ell$ be the touching sequence starting with {$k_1=i$} (Definition \ref{def:touching_seq}).

\underline{Case 2: {$k$ ends with condition (ii)}.} {Suppose that $k_w$ is the last agent to touch $k_\ell$ for some $1 \leq w < \ell$.} Let $(\widehat{X}_1, \ldots, \widehat{X}_n)$ be the following allocation
\begin{align*}
    &\widehat{X}_{k_w} = (X_{k_w} \setminus Z_{k_w}) \cup Z_{k_{\ell}} \\
    &\widehat{X}_{k_s} = (X_{k_s} \setminus Z_{k_s}) \cup Z_{k_{s-1}}  && \text{\ for $w+1 \leq s \leq \ell$} 
\end{align*}
and $\widehat{X}_r = X_r$ for any $r \notin \{k_w, \ldots, k_\ell\}$. By {Claim \ref{cla:last_touched_simple}(\ref{list:touching_1_inequality})}, $v_{k_s}(Z_{k_{s-1}}) > v_{k_s}(Z_{k_s})$, and so
\begin{align*}
    v_{k_s}(\widehat{X}_{k_s}) = v_{k_s}((X_{k_s} \setminus Z_{k_s}) \cup Z_{k_{s-1}}) = v_{k_s}(X_{k_s}) - v_{k_s}(Z_{k_s}) + v_{k_s}(Z_{k_{s-1}}) > v_{k_s}(X_{k_s})
\end{align*}
and similarly, $v_{k_w}(Z_{k_\ell}) > v_{k_w}(Z_{k_w})$, and so
\begin{align*}
    v_{k_w}(\widehat{X}_{k_w}) = v_{k_w}((X_{k_w} \setminus Z_{k_w}) \cup Z_{k_{\ell}}) = v_{k_w}(X_{k_w}) - v_{k_w}(Z_{k_w}) + v_{k_w}(Z_{k_{\ell}}) > v_{k_w}(X_{k_w})
\end{align*}
and hence $\widehat{X}$ is a Pareto improvement of $X$, contradicting the NW-maximality of $X$.

{For the remainder of the proof, we assume that $k$ ends with condition (i).}

\underline{Case 3: {$\ell = 1$}.} If {$k_1 = i$} is not touched at the end of the run of the algorithm, i.e., $Z_i = X_i$, then
consider the allocation $(\widehat{X}_1, \ldots, \widehat{X}_n)$ defined by 
\begin{align*}
 &\widehat{X}_i = X_i \cup (X_j \setminus Z_j) \\
 &\widehat{X}_j = Z_j
\end{align*}
and $\widehat{X}_r = X_r$ for all $r \notin \{i,j\}$. The following must hold
\begin{align*}
    1 &\geq \frac{v_i(\widehat{X}_i)}{v_i({X}_i) } \cdot \frac{v_j(\widehat{X}_j) }{ v_j({X}_j) } && (\text{by the optimality of } X) \\
    &=\frac{v_i(X_i \cup (X_j\setminus Z_j))}{v_i(X_i)} \cdot \frac{v_j(Z_j)}{v_j(X_j)} &&  (\text{by the construction of } \widehat{X})\\
    &\geq \frac{v_i(X_i \cup (X_j\setminus Z_j))}{v_i(X_i)} \cdot \frac{1}{\alpha+1} && (\text{by Lemma \ref{lem:additive_bound}}) \\
    &= \left(1 + \frac{v_i(X_j\setminus Z_j)}{v_i(X_i)}\right) \cdot \frac{1}{\alpha+1} && (\text{by additivity of } v_i)
\end{align*}
and so 
\begin{align*}
    v_i(x) &\leq v_i(X_j \setminus Z_j) && (\text{by monotonicity of } v_i) \\
    &\leq \alpha\cdot v_i(X_i) && (\text{by the above})\\
    &= \alpha\cdot v_i(Z_i)&& (\text{since }Z_i = X_i)
\end{align*}
which is the desired bound.

\underline{Case 4: {$\ell > 1$} {and $j \neq k_t$ for all $2 \leq t \leq \ell$}.}
    Let $(\widehat{X}_1, \ldots, \widehat{X}_n)$ be the following allocation
    \begin{align*}
        &\widehat{X}_j = Z_j \\
        &\widehat{X}_i = (X_i \setminus Z_i) \cup (X_j \setminus Z_j) \\
        &\widehat{X}_{k_s} = (X_{k_s} \setminus Z_{k_s}) \cup Z_{k_{s-1}} && \text{for } {2} \leq s < \ell \\
        &\widehat{X}_{k_\ell} = X_{k_\ell} \cup Z_{k_{\ell-1}}
    \end{align*}
    where $\widehat{X}_r = X_r$ for any $r \notin \{k_1, \ldots, k_\ell\}$.  
    By Lemma \ref{lem:additive_bound},
    \begin{align*}
        v_j(\widehat{X}_j) = v_{j}(Z_j)
        \geq (1/(\alpha+1)) \cdot v_j(X_j).
    \end{align*}
    By {Claim \ref{cla:last_touched_simple}(\ref{list:touching_1_inequality})}, for ${2} \leq s < \ell$, it holds that
       \begin{align*}
        v_{k_s}(\widehat{X}_{k_s}) = v_{k_s}((X_{k_s} \setminus Z_{k_s}) \cup Z_{k_{s-1}})
        = v_{k_s}(X_{k_s})-v_{k_s}(Z_{k_s}) + v_{k_s}(Z_{k_{s-1}}) 
        \geq v_{k_s}(X_{k_s}).
    \end{align*}
    Finally, the assumption that $X_{k_\ell} = Z_{k_\ell}$ and {Claim~\ref{cla:last_touched_simple}(\ref{list:touching_a_inequality})} give
       \begin{align*}
        v_{k_\ell}(\widehat{X}_{k_\ell}) = v_{k_\ell}(X_{k_\ell} \cup Z_{k_{\ell-1}}) 
        = v_{k_\ell}(X_{k_\ell})+ v_{k_\ell}(Z_{k_{\ell-1}}) 
        \geq (1 + 1/\alpha) \cdot v_{k_\ell}(X_{k_\ell}).
    \end{align*}
    Combining the inequalities above gives, by optimality of $X$, 
    \begin{align*}
    1 &\geq \frac{v_j(\widehat{X}_j)}{v_j(X_j)} \cdot 
    \frac{v_i(\widehat{X}_i)}{v_i(X_i)} \cdot 
    \frac{v_{k_3}(\widehat{X}_{k_3})}{v_{k_3}(X_{k_3})} \cdots
    \frac{v_{k_\ell}(\widehat{X}_{k_\ell})}{v_{k_3}(X_{k_\ell})} \\
    &\geq (1/(\alpha+1)) \cdot\frac{v_i(\widehat{X}_i)}{v_i(X_i)} \cdot 1 \cdots 1 \cdot (1 + 1/\alpha) \\
    &= {\frac{1}{\alpha}} \cdot \frac{v_i(\widehat{X}_i)}{v_i(X_i)} 
\end{align*}
and after rearranging we can apply Lemma~\ref{lem:trick}.

\underline{Case 5: {$\ell > 1$ and $j = k_t$ for some $2 \leq t \leq \ell$}.} {Note that since we assume that $X_j \setminus Z_j$ is non-empty, it cannot be that $t = \ell$.} Let $(\widehat{X}_1, \ldots, \widehat{X}_n)$ be the following allocation
{\begin{align*}
        &\widehat{X}_j = Z_{k_{t-1}}\\
        &\widehat{X}_i = (X_i \setminus Z_i) \cup (X_j \setminus Z_j) \\
        &\widehat{X}_{k_\ell} = X_{k_\ell} \cup Z_{k_{\ell-1}} \\ 
        &\widehat{X}_{k_s} = (X_{k_s} \setminus Z_{k_s}) \cup Z_{k_{s-1}} && \text{for } s \in \{2,\ldots,t-1,t+1,\ldots,\ell-1\}     \end{align*}}
    where $\widehat{X}_r = X_r$ for any $r \notin \{k_1, \ldots, k_\ell\}$.  
    By Lemma \ref{lem:additive_bound} and {Claim \ref{cla:last_touched_simple}(\ref{list:touching_1_inequality})},
    \begin{align*}
        {v_j(\widehat{X}_j) = v_j(Z_{k_{t-1}}) 
        \geq v_j(Z_j)
        \geq (1/({\alpha+1})) \cdot v_j(X_j).}
    \end{align*}
    By {Claim~\ref{cla:last_touched_simple}(\ref{list:touching_1_inequality})}, for {$s \in \{2,\ldots,t-1,t+1,\ldots,\ell-1\}$}, it holds that
       \begin{align*}
        v_{k_s}(\widehat{X}_{k_s}) = v_{k_s}((X_{k_s} \setminus Z_{k_s}) \cup Z_{k_{s-1}}) 
        = v_{k_s}(X_{k_s})-v_{k_s}(Z_{k_s}) + v_{k_s}(Z_{k_{s-1}}) 
        \geq v_{k_s}(X_{k_s}).
    \end{align*}
    {
    By Claim~\ref{cla:last_touched_simple}(\ref{list:touching_a_inequality}) and the fact that $X_{k_\ell} = Z_{k_\ell}$, it holds that
    \begin{align*}
        v_{k_\ell}(\widehat{X}_{k_\ell}) = v_{k_{\ell}}(X_{k_\ell} \cup Z_{k_{\ell-1}}) \geq v_{k_{\ell}}(X_{k_\ell}) + (1/\alpha) \cdot v_{k_{\ell}}(X_{k_{\ell}}) = ((\alpha+1)/{\alpha}) \cdot v_{k_\ell}(X_{k_\ell})
    \end{align*}
    }
      Combining the inequalities above gives, by optimality of $X$, that
    \begin{align*}
    1 &\geq \frac{v_j(\widehat{X}_j)}{v_j(X_j)} \cdot 
    \frac{v_i(\widehat{X}_i)}{v_i(X_i)} \cdot 
    \frac{v_{k_3}(\widehat{X}_{k_3})}{v_{k_3}(X_{k_3})} \cdots
    \frac{v_{k_\ell}(\widehat{X}_{k_\ell})}{v_{k_\ell}(X_{k_\ell})}  \\
    &\geq {(1/({\alpha+1}))} \cdot \frac{v_i(\widehat{X}_i)}{v_i(X_i)} \cdot 1 \cdots 1 \cdot {((\alpha+1)/{\alpha})}.
\end{align*}
Similarly to before, after rearranging we can apply Lemma~\ref{lem:trick}. This completes the proof.
\end{proof}

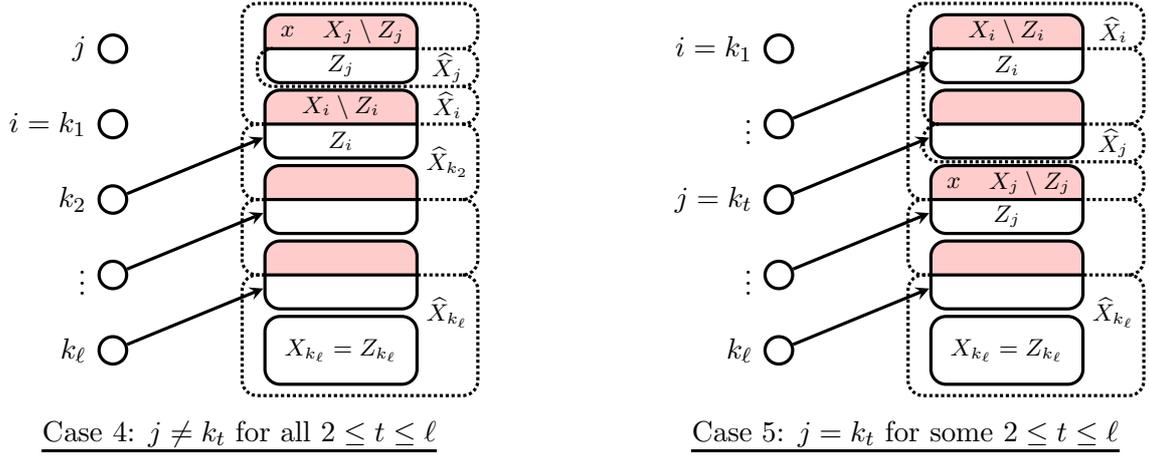
\begin{figure}[t] 
\begin{center}
\def\yscale{.5}
\def\xscale{1}
\begin{minipage}[b]{.45\textwidth}
\centering
\begin{tikzpicture}[very thick,xscale=\xscale,yscale=\yscale]
\foreach \i in {2,...,5}
{
    \begin{scope}
        \clip (1,2*\i) rectangle (3,2*\i+1);
        \node[fill, red!20!white, rounded corners=.2cm, minimum width=2*\xscale cm, minimum height=1.8*\yscale cm] at (2,2*\i){};
    \end{scope}
    \draw (1,2*\i) -- (3,2*\i);
}
\foreach \i in {1,...,5}
{
    \node[draw,circle] (j\i) at (-1,2*\i) {};
    \node[draw, rounded corners=.2cm, minimum width=2*\xscale cm, minimum height=1.8*\yscale cm] (X\i) at (2,2*\i){};
}
\node[anchor=east] at (-1.2,10) {$j$};
\node[anchor=east] at (-1.2,8) {$i=k_1$};
\node[anchor=east] at (-1.2,6) {$k_2$};
\node[anchor=east] at (-1.2,4) {$\vdots$};
\node[anchor=east] at (-1.2,2) {$k_\ell$};
{\footnotesize
    \node at (1.3,10.45) {$x$};
    \node at (2.3,10.45) {$X_{j}\setminus Z_{j}$};
    \node at (2,9.55) {$Z_{j}$};
    \node at (2,8.45) {$X_{i}\setminus Z_{i}$};
    \node at (2,7.55) {$Z_{i}$};
    \node at (2,2) {$X_{k_\ell} = Z_{k_\ell}$};
    \node at (3.4,9.5) {$\widehat{X}_{j}$};
    \node at (3.4,8.5) {$\widehat{X}_{i}$};
    \node at (3.4,7) {$\widehat{X}_{k_2}$};
    \node at (3.4,3) {$\widehat{X}_{k_\ell}$};
}
\draw [->,>=stealth] (j3) -- (X4.190);
\draw [->,>=stealth] (j2) -- (X3.190);
\draw [->,>=stealth] (j1) -- (X2.190);
\begin{scope}[densely dotted, rounded corners=.2cm]
    \draw (1,8) -- (.7,8) -- (.7,10) -- (.7,11.2) -- (3.8,11.2) -- (3.8,10) -- (3,10); 
    \draw (1,10) -- (.9,10) -- (.9,9) -- (3.8,9) -- (3.8,10) -- (3.6,10); 
    \draw (3.6,9) -- (3.8,9) -- (3.8,8) -- (3,8);
    \draw (3.6,8) -- (3.8,8) -- (3.8,6) -- (3,6);
    \draw (3.6,6) -- (3.8,6) -- (3.8,4) -- (3,4);
    \draw (1,8) -- (.7,8) -- (.7,6) -- (1,6);
    \draw (1,6) -- (.7,6) -- (.7,4) -- (1,4);
    \draw (3.6,4) -- (3.8,4) -- (3.8,.8) -- (.7,.8) -- (.7,4) -- (1,4);
\end{scope}
\end{tikzpicture}
\smallbreak
\underline{Case 4: $j\neq k_t$ for all $2 \leq t \leq \ell$}
\end{minipage}
\hfill
\begin{minipage}[b]{.45\textwidth}
\centering
\begin{tikzpicture}[very thick,xscale=\xscale,yscale=\yscale]
\foreach \i in {2,...,5}
{
    \begin{scope}
        \clip (1,2*\i) rectangle (3,2*\i+1);
        \node[fill, red!20!white, rounded corners=.2cm, minimum width=2*\xscale cm, minimum height=1.8*\yscale cm] at (2,2*\i){};
    \end{scope}
    \draw (1,2*\i) -- (3,2*\i);
}
\foreach \i in {1,...,5}
{
    \node[draw,circle] (j\i) at (-1,2*\i) {};
    \node[draw, rounded corners=.2cm, minimum width=2*\xscale cm, minimum height=1.8*\yscale cm] (X\i) at (2,2*\i){};
}
\node[anchor=east] at (-1.2,10) {$i=k_1$};
\node[anchor=east] at (-1.2,8) {$\vdots$};
\node[anchor=east] at (-1.2,6) {$j=k_t$};
\node[anchor=east] at (-1.2,4) {$\vdots$};
\node[anchor=east] at (-1.2,2) {$k_\ell$};
{\footnotesize
    \node at (1.3,6.45) {$x$};
    \node at (2.3,6.45) {$X_{j}\setminus Z_{j}$};
    \node at (2,5.55) {$Z_{j}$};
    \node at (2,10.45) {$X_{i}\setminus Z_{i}$};
    \node at (2,9.55) {$Z_{i}$};
    \node at (2,2) {$X_{k_\ell} = Z_{k_\ell}$};
    \node at (3.4,7.5) {$\widehat{X}_{j}$};
    \node at (3.4,10.5) {$\widehat{X}_{i}$};
    \node at (3.4,3) {$\widehat{X}_{k_\ell}$};
}
\draw [->,>=stealth] (j4) -- (X5.190);
\draw [->,>=stealth] (j3) -- (X4.190);
\draw [->,>=stealth] (j2) -- (X3.190);
\draw [->,>=stealth] (j1) -- (X2.190);
\begin{scope}[densely dotted, rounded corners=.2cm]
    \draw (1,6) -- (.7,6) -- (.7,10) -- (.7,11.2) -- (3.8,11.2) -- (3.8,10) -- (3,10); 
    \draw (1,8) -- (.9,8) -- (.9,7) -- (3.8,7) -- (3.8,8) -- (3.6,8); 
    \draw (3.6,10) -- (3.8,10) -- (3.8,8) -- (3,8);
    \draw (3.6,7) -- (3.8,7) -- (3.8,6) -- (3,6);
    \draw (1,10) -- (.9,10) -- (.9,8) -- (1,8);
    \draw (3.6,6) -- (3.8,6) -- (3.8,4) -- (3,4);
    \draw (1,6) -- (.7,6) -- (.7,4) -- (1,4);
    \draw (3.6,4) -- (3.8,4) -- (3.8,.8) -- (.7,.8) -- (.7,4) -- (1,4);
\end{scope}
\end{tikzpicture}
\smallbreak
\underline{Case 5: $j = k_t$ for some $2 \leq t \leq \ell$}
\end{minipage}
\end{center}
\caption{The structure of the new allocation $\widehat{X}$ in Case 4 and Case 5: the touching sequence has length $\ell >2$, and ends with condition (1). The edges represent the touching relation.}
\label{fig:additive_better_sep}
\end{figure}

Finally, we prove the main theorem of this section.

\begin{proof}[Proof of Theorem \ref{thm:complete_additive}]
    {Since $v_i(M_i) \geq v_i(Z_i)$ by {Claim~\ref{cla:add_ineq}(\ref{list:1_inequality})}, it follows from Lemma~\ref{lem:additive_better_sep} that the allocation $(M_1, \ldots, M_n)$ returned by the algorithm is $\alpha$-separated.} Hence, by Lemma~\ref{lem:make_complete}, there exists a complete  allocation {that is $\min(\alpha,\frac{1}{\alpha+1})$-EFX{, EF1,} and $\frac{1}{\alpha+1}$-MNW}. This proves the theorem  since for $\alpha \leq \varphi-1$ it holds that $\frac{1}{\alpha+1} \geq \alpha$ and so $\min(\alpha,\frac{1}{\alpha+1}) = \alpha$.
\end{proof}

\section{Subadditive Valuations}\label{sec:subadditive}

In this section, we study settings with subadditive valuations. 
In Section~\ref{sec:subadditive-partial}, we prove Theorem~\ref{thm:subadditive-partial-better} by constructing a partial allocation with the desired EFX and MNW guarantees. 
In Section~\ref{sec:subadditive-complete}, we prove Theorem~\ref{thm:complete_subadditive} by using the 
{techniques from Section~\ref{sec:par_to_com}}
to turn the partial allocation into a complete allocation with the same guarantees.

\subsection{Partial Allocations}\label{sec:subadditive-partial}

In this section, we prove the following theorem.

\begin{theorem}
     {Every instance with subadditive valuations admits a partial allocation that is $\alpha$-EFX and $\frac{1}{\alpha+1}$-MNW, for every $0 \leq \alpha \leq 1/2$.}
    \label{thm:subadditive-partial-better}
\end{theorem}

Algorithm~\ref{alg:new_subadditive_alg} receives 
{as input an arbitrary complete}
allocation $X = (X_1, \ldots, X_n)$, and
produces an allocation $M = (M_1, \ldots, M_n)$ that is $\alpha$-EFX and 
{satisfies $\nw(M) \geq \frac{1}{\alpha+1} \cdot \nw(X)$}.
{To prove Theorem~\ref{thm:subadditive-partial-better} we will eventually let $X$ be a MNW allocation. For the analysis of the algorithm, however, we treat $X$ as an arbitrary complete allocation because this then allows us to immediately prove Theorem~\ref{thm:subadditive-complete-poly} in Appendix~\ref{sec:comp}.}
For clarity, we present some of the operations that Algorithm~\ref{alg:new_subadditive_alg} performs as the procedure $\textsc{split}$ in Algorithm~\ref{alg:split_bundle}. The algorithm maintains a set of bundles $(Z_1, \ldots, Z_n)$ which initially satisfy $Z_i = X_i$. We say that the bundles $Z_j$ are {\em white}, the bundles $X_j \setminus Z_j$ are {\em red}, and any other bundle is {\em blue}.  See Figure~\ref{fig:subadditive_plit} for an additional description of the procedure $\textsc{split}$.

\begin{algorithm}[h!]
\caption{{Subadditive valuations.}}
\label{alg:new_subadditive_alg}
\begin{flushleft}
\hspace*{\algorithmicindent} \textbf{Input} $(X_1, \ldots, X_n)$ is {an arbitrary} complete allocation. \\
\hspace*{\algorithmicindent} \textbf{Output} $(M_1, \ldots, M_n)$ is a partial $\alpha$-EFX {allocation with $\nw(M) \geq \frac{1}{\alpha+1} \cdot \nw(X)$}.
\end{flushleft}
\vspace{-0.15in}
\begin{algorithmic}
[1]\State match $i$ to $J$ means $M_i \gets J$
\State unmatch $Z_i$ means $M_u \gets \bot$ if there is $u$ matched to $Z_i$
\State change $Z_i$ to $J$ means $Z_i \gets J$
\State shrink $X_i$ means $X_i \gets Z_i$
\Procedure{alg}{}
\State $Z \gets(X_1, \ldots, X_n)$
\State $M \gets(\bot, \ldots, \bot)$
\While{there is an agent $i$ with $M_i = \bot$}\label{line:while}
\State $i \gets $ any agent with $M_{i} = \bot$ \Comment{$i$ is an unmatched agent}\label{line:i_def}
\State $\mathcal{B} \gets \{ Z_{j}-g \bar j \in {[n]}, g \in Z_{j} \} \,\cup $\Comment{$\mathcal{B}$ is the set of available bundles}\label{line:b_def}
\State \;\;\;\;\;\;\; $\{ X_{j} \setminus Z_{j} \bar j \in {[n]} \} \,\cup $
\State \;\;\;\;\;\;\; $\{ M_{j}-g \bar j \in {[n]}, g \in M_{j} \}$
\State $J \gets$ any $J \in \mathcal{B}$ that maximizes $v_{i}(J)$ \Comment {$J$ is the favorite bundle of $i$}\label{line:j_def}
\If{$v_{i}(Z_{i}) \geq \alpha \cdot v_{i}(J')$ for all $J' \in \mathcal{B}$} \Comment{$Z_{i}$ is good enough for $i$ (Case 1)}\label{line:self_cond}
    \State unmatch $Z_i$
    \label{line:first_unmatch}
    \State match $i$ to $Z_i$\label{line:first_match}
\ElsIf{$J = Z_{j}-g$ for some $j$ and $g \in Z_{j}$} \Comment{$J$ is white (Case 2)}\label{line:if_white}
   \State run procedure \textsc{split} 
\ElsIf{$J = X_{j} \setminus Z_{j}$ for some $j$} \Comment{$J$ is red (Case 3)}
    \State shrink $X_j$\label{line:shrink1}
    \State match $i$ to $J$\label{line:m2}
\ElsIf{$J = M_{j}-g$ for some $j$ and $g \in M_{j}$} \Comment{$J$ is blue (Case 4)}
    \State unmatch $j$ (i.e., $M_j \gets \bot$)
    \label{line:third_unmatch}
    \State match $i$ to $J$\label{line:m3}
\EndIf
\EndWhile
\State \textbf{return} $(M_1, \ldots, M_n)$
\EndProcedure
\end{algorithmic}
\end{algorithm}

\begin{samepage}
\begin{algorithm}[p]
\caption{Case 2 of Algorithm~\ref{alg:new_subadditive_alg}.}
\label{alg:split_bundle}
\begin{algorithmic}
[1]\Procedure{split}{}
    \State unmatch $Z_j$
    \label{line:second_unmatch}
    \State match $i$ to $J$
    \label{line:match}
    \State $R \gets X_{j} \setminus J$ \Comment{$J \cup R = X_{j}$}\label{line:r_def}
    \State $\mathcal{K}_S \gets \{ k  \bar M_k = Z_k, \; v_k(Z_k) < \alpha \cdot v_k(S) \} $\Comment{{$\mathcal{K}_S$ is defined for any subset of $R$}}
    \If{$v_{j}(g) \geq (\alpha/(\alpha+1)) \cdot v_{j}(X_{j})$} \Comment{$\{g\}$ is good enough for $j$ (Case 2.1)}\label{line:g_cond}
        \State change $Z_j$ to $\{ g \}$\label{line:change1}
        \State shrink $X_j$\label{line:shrink2}
    \ElsIf{$\mathcal{K}_S = \emptyset$ for all $S \subseteq R$}\label{line:if_s}
        \If{$v_{j}(R) < (\alpha/(\alpha+1)) \cdot v_{j}(X_{j})$} \Comment{$J$ is good enough for $j$ (Case 2.2)}\label{line:l_perp_cond}
            \State change $Z_{j}$ to $J$\label{line:mod_z}
        \Else \Comment{$R$ is good enough for $j$ (Case 2.3)}
            \State change $Z_j$ to $R$\label{line:change2}
            \State shrink $X_j$\label{line:shrink3}
        \EndIf
    \Else
        \State $S \gets$ any $S \subseteq R$ of minimal size with $\mathcal{K}_S \neq \emptyset$\Comment{$S \subseteq R$}
        \State $k \gets $ any agent from $\mathcal{K}_{S}$
        \If{$v_{j}(J) \geq \alpha \cdot v_{j}(X_{j})$} \Comment{$J$ is good enough for $j$ (Case 2.4)}\label{line:l_cond}
            \State change $Z_j$ to $J$\label{line:change3}
            \State shrink $X_k$ and $X_j$\label{line:shrink4}
            \State match $k$ to $S$\label{line:k1}
        \ElsIf{$v_{j}(S) \geq (\alpha/(\alpha+1)) \cdot v_{j}(X_{j})$} \Comment{$S$ is good enough for $j$ (Case 2.5)}\label{line:s_cond}
            \State change $Z_{j}$ to $S$\label{line:change4}
            \State shrink $X_j$\label{line:shrink5}
        \Else \Comment{$R \setminus S$ is good enough for $j$ (Case 2.6)}
            \State change $Z_j$ to $R\setminus S$\label{line:change5}
            \State shrink $X_k$ and $X_j$\label{line:shrink6}
            \State match $k$ to $S$\label{line:k2}
        \EndIf
    \EndIf
\EndProcedure
\end{algorithmic}
\end{algorithm}

\begin{figure}[p]
    \begin{center}
    \begin{tikzpicture}[very thick,yscale=0.65,xscale=0.88]
\node at (3.5,1) {\Large$\Rightarrow$};
\node at (3.5,-2) {\Large$\Rightarrow$};
\begin{scope}
    \begin{scope}
        \clip (0,1) rectangle (3,2);
        \fill[red!20!white,rounded corners=.2cm]
            (0,0) rectangle (3,2);
    \end{scope}
    \draw[rounded corners=.2cm] (0,0) rectangle (3,2);
    \draw (0,1) -- (3,1);
    \node[draw,circle] (i) at (-1,0.5) {};
    \node at (-1.4,0.5) {$i^t$};
    \draw[->,>=stealth] (i) -- (0,.5);
    \node at (1.5,0.5) {$Z_{j^t}^{t-1}$};
    \node at (1.5,1.5) {$X_{j^t}^{t-1}\setminus Z_{j^t}^{t-1}$};
    \node at (2.6,0.5) {\small$g^t$};
\end{scope}
\begin{scope}[yshift=-3cm]
    \draw[rounded corners=.2cm]
    (0,0) rectangle (2.2,1);
    \draw[rounded corners=.2cm]
    (0,1) -- (2.2,1) -- (2.2,0) -- (3,0) -- (3,2) -- (0,2) -- cycle;
    \draw[rounded corners=.2cm]
    (2.3,0.1) rectangle (2.9,1.9);
    \node[draw,circle] (i) at (-1,0.5) {};
    \node[draw,circle] (k) at (-1,-.2) {};
    \node at (-1.4,0.5) {$i^t$};
    \node at (-1.4,-.2) {$k^t$};
    \draw[->,>=stealth] (i) -- (0,.5);
    \draw[->,>=stealth, rounded corners=.2cm]
    (k) -- (2.3,-.2) -- (2.4,0.2);
    \node at (1.25,0.5) {$J^t$};
    \node at (2.6,1) {$S^t$};
    \node at (1.25,1.6) {$R^t$};
\end{scope}
\begin{scope}[xshift=4cm]
    \begin{scope}
        \clip (0,1) rectangle (3,2);
        \fill[pattern=north east lines,rounded corners=.2cm]
        (0,1) -- (2.2,1) -- (2.2,0) -- (3,0) -- (3,2) -- (0,2) -- cycle;
    \end{scope}
    \draw (2,1) -- (3,1);
    \draw[fill=blue!20!white,rounded corners=.2cm] 
    (0,0) rectangle (2.2,1);
    \draw[rounded corners=.2cm]
    (0,1) -- (2.2,1) -- (2.2,0) -- (3,0) -- (3,2) -- (0,2) -- cycle;
    \node at (1.1,0.5) {\small$M_{i^t}^t$};
    \node at (2.6,0.5) {\small$g^t$};
    \node at (1.5,-.4) {\underline{Case 2.1}};
\end{scope}
\begin{scope}[xshift=7.5cm]
    \fill[red!20!white,rounded corners=.2cm]
    (0,0) rectangle (3,2);
    \fill[white] (0,0) rectangle (2.2,1);
    \draw[rounded corners=.2cm]
    (0,0) rectangle (3,2);
    \draw (0,1) -- (2.2,1) -- (2.2,0);
    \node at (1.1,0.5) {\small$M_{i^t}^t$};
    \node at (2.6,0.5) {\small$g^t$};
    \node at (1.5,-.4) {\underline{Case 2.2}};
\end{scope}
\begin{scope}[xshift=11cm]
    \draw[fill=blue!20!white,rounded corners=.2cm]
    (0,0) rectangle (2.2,1);
    \draw[rounded corners=.2cm]
    (0,1) -- (2.2,1) -- (2.2,0) -- (3,0) -- (3,2) -- (0,2) -- cycle;
    \node at (1.1,0.5) {\small$M_{i^t}^t$};
    \node at (2.6,0.5) {\small$g^t$};
    \node at (1.5,-.4) {\underline{Case 2.3}};
\end{scope}
\begin{scope}[yshift=-3cm,xshift=4cm]
    \draw[pattern=north east lines,opacity=1,rounded corners=.2cm]
    (0,1) -- (2.2,1) -- (2.2,0) -- (3,0) -- (3,2) -- (0,2) -- cycle;
    \draw[fill=blue!20!white,rounded corners=.2cm]
    (2.3,0.1) rectangle (2.9,1.9);
    \draw[rounded corners=.2cm] (0,0) rectangle (2.2,1);
    \node at (1.1,0.5) {\small$M_{i^t}^t$};
    \node at (2.6,1) {\footnotesize$M_{k^t}^t$};
    \node at (1.5,-.4) {\underline{Case 2.4}};
\end{scope}
\begin{scope}[yshift=-3cm,xshift=7.5cm]
    \draw[pattern=north east lines,opacity=1,rounded corners=.2cm]
    (0,1) -- (2.2,1) -- (2.2,0) -- (3,0) -- (3,2) -- (0,2) -- cycle;
    \draw[fill=blue!20!white,rounded corners=.2cm]
    (0,0) rectangle (2.2,1);
    \draw[fill=white,rounded corners=.2cm]
    (2.3,0.1) rectangle (2.9,1.9);
    \node at (1.1,0.5) {\small$M_{i^t}^t$};
    \node at (1.5,-.4) {\underline{Case 2.5}};
\end{scope}
\begin{scope}[yshift=-3cm,xshift=11cm]
    \draw[fill=blue!20!white,rounded corners=.2cm]
    (0,0) rectangle (2.2,1);
    \draw[rounded corners=.2cm]
    (0,1) -- (2.2,1) -- (2.2,0) -- (3,0) -- (3,2) -- (0,2) -- cycle;
    \draw[fill=blue!20!white,rounded corners=.2cm]
    (2.3,0.1) rectangle (2.9,1.9);
    \node at (1.1,0.5) {\small$M_{i^t}^t$};
    \node at (2.6,1) {\footnotesize$M_{k^t}^t$};
    \node at (1.5,-.4) {\underline{Case 2.6}};
\end{scope}

\end{tikzpicture}
    \end{center}
    \vspace{-0.5cm}
    \caption{Split operation, as defined in Algorithm~\ref{alg:split_bundle}.
    The new allocation is described by the matching $M^t$.
    Items that are thrown away (not in any white, red, or blue bundle) are hatched.
    Regarding the new structure of $X_{j^t}$, in each case the bundle $Z_{j^t}^{t}$ is colored in white and the bundle $X_{j^t}^{t}\setminus Z_{j^t}^{t}$ is colored in red.}
    \label{fig:subadditive_plit}
    \vspace{-0.5cm}
\end{figure}
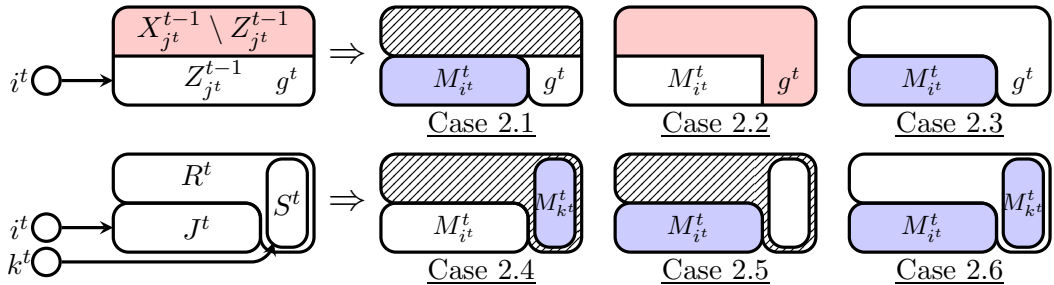
 
\end{samepage}

The high-level idea of the algorithm is as follows: initially, 
no agent is matched. 
The while loop inside the algorithm continues as long as there is an unmatched agent $i$. 
During each iteration, we define the set of available bundles to include {(1)} all the strict subsets of white bundles, {(2)} all the red bundles, and {(3)} all the strict subsets of the matched blue bundles\footnote{Since the valuations are monotone, it suffices to include the subsets obtained after removing a single element.}. 
Then, if $Z_{i}$ is ``good enough'' for $i$ (i.e., if $i$ gets $Z_{i}$, then she does not envy any other available bundle in the $\alpha$-EFX sense), the algorithm matches $Z_{i}$ to $i$. Otherwise, the algorithm picks the most valuable (from the perspective of $i$) available bundle $J$ and depending on whether $J$ is white, red, or blue, the algorithm modifies the bundles and the matching in an appropriate way to eliminate the envy of agent $i$.

We introduce the following timestamps to be able to talk about the state of the algorithm during different phases.

\begin{definition}[Timestamps]
    Let $T$ denote the total number of iterations in a single run of the algorithm. Let $Z_i^t$, $X_i^t$, and $M_i^t$ denote the values of the global variables $Z_i$, $X_i$, and $M_i$ after exactly $t$ iterations for $0 \leq t \leq T$.  Let $i^t$, $j^t$, $u^t$, $k^t$, $g^t$, $R^t$, $S^t$, $J^t$, $\mathcal{B}^t$, $\mathcal{K}_S^t$ denote the values corresponding to the local variables defined during the $t$-th iteration for $1 \leq t \leq T$.   For simplicity of notation, also define $X_i^{T+1} = X_i^T$, $Z_i^{T+1} = Z_i^T$, $M_i^{T+1} = M_i^T$, and $\mathcal{B}^{T+1} = \{ Z_{j}^{T} - g \bar j \in \agents, g \in Z_j^{T}\} \cup \{X_j^{T} \setminus Z_j^{T} \bar j \in \agents\} \cup \{ M_j^{T} -g \bar j \in \agents, g  \in M_j^{T} \}.$
\end{definition}

In particular, $Z_i^0$, $X_i^0$, and $M_i^0$ are the initial values of the global variables, and $Z_i^T$, $X_i^T$, and $M_i^T$ are the final values. {Note that some of the local variables might not be defined in every case, but in the algorithm analysis, we will refer to them only in the cases where they are well-defined.} To see the relation between the timestamps for the global and local variables, consider the following example: line \ref{line:i_def} of Algorithm~\ref{alg:new_subadditive_alg} implies that $M_{i^t}^{t-1} = \bot$.

Algorithm~\ref{alg:new_subadditive_alg}, {like} Algorithm~\ref{alg:alloc}, sometimes shrinks $Z_i$ by removing some of its elements (lines \ref{line:change1} and \ref{line:change3} of Algorithm~\ref{alg:split_bundle}).
However, Algorithm~\ref{alg:new_subadditive_alg} might also modify the bundles $Z$ and $X$ in additional ways.
Specifically, Algorithm~\ref{alg:new_subadditive_alg} might change $Z_i$ to a subset of the red bundle $X_i \setminus Z_i$ (lines \ref{line:change2}, \ref{line:change4}, and \ref{line:change5} of Algorithm~\ref{alg:split_bundle}) or it might shrink $X_i$ by setting it to $Z_i$ (line \ref{line:shrink1} of Algorithm~\ref{alg:new_subadditive_alg} and lines \ref{line:shrink2}, \ref{line:shrink3}, \ref{line:shrink4}, \ref{line:shrink5}, \ref{line:shrink6} of Algorithm~\ref{alg:split_bundle}). 
Nonetheless, Algorithm~\ref{alg:new_subadditive_alg} preserves the same invariant as Algorithm~\ref{alg:alloc} so that $v_i(Z_i) \geq (1/(\alpha+1)) \cdot v_i(X_i)$. 
This and some other basic properties of the algorithm are captured in the following technical claim. The proof is deferred to Appendix~\ref{sec:proofs-subadditive}.

\begin{restatable}
{claim}{clabasic}
After $t$ iterations of Algorithm~\ref{alg:new_subadditive_alg}, the following holds
\begin{enumerate}[(i)]
        \item $v_i(X_i^t \setminus Z_i^t) \leq (\alpha/(\alpha+1)) \cdot v_i(X_i^t)$ for each agent $i$,
        \item $v_i(Z_i^t) \geq (1/(\alpha+1)) \cdot v_i(X_i^t)$ for each agent $i$,
        \item $i^t$, $j^t$, and $k^t$ are distinct, and
        \item $J^t$ and $S^t$ are non-empty.
    \end{enumerate}
\label{cla:basic}
\end{restatable}

We refer to $M^t$ as a matching between agents and bundles of items. 
We say that after $t$ iterations, an agent $i$ is {\em matched} to $M_i^t$ if $M_i^t \neq \bot$, and that $i$ is {\em unmatched} otherwise. We {also} define the matching time as follows.

\begin{definition}[Matching time] Let $i$ be any matched agent. We define the $\emph{matching time}$ (before $t$) of $i$ to be the largest $1 \leq r \leq t$ so that $M_i^{r} \neq M_i^{r-1}$, i.e, $i$ was matched during the $r$-th iteration.
\end{definition}
Note that the matching time is well-defined because $M_i^0 = \bot$ and $M_i^{T} \neq \bot$.

 {Let us first make a few simple observations that are crucial in the analysis of the algorithm. First,} throughout the execution of the algorithm, every matched agent is either matched to a white bundle $Z_j^t$ or $i$ is matched to a blue bundle $M_i^t$, which is disjoint from all white and red bundles. 
 {Second,} the matching is constructed so that no item is allocated to two agents. 
 {Third,} between any two iterations during which agent $i$ is matched, the bundle $M_i^t$ is not changed in any way. These properties are formalized in the following claim, {whose} proof is deferred to Appendix~\ref{sec:proofs-subadditive}.

\vspace{0.1in}

\begin{restatable}
{claim}{clazconst}
After $t$ iterations of the algorithm, the following holds
\begin{enumerate}[(i)]
    \item all the white bundles $M_i^t$ are disjoint, and
    \item all the blue bundles $M_i^t$ and the bundles $X_1^t, \ldots, X_n^t$ are disjoint.
\end{enumerate}
Moreover, for any matched agent $i$ with $r$ defined as the matching time of $i$ (before $t$),
\begin{enumerate}[(i)]
    \setcounter{enumi}{2}
    \item $M_i^r = M_i^{r+1} = \ldots = M_i^t$, and
    \item if $M_i^r = Z_a^r$ for some agent $a$, then $Z_a^r = Z_a^{r+1} = \ldots = Z_a^t$.
\end{enumerate}
\label{cla:z_const}\end{restatable}

{Furthermore, observe that} the matching constructed by the algorithm induces a certain structure of chains and cycles of agents which are described in the following claim. See Figure~\ref{fig:subadditive_structure} for a pictorial representation of this structure.
The proof of the claim is deferred to Appendix~\ref{sec:proofs-subadditive}.

\begin{restatable}
{claim}{clachains}
    After $t$ iterations of the algorithm, all the agents are divided into disjoint chains and cycles where each chain $i_1, \ldots, i_\ell$ is of one of the following forms
    \begin{enumerate}[(i)]
        \item $Z_{i_1}^t = X_{i_1}^t$, $M_{i_j}^t = Z_{i_{j+1}}^t$ for $1 \leq j < \ell$, and $M_{i_\ell}^t = \bot$, or
        \item $Z_{i_1}^t = X_{i_1}^t$, $M_{i_j}^t = Z_{i_{j+1}}^t$ for $1 \leq j < \ell$, and $M_{i_\ell}^t$ is blue,
    \end{enumerate}
    and each cycle $i_1, \ldots, i_\ell$ is of one of the following forms
    \begin{enumerate}[(i)]
        \setcounter{enumi}{2}
        \item $\ell = 1$, and $M_{i_1}^t = Z_{i_1}^t$, or
        \item $\ell > 1$, $M_{i_j}^t = Z_{i_{j+1}}^t$ for $1 \leq j < \ell$, and $M_{i_\ell}^t = Z_{i_1}^t$.
    \end{enumerate}
\label{cla:chains}
\end{restatable}

\begin{figure}[t]
    \begin{center}
\def\yscale{0.5}
\begin{minipage}[b]{.32\textwidth}
\centering
\begin{tikzpicture}[very thick,yscale=\yscale]
\foreach \i in {1,...,3}
{
    \begin{scope}
        \clip (1,2*\i) rectangle (3,2*\i+0.9);
        \node[fill, red!20!white, rounded corners=.2cm, minimum width=2cm, minimum height=1.8*\yscale cm] at (2,2*\i){};
    \end{scope}
    \draw (1,2*\i) -- (3,2*\i);
}
\foreach \i in {1,...,4}
{
    \node[draw,circle] (j\i) at (0,2*\i) {};
    \node[draw, rounded corners=.2cm, minimum width=2cm, minimum height=1.8*\yscale cm] (X\i) at (2,2*\i){};
}
{\footnotesize
    \draw (j4) edge node[pos=0.2,above] {
    } (X3.190);
    \draw (j3) edge node[pos=0,below] {
    } (X2.190);
    \draw (j2) edge node[pos=0,below] {
    } (X1.190);
    \node at (2,8) {$Z^t_{i_1} = X^t_{i_1}$};
    \node at (2,6.45) {$X_{i_2}^t\setminus Z^t_{i_2}$};
    \node at (2,2.45) {$X_{i_\ell}^t\setminus Z^t_{i_\ell}$};
    \node at (2,5.55) {$Z^t_{i_2}$};
    \node at (2,1.55) {$Z^t_{i_\ell}$};
}
\node[anchor=east] at (-.2,8) {$i_1$};
\node[anchor=east] at (-.2,6) {$i_2$};
\node[anchor=east] at (-.3,4.2) {$\vdots$};
\node[anchor=east] at (-.2,2) {$i_\ell$};
\end{tikzpicture}
\par\vspace{.9cm}\underline{Chain of type (i)}
\end{minipage}
\begin{minipage}[b]{.32\textwidth}
\centering
\begin{tikzpicture}[very thick,yscale=\yscale]
\foreach \i in {1,...,3}
{
    \begin{scope}
        \clip (1,2*\i) rectangle (3,2*\i+0.9);
        \node[fill, red!20!white, rounded corners=.2cm, minimum width=2cm, minimum height=1.8*\yscale cm] at (2,2*\i){};
    \end{scope}
    \draw (1,2*\i) -- (3,2*\i);
}
\foreach \i in {1,...,4}
{
    \node[draw,circle] (j\i) at (0,2*\i) {};
    \node[draw, rounded corners=.2cm, minimum width=2cm, minimum height=1.8*\yscale cm] (X\i) at (2,2*\i){};
}
\node[draw, fill=blue!20!white, rounded corners=.2cm, minimum width=1cm, minimum height=1*\yscale cm] (X0) at (1.5,0) {};
{\footnotesize
    \draw (j4) edge node[pos=0.2,above] {
    } (X3.190);
    \draw (j3) edge node[pos=0,below] {
    } (X2.190);
    \draw (j2) edge node[pos=0,below] {
    } (X1.190);
    \draw (j1) edge node[pos=0,below] {
    }  (X0.180);
    \node at (2,8) {$Z^t_{i_1} = X^t_{i_1}$};
    \node at (2,6.45) {$X_{i_2}^t\setminus Z^t_{i_2}$};
    \node at (2,2.45) {$X_{i_\ell}^t\setminus Z^t_{i_\ell}$};
    \node at (2,5.55) {$Z^t_{i_2}$};
    \node at (2,1.55) {$Z^t_{i_\ell}$};
}
\node[anchor=east] at (-.2,8) {$i_1$};
\node[anchor=east] at (-.2,6) {$i_2$};
\node[anchor=east] at (-.3,4.2) {$\vdots$};
\node[anchor=east] at (-.2,2) {$i_\ell$};
\end{tikzpicture}
\smallbreak\underline{Chain of type (ii)}
\end{minipage}
\begin{minipage}[b]{.32\textwidth}
\centering
\begin{tikzpicture}[very thick,yscale=\yscale]
\foreach \i in {1}
{
    \begin{scope}
        \clip (1,2*\i) rectangle (3,2*\i+0.9);
        \node[fill, red!20!white, rounded corners=.2cm, minimum width=2cm, minimum height=1.8*\yscale cm] at (2,2*\i){};
    \end{scope}
    \draw (1,2*\i) -- (3,2*\i);
}
\foreach \i in {1}
{
    \node[draw,circle] (j\i) at (0,2*\i) {};
    \node[draw, rounded corners=.2cm, minimum width=2cm, minimum height=1.8*\yscale cm] (X\i) at (2,2*\i){};
}
\node[anchor=east] at (-.2,2) {$i_1$};
{
    \footnotesize
    \draw (j1) edge node[pos=.4,below] {
    }  (X1.190);
    \node at (2,2.45) {$X_{i_1}^t\setminus Z^t_{i_1}$};
    \node at (2,1.55) {$Z^t_{i_1}$};
}
\end{tikzpicture}
\underline{Cycle of type (iii)}
\bigbreak
\begin{tikzpicture}[very thick,yscale=\yscale]
\foreach \i in {1,...,3}
{
    \begin{scope}
        \clip (1,2*\i) rectangle (3,2*\i+0.9);
        \node[fill, red!20!white, rounded corners=.2cm, minimum width=2cm, minimum height=1.8*\yscale cm] at (2,2*\i){};
    \end{scope}
    \draw (1,2*\i) -- (3,2*\i);
}
\foreach \i in {1,...,3}
{
    \node[draw,circle] (j\i) at (0,2*\i) {};
    \node[draw, rounded corners=.2cm, minimum width=2cm, minimum height=1.8*\yscale cm] (X\i) at (2,2*\i){};
}
{
    \footnotesize
    \node at (2,6.45) {$X_{i_1}^t\setminus Z^t_{i_1}$};
    \node at (2,5.55) {$Z^t_{i_1}$};
    \node at (2,2.45) {$X_{i_\ell}^t\setminus Z^t_{i_\ell}$};
    \node at (2,1.55) {$Z^t_{i_\ell}$};
    \draw (j3) edge node[pos=0,below] {
    }  (X2.190);
    \draw (j2) edge node[pos=0,below] {
    }  (X1.190);
    \draw (j1) edge node[pos=-.06,below] {
    }  (X3.190);
}
\node[anchor=east] at (-.2,6) {$i_1$};
\node[anchor=east] at (-.3,4.2) {$\vdots$};
\node[anchor=east] at (-.2,2) {$i_\ell$};
\end{tikzpicture}
\underline{Cycle of type (iv)}
\end{minipage}
\end{center}
    \vspace{-.5cm}
    \caption{The structure of the matching $M^t$ over the course of Algorithms~\ref{alg:new_subadditive_alg}, as described in Claim~\ref{cla:chains}. Each edge connects an agent $k$ to her bundle $M_k^t$. In the beginning, everyone is unmatched and there are only chains of type (i) of length $\ell=1$. In the end, everyone is matched and there are no chains of type (i).}
    \label{fig:subadditive_structure}
\end{figure}
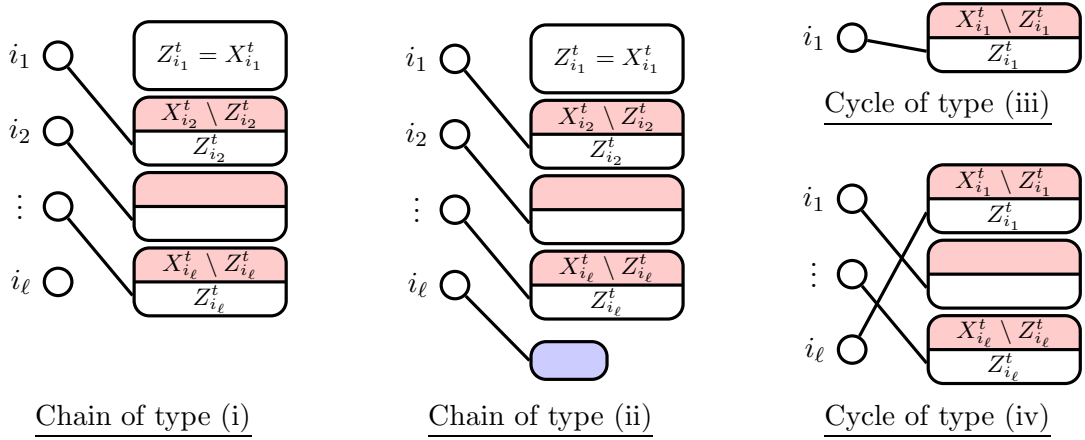

Let us first analyze the running time of the algorithm.

\begin{restatable}{lemma}{lemsubaddruntime}
    The algorithm terminates in a polynomial (in $n$ and $m$) number of iterations. Moreover, at the end of the run of the algorithm, it holds that $M_i^t \neq \bot$ for all $i$.
\end{restatable}

\begin{proof}
Let $\mathsf{White}_t = \bigcup_{i \in \agents} Z_i^t$, $\mathsf{Blue}_t = \bigcup_{M_i^t \text{ is blue}} M_i^t$, and $\mathsf{Red}_t = \bigcup_{i \in \agents} (X_i^t \setminus Z_i^t)$ be the sets of items contained in white, blue, and red bundles, respectively. Let $\mathsf{Deleted}_t = [m] \setminus (\mathsf{White}_t \cup \mathsf{Blue}_t \cup \mathsf{Red}_t)$. Consider the tuple ${\phi}_t = (|\mathsf{Deleted}_t|, |\mathsf{Blue}_t|, |\mathsf{Red}_t|)$. There are $(m+1)^3$ possible values that ${\phi}_t$ can take and by the design of the algorithm, ${\phi}_t$ increases lexicographically with every iteration of the algorithm. Hence, the number of iterations is bounded by $(m+1)^3$.

The second part of the statement follows because of the while loop condition.
\end{proof}

{We now present the first key lemma in the analysis.}

\begin{restatable}{lemma}{lemfeas}\label{lem:feas}
    The partial allocation $(M_1^T, \ldots, M_n^T)$ returned by the algorithm is $\alpha$-EFX.
\end{restatable}

{To prove that the allocation returned by the algorithm is $\alpha$-EFX, 
we  show that a stronger property is preserved throughout the execution of the algorithm, namely, for every matched agent $i$, it holds that $v_i(M_i^t) \geq \alpha \cdot v_i(J)$ for all $J \in \mathcal{B}^t$. It is not too hard to see that this property is satisfied when an agent gets matched. Showing that the property is preserved when the set of available bundles changes requires a bit more care.}

First, recall that, as we discussed earlier, the variables $Z$ and $X$ might be modified in multiple ways, which also means that the set of available bundles might change in {a few} different ways. Hence, we introduce the following definition to capture the difference between the set of available bundles before and after the $t$-th iteration.

\begin{definition}[New available bundles] For any $t > 0$, we say that $H \in \mathcal{B}^{t+1}$ is a \emph{new available bundle} (at time $t$) if there is no $H' \supseteq H$ so that $H' \in \mathcal{B}^{t}$, i.e., the bundle $H$ becomes available due to the changes made during the $t$-th 
iteration.
\end{definition}

Below, we {characterize} all the new available bundles. The proof is deferred to Appendix~\ref{sec:proofs-subadditive}.

\begin{restatable}
{claim}{clanewbund}
 Let $t > 0$ {and $H \in \mathcal{B}^{t+1}$ be a new available bundle (at time $t$). Then, either
 \begin{itemize}
     \item $H \subseteq R^{t}$ (and the algorithm went into Case 2.2 or 2.3), or
     \item $H \subsetneq S^{t} \subseteq R^{t}$ {or $H \subseteq R^t \setminus S^t$} (and the algorithm went into Case 2.4, 2.5, or 2.6).
 \end{itemize}
In the remaining cases, there are no new available bundles.}
\label{cla:new_bund}
\end{restatable}

{The following lemma captures the final observation needed in the proof of Lemma~\ref{lem:feas}.
}

\begin{lemma}\label{lem:remove_one}
Suppose that the algorithm went into Case 2 during the $t$-th iteration. Then, for any $r < t$, there is some $g' \in Z_{j^t}^{r-1}$ so that $R^t \subseteq X_{j^t}^{r-1}-g'$.
\end{lemma}
\begin{proof}
{Note that for any $1 \leq w \leq T+1$, it holds that $Z_a^w = Z_a^{w-1}$ for any $a \neq j^w$, and either $X_{j^w}^w = Z_{j^w}^w \subseteq R^{w}$ (in Cases 2.3, 2.5, and 2.6) or $Z_{j^w}^w \subseteq Z_{j^w}^{w-1}$ (in the remaining cases).
    Hence,} it must be the case that either
        (i) $Z_{j^t}^{t-1} \subseteq Z_{j^t}^{r-1}$, 
        or (ii) there is some $r \leq w \leq t-1$ so that $X_{j^t}^w = Z_{j^t}^w \subseteq R^w$ and $Z_{j^t}^{w-1} \subseteq Z_{j^t}^{r-1}$ (i.e., the $w$-th iteration was the first iteration after the $(r-1)$-st one where the bundle $Z_{j^t}$ was changed to something else than a subset of itself).
        
        In case (i), pick any $g' \in J^t$,  {which is well-defined since} $J^t$ is non-empty by Claim~\ref{cla:basic}(iv). Since $J^t \subseteq Z_{j^t}^{t-1} \subseteq Z_{j^t}^{r-1}$, it holds that $g' \in Z_{j^t}^{r-1}$. 
        The claim follows because since $J^t \cap R^t = \emptyset$, it holds that $g' \notin R^t$, and so $R^t \subseteq X_{j^t}^{t-1} -g' \subseteq X_{j^t}^{r-1} - g'$.

        In case (ii), pick any $g' \in J^w$,  {which is well-defined since} $J^w$ is non-empty by Claim~\ref{cla:basic}(iv). Since  $J^w \subseteq Z_{j^t}^{w-1} \subseteq Z_{j^t}^{r-1}$, it holds that $g' \in Z_{j^t}^{r-1}$.
        The claim follows because since $J^w \cap R^w = \emptyset$, it holds that $g' \notin R^w$, and so
        $R^t \subseteq X_{j^t}^{t-1} \subseteq X_{j^t}^w \subseteq R^w \subseteq X_{j^t}^{w-1}-g' \subseteq X_{j^t}^{r-1}-g'$.
\end{proof}

We are now ready to prove the key lemma.

\begin{proof}[Proof of Lemma~\ref{lem:feas}]
    We show by induction on $t$ that for every $0 \leq t \leq T+1$ and every matched agent $i$, it holds that $v_i(M_i^t) \geq \alpha \cdot  v_i(J)$ for all $J \in \mathcal{B}^t$. Then, the result follows since it holds that $M_j^T-g \in \mathcal{B}^{T+1}$ for any $j$ and $g \in M_j^T$ and so for any agent $i$, it must be that $v_i(M_i^T) \geq \alpha \cdot v_i(M_j^T-g)$. At $t = 0$, no agents are matched and hence the statement holds. Let $t\geq 0$ and $i$ be such that $M_{i}^{t+1} \neq \bot$. {The inductive assumption is that if $M_a^t \neq \bot$, then $v_a(M_a^t) \geq \alpha \cdot v_a(J)$ for all agents $a$ and  $J \in \mathcal{B}^t$.} Let $r$ be the matching time of $i$ (before $t+1$).

    \underline{Assume that $r=t+1$ and $M_i^{t+1} = Z_i^{t+1}$.} 
    It must be the case that the algorithm went into Case 1 during the $r$-th iteration with $i^r = i$.
    By the if condition in line \ref{line:self_cond} of Algorithm~\ref{alg:new_subadditive_alg}, we have $v_i(Z_i^r) \geq \alpha \cdot v_i(J)$ for all $J \in \mathcal{B}^{r}$ which gives the desired {guarantee}.

\underline{Assume that $r = t+1$ and $M_i^{t+1} = J^r$.} It holds that
    $v_i(M_i^{t+1}) = v_i(J^{t+1}) \geq v_i(J)$ for all $J \in \mathcal{B}^{t+1}$ since $\textstyle{J^{t+1} \in \argmax_{J \in \mathcal{B}^{t+1}} v_i(J)}${, and the statement follows.}

    Note that if $r \leq t$, then for any $H \in \mathcal{B}^{t+1}$ that is {\em not} a new available bundle (at time $t$), there is some $H' \supseteq H$ so that $H' \in \mathcal{B}^t$. By the inductive assumption, it holds that  $v_i(M_i^{t+1}) = v_i(M_i^t) \geq \alpha \cdot v_i(H') \geq \alpha \cdot v_i(H)$. Thus, it is enough to consider new available bundles.

    \underline{Assume that $r \leq t$ and $M_i^{t+1} = Z_i^{t+1}$.} 
    If $r = t$, then the algorithm went to Case 1 during the $r$-th iteration, and there are no new available bundles (at time $t$). {Assume that $r < t$ and l}et $H \in \mathcal{B}^{t+1}$ be a new available bundle (at time $t$).

    {First, observe that if the algorithm went into Cases 2.4, 2.5, or 2.6 during the $t$-th iteration, then it holds that $g^t \in S^t$. Indeed, if $g^t \notin S^t$, then $S^t \subseteq X_{j^t}^{t-1} \setminus Z_{j^t}^{t-1} \in \mathcal{B}^t$, and so, by the inductive assumption, $v_{k^t}(Z_{k^t}^{t-1}) = v_{k^t}(Z_{k^t}^t) \geq \alpha \cdot v_{k^t}(S^t)$ which contradicts  $k^t \in \mathcal{K}_{S^t}^t$. }
    
    We {now} claim that $\mathcal{K}_{H}^{t}$ is an empty set. For the purpose of contradiction, assume that the contrary holds. 
    It cannot be the case that the algorithm went into Case 2.2 or 2.3 during the $t$-th iteration because of the if condition in line~\ref{line:if_s} (Algorithm~\ref{alg:split_bundle}).
    Hence, by Claim~\ref{cla:new_bund}, it must be that $H \subsetneq S^{t}$ {or $H \subseteq R^t \setminus S^t$}. {If $H \subsetneq S^t$, then recall that} $S^{t}$ was chosen to be the minimal set with non-empty $\mathcal{K}_{S^{t}}^{t}$ which contradicts the assumption that $\mathcal{K}_{H}^{t}$ is non-empty. {If $H \subseteq R^t \setminus S^t$, then note that, since $g^t \in S^t$, it holds that $R^t \setminus S^t \subseteq X_{j^t}^{t-1} \setminus Z_{j^t}^{t-1} \in \mathcal{B}^t$ which contradicts the assumption that $H$ is a new available bundle.} This proves the claim. 
    
    It follows that $v_i(Z_{i}^{t-1})\geq \alpha \cdot v_i(H)$ since $i \notin \mathcal{K}_{H}^{t}$ {while} $M_{i}^{t-1} = Z_{i}^{t-1}$. {This} gives the desired inequality as $M_i^{t+1} = Z_i^{t-1}$ by Claim~\ref{cla:z_const}(iv) since $r < t$. 
    
    \underline{Assume that $r \leq t$ and $M_i^{t+1} = J^r$.} Consider $H \in \mathcal{B}^{t+1}$ that is a new available bundle (at time $t$). By Lemma~\ref{lem:remove_one}, pick $g' \in Z_{j^t}^{r-1}$ so that $R^t \subseteq X_{j^t}^{r-1}-g'$.
    Observe that $v_{i}(Z_{j^t}^{r-1}-g') \leq v_{i}(M_i^{t+1})$ and $v_i(X_{j^t}^{r-1}\setminus Z_{j^t}^{r-1}) \leq v_i(M_i^{t+1})$ since both $Z_{j^t}^{r-1}-g'$ and $X_{j^t}^{r-1}\setminus Z_{j^t}^{r-1}$  belong to $\mathcal{B}^r$ and $M_{i}^{t+1} = M_i^r \in \argmax_{J \in \mathcal{B}^r} v_i(J)$. 
     It follows that 
     \begin{align*}
       v_i(H) 
       &\leq v_i(R^t) && (\text{by Claim~\ref{cla:new_bund}}) \\
       &\leq v_i(X_{j^t}^{r-1}-g') && (\text{by the choice of }g') \\
       &\leq v_i(Z_{j^t}^{r-1}-g') + v_i(X_{j^t}^{r-1}\setminus Z_{j^t}^{r-1}) && (\text{by subadditivity of }v_i) \\
       &\leq 2 \cdot v_i(M_i^{t+1})  && (\text{by the above}) \\
       &\leq (1/\alpha) \cdot v_i(M_i^{t+1}) && (\text{since } \alpha \leq 1/2)
     \intertext{
     which is the desired guarantee. \endgraf
 \underline{Assume that $r \leq t+1$ and $M_i^{t+1} = S^r$.} Observe that
 }
        v_i(M_i^{t+1}) 
        &= v_i(S^{r}) && (\text{by the assumption})  \\
        &> (1/\alpha) \cdot v_i(Z_i^{r-1}) && (\text{since } i \in \mathcal{K}_{S^r}^r)  \\
        &= (1/\alpha) \cdot v_i(M_i^{r-1}) && (\text{since } i \in \mathcal{K}_{S^r}^r) \\
        &\geq v_i(J) && (\text{by the inductive assumption})
    \end{align*}
    for all $J \in \mathcal{B}^{r-1}$. {Moreover, it holds that $r \geq 2$ since for the algorithm to match $i$ to $S^r$ during the $r$-th iteration, $i$ must have already been matched during some earlier iteration.}
    
    First, if $r = t+1$, consider any $H \in \mathcal{B}^{t+1}$ that is {\em not} a new available bundle (at time $t$). Then, there is some $H' \in \mathcal{B}^{t}$ so that $H' \supseteq H$ and it follows from the observation above that $v_i(M_i^{t+1}) \geq v_i(H') \geq v_i(H)$.
    
Now, for any $r \leq t+1$, consider $H$ that is a new available bundle (at time $t$). By Lemma~\ref{lem:remove_one}, pick $g' \in Z_{j^t}^{r-2}$ so that $R^t \subseteq X_{j^t}^{r-2}-g'$.
    It follows from the observation above that $v_{i}(Z_{j^t}^{r-2}-g') \leq v_{i}(M_i^{t+1})$ and $v_i(X_{j^t}^{r-2}\setminus Z_{j^t}^{r-2}) \leq v_i(M_i^{t+1})$ since both $Z_{j^t}^{r-2}-g'$ and $ X_{j^t}^{r-2}\setminus Z_{j^t}^{r-2}$ belong to $\mathcal{B}^{r-1}$. Hence, it holds that
    \begin{align*}
     v_i(H) 
       &\leq v_i(R^t) && (\text{by Claim~\ref{cla:new_bund}}) \\
       &\leq v_i(X_{j^t}^{r-2}-g') && (\text{by the choice of }g') \\
       &\leq v_i(Z_{j^t}^{r-2}-g') + v_i(X_{j^t}^{r-2}\setminus Z_{j^t}^{r-2}) && (\text{by subadditivity of }v_i) \\
       &\leq 2 \cdot v_i(M_i^{t+1})  && (\text{by the above}) \\
       &\leq (1/\alpha) \cdot v_i(M_i^{t+1}) && (\text{since } \alpha \leq 1/2)
     \end{align*}
     which completes the proof of the lemma.
\end{proof}

{We now present the second key lemma in the analysis.}

\begin{lemma}\label{lem:subadd_mnw}
    The partial allocation $(M_1^T, \ldots, M_n^T)$ 
    {satisfies $\nw(M) \geq \frac{1}{\alpha+1} \cdot \nw(X)$}.
\end{lemma}

{{To prove the desired guarantee on the Nash welfare of the final allocation}, we establish several additional auxiliary lemmas.

{First, observe that for any agent $a$, the bundle $Z_a$ never becomes more valuable to $a$ after an iteration of the algorithm than before that iteration.}

\begin{restatable}{lemma}{lemzincr}
\label{lem:z_incr}
    It holds that $v_a(Z_a^{t}) \leq v_a(Z_a^{t-1})$ for any agent $a$ and $t > 0$.
\end{restatable}

\begin{proof}
    It only happens that $Z_a^{t} \neq Z_a^{t-1}$ for some {agent} $a$ if the algorithm goes into Case 2 during the $t$-th iteration {and $j^t = a$}. In Cases 2.1, 2.2, and 2.4, it holds that $Z_{j^t}^{t} \subseteq Z_{j^t}^{t-1}$. In Cases 2.3, 2.5, and 2.6, it holds that $Z_{j^t}^{t} \subseteq R^t = X_{j^t}^{t-1} \setminus (Z_{j^t}^{t-1}  -g^t )$. Moreover, $v_{j^t}(X_{j^t}^{t-1} \setminus Z_{j^t}^{t-1}) \leq (\alpha/(\alpha+1)) \cdot v_{j^t}(X_{j^t}^{t-1}) $ by Claim~\ref{cla:basic}(i) and $v_{j^t}(g^t) \leq (\alpha/(\alpha+1)) \cdot v_{j^t}(X_{j^t}^{t-1}) $ by the if condition in line \ref{line:g_cond} of Algorithm \ref{alg:split_bundle}.
    Observe that
    \begin{align*}
        v_{j^t}(Z_{j^t}^{t}) 
        &\leq v_{j^t}(X_{j^t}^{t-1} \setminus (Z_{j^t}^{t-1} - g^t)) && (\text{by monotonicity of } v_{j^t}) \\
        &= v_{j^t}((X_{j^t}^{t-1} \setminus Z_{j^t}^{t-1}) + g^t) && (\text{by simple algebra})  \\
        &\leq v_{j^t}(X_{j^t}^{t-1} \setminus Z_{j^t}^{t-1})  + v_{j^t}(g^t) && (\text{by subadditivity of }v_{j^t}) \\
        &\leq 2 \cdot (\alpha/(\alpha+1)) \cdot v_{j^t}(X_{j^t}^{t-1}) && (\text{by the above}) \\
        &\leq 2 \cdot \alpha \cdot v_{j^t}(Z_{j^t}^{t-1}) && (\text{by Claim \ref{cla:basic}(ii)} ) \\
        &\leq v_{j^t}(Z_{j^t}^{t-1}) && (\text{since } \alpha \leq 1/2)
    \end{align*}
    and hence the result follows.
\end{proof}

{Second, observe} that after $t$ iterations, any agent $i$ contained in a chain of type (ii) or a cycle of type (iv) is better off by at least a factor of $(1/\alpha)$ compared to the subset $Z_i^t$ of her initial allocation.}

\begin{restatable}{lemma}{lemsubineq}
    \label{lem:sub_ineq}
    After $t$ iterations, for every matched agent $i$ with $M_i^t \neq Z_i^t$, it holds that $v_i(M_i^t) > (1/\alpha) \cdot v_i(Z_i^t)$.
\end{restatable}

\begin{proof}[Proof of Lemma~\ref{lem:sub_ineq}]    
    Let $r$ be the matching time of $i$ (before $t$). Since $i$ is matched during the $r$-th iteration of the algorithm and $M_i^t \neq Z_i^t$, it must be the case that either $i^r=i$ and $M_i^r = J^r$ or $k^r = i$ and $M_i^r = S^r$. In the first case, by the if condition in line \ref{line:self_cond} of Algorithm~\ref{alg:new_subadditive_alg}, there is some $J \in \mathcal{B}^r$ with $v_i(Z_i^{r-1}) < \alpha \cdot v_i(J)$, and hence, 
    \begin{align*}
     v_i(M_i^t) &= v_i(J^r) && (\text{since } M_i^r = M_i^t)  \\
     &\geq v_i(J) && (\textstyle{\text{since }J^r \in \argmax_{J \in \mathcal{B}^r} v_i(J)}) \\
     &> (1/\alpha) \cdot v_i(Z_i^{r-1}) && (\text{by the choice of }J) \\
     &\geq (1/\alpha) \cdot v_i(Z_i^t)  && (\text{by Lemma \ref{lem:z_incr}})
     \intertext{and in the second case, it holds that}
        v_i(M_i^t) &= v_i(S^r) && (\text{since } M_i^r = M_i^t)\\  
        &> (1/\alpha) \cdot v_i(Z_i^{r-1}) && (\text{since }i \in \mathcal{K}_{S^t}^t) \\
        &\geq (1/\alpha) \cdot v_i(Z_i^t) && (\text{by Lemma \ref{lem:z_incr}})
    \end{align*}
    This proves the lemma.
\end{proof}

{Third,} let $\ell_t$ be the number of chains of type (ii) after $t$ iterations. Notice that because of the shrink operations in the algorithm, the product $\prod_{i\in \agents}v_i(X_i^t)$ might decrease after some iterations and so Theorem~\ref{thm:subadditive-partial-better} does not follow directly from Claim~\ref{cla:basic}(ii) analogously to the proof of Theorem~\ref{thm:additive-partial}. However, we consider a different potential function, namely $((\alpha+1)/\alpha)^{\ell_t}  \cdot \prod_{i \in \agents} v_i(X_i^t)$. {Observe that} this value does not decrease {because each time some $X_i$ is shrunk, the algorithm creates a new chain of type (ii) which compensates for the loss in the potential function by Lemma~\ref{lem:sub_ineq}.}

\begin{restatable}
{lemma}{lemxilem}
    It holds that $$\Big(\frac{\alpha+1}{\alpha}\Big)^{\ell_t}   \prod_{i \in \agents} v_i(X_i^t) \geq \Big(\frac{\alpha+1}{\alpha}\Big)^{\ell_{t-1}}   \prod_{i \in \agents} v_i(X_i^{t-1})$$
    {for any $1 \leq t \leq T+1$.}
    \label{lem:xi_lem}
\end{restatable}
\begin{proof}
    Let $\xi = (\alpha+1)/\alpha$.
    Observe that $v_{j^t}(Z_{j^t}^{t-1}) \geq (1/(\alpha+1)) \cdot v_{j^t}(X_{j^t}^{t-1})$ and $v_{k^t}(Z_{k^t}^{t-1}) \geq (1/(\alpha+1)) \cdot v_{k^t}(X_{k^t}^{t-1})$ by Claim~\ref{cla:basic}(ii). {Consider the $t$-th iteration of the algorithm. Note that the number of chains of type (ii) is precisely the number of the blue bundles $M_i^t$. Moreover, the only time a blue bundle $M_i^t$ is removed or changed is in Case 4, where $M_{j^t}^t = \bot$ is no longer a blue bundle. But then, it is also the case that $M_{i^t}^t$ is a new blue bundle, and so the number of blue bundles does not change. Hence, in the remaining cases, it is enough to consider the bundles that become blue during the $t$-th iteration.}

    In Case 3, {the bundle $J^t$ becomes blue which means that there is one new chain of type (ii), i.e., $\ell_t = \ell_{t-1} + 1$, and} it holds that
    \begin{align*}
     \xi^{\ell_{t}} \cdot  v_{j^t}(X_{j^t}^{t})  
     &= \xi  \cdot \xi^{\ell_{t}-1}  \cdot v_{j^t}(Z_{j^t}^{t-1})  \\
     &\geq \xi  \cdot \xi^{\ell_{t-1}}  \cdot  \Big( \frac{1}{\alpha+1} \Big)  \cdot v_{j^t}(X_{j^t}^{t-1}) \\
     &=  \Big( \frac{1}{\alpha} \Big) \cdot  \xi^{\ell_{t-1}}  \cdot v_{j^t}(X_{j^t}^{t-1}) \\
     &> \xi^{\ell_{t-1}}  \cdot v_{j^t}(X_{j^t}^{t-1}).
     \intertext{In Case 2.1, {the bundle $J^t$ becomes blue which means that there is one new chain of type (ii), i.e., $\ell_t = \ell_{t-1} + 1$, and} since $v_{j^t}(g) \geq (\alpha/(\alpha+1)) \cdot v_{j^t}(X_{j^t}^{t-1})$ by the if condition in line \ref{line:g_cond}, it holds that}
     \xi^{\ell_{t}}  \cdot v_{j^t}(X_{j^t}^{t})  
     &= \xi  \cdot \xi^{\ell_{t}-1} \cdot  v_{j^t}(g^t) \\
     &\geq \xi  \cdot \xi^{\ell_{t-1}}  \cdot \Big(\frac{\alpha}{\alpha+1}\Big)  \cdot v_{j^t}(X_{j^t}^{t-1}) \\
     &= \xi^{\ell_{t-1}}  \cdot v_{j^t}(X_{j^t}^{t-1}).
     \intertext{In Case 2.3, {the bundle $J^t$ becomes blue which means that there is one new chain of type (ii), i.e., $\ell_t = \ell_{t-1} + 1$, and} since $v_{j^t}({R}^t) \geq (\alpha/(\alpha+1))\cdot v_{j^t}(X_{j^t}^{t-1})$ by the if condition in line \ref{line:l_perp_cond}, it holds that}
     \xi^{\ell_{t}} \cdot v_{j^t}(X_{j^t}^{t}) 
     &= \xi \cdot \xi^{\ell_{t}-1}  \cdot v_{j^t}({R}^t)\\
     &\geq \xi \cdot \xi^{\ell_{t-1}} \cdot \Big(\frac{\alpha}{\alpha+1}\Big) \cdot v_{j^t}(X_{j^t}^{t-1}) \\
     &= \xi^{\ell_{t-1}}  \cdot v_{j^t}(X_{j^t}^{t-1}).
     \intertext{In Case 2.4, {the bundle $S^t$ becomes blue which means that there is one new chain of type (ii), i.e., $\ell_t = \ell_{t-1} + 1$, and} since $v_{j^t}(J^t) \geq \alpha \cdot v_{j^t}(X_{j^t}^{t-1})$ by the if condition in line \ref{line:l_cond}, it holds that}
     \xi^{\ell_{t}} \cdot v_{j^t}(X_{j^t}^{t})  \cdot v_{k^t}(X_{k^t}^{t}) 
     &= \xi  \cdot \xi^{\ell_{t}-1}  \cdot v_{j^t}(J^t)  \cdot v_{k^t}(Z_{k^t}^{t-1}) \\
     &\geq \xi  \cdot \xi^{\ell_{t-1}} \cdot \alpha   \cdot v_{j^t}(X_{j^t}^{t-1}) \cdot \Big(\frac{1}{\alpha+1}\Big) \cdot v_{k^t}(X_{k^t}^{t-1}) \\
     &= \xi^{\ell_{t-1}}  \cdot v_{j^t}(X_{j^t}^{t-1}) \cdot v_{k^t}(X_{k^t}^{t-1}).
     \intertext{In Case 2.5, {the bundle $J^t$ becomes blue which means that there is one new chain of type (ii), i.e., $\ell_t = \ell_{t-1} + 1$, and} since $v_{j^t}(S^t) \geq (\alpha/(\alpha+1)) \cdot v_{j^t}(X_{j^t}^{t-1})$ by the if condition in line \ref{line:s_cond}, it holds that}
     \xi^{\ell_{t}} \cdot v_{j^t}(X_{j^t}^{t}) 
     &= \xi  \cdot \xi^{\ell_{t}-1}  \cdot v_{j^t}(S^t) \\
     &\geq \xi \cdot \xi^{\ell_{t-1}}  \cdot \Big(\frac{\alpha}{\alpha+1}\Big)  \cdot v_{j^t}(X_{j^t}^{t-1}) \\
     &= \xi^{\ell_{t-1}}  \cdot v_{j^t}(X_{j^t}^{t-1}).
     \intertext{In Case 2.6, {the bundles $J^t$ and $S^t$ become blue which means that there are two new chains of type (ii), i.e., $\ell_t = \ell_{t-1} + 2$, and} by subadditivity of $v_{j^t}$ and the if conditions in lines \ref{line:l_cond} and \ref{line:s_cond}, it holds that}
     v_{j^t}(R^t \setminus S^t) 
     &= v_{j^t}(X_{j^t}^{t-1} \setminus (S^t \cup J^t)) \\
     &\geq v_{j^t}(X_{j^t}^{t-1}) - v_{j^t}(S^t) - v_{j^t}(J^t)\\
     &\geq \Big(1-\frac{\alpha}{\alpha+1}-\alpha\Big)\cdot v_{j^t}(X_{j^t}^{t-1})  \\
     &\geq \Big(\frac{\alpha^2}{\alpha+1}\Big) \cdot v_{j^t}(X_{j^t}^{t-1}) 
     \intertext{where the last inequality follows by the assumption that $\alpha \leq 1/2$. Thus,}
     \xi^{\ell_{t}} \cdot v_{j^t}(X_{j^t}^{t}) \cdot v_{k^t}(X_{k^t}^{t}) 
     &= \xi^2 \cdot \xi^{\ell_{t}-2} \cdot v_{j^t}(R^t \setminus S^t) \cdot v_{k^t}(Z_{k^t}^{t-1}) \\
     &\geq \xi^2 \cdot \xi^{\ell_{t-1}} \cdot \Big(\frac{\alpha^2}{\alpha+1}\Big)  \cdot v_{j^t}(X_{j^t}^{t-1}) \cdot \Big(\frac{1}{\alpha+1}\Big)  \cdot v_{k^t}(X_{k^t}^{t-1})  \\
     &= \xi^{\ell_{t-1}} \cdot v_{j^t}(X_{j^t}^{t-1}) \cdot v_{k^t}(X_{k^t}^{t-1}).
    \end{align*}
    The result follows since in the remaining cases $X_i^{t+1} = X_i^t$ for all $i$.
\end{proof}

{We are now ready to present the proof of Lemma~\ref{lem:subadd_mnw}.}

\begin{proof}[Proof of Lemma~\ref{lem:subadd_mnw}]
Consider the state of the algorithm after $T$ iterations. Observe that there are no chains of type (i) since all the agents are matched.
    {F}or any chain $i_1, \ldots, i_k$ of type (ii), it must be that $$v_{i_1}(M_{i_1}^T) > (1/\alpha) \cdot v_{i_1}(Z_{i_1}^T) = (1/\alpha) \cdot v_{i_1}(X_{i_1}^T)$$ by Lemma~\ref{lem:sub_ineq}, while for $1 < j \leq \ell$, it holds that $$v_{i_j}(M_{i_j}^T) > (1/\alpha) \cdot v_{i_j}(Z_{i_j}^T) \geq v_{i_j}(Z_{i_j}^T) \geq (1/(\alpha+1)) \cdot v_{i_j}(X_{i_j}^T)$$ by Lemma~\ref{lem:sub_ineq} and Claim~\ref{cla:basic}(ii). Moreover, for any cycle $i_1$ of type (iii), it must be that $$v_{i_1}(M_{i_1}^T) 
     =  v_{i_1}(Z_{i_1}^T) \geq (1/(\alpha+1)) \cdot v_{i_1}(X_{i_1}^T)$$ by Claim~\ref{cla:basic}(ii), and for any cycle $i_1, \ldots, i_\ell$ of type (iv), $M_{i_j}^T \neq Z_{i_j}^T$ for all $1 \leq j \leq \ell$, and so $$v_{i_j}(M_{i_j}^T) \geq (1/\alpha) \cdot v_{i_j}(Z_{i_{j}}^T) \geq (1/(\alpha+1)) \cdot v_{i_j}(Z_{i_{j}}^T)$$ by Lemma \ref{lem:sub_ineq}.
    It follows that
    \begin{align*}
       \prod_{i \in \agents} {v_i(M_i)} 
       &\geq \Big(\frac{1}{\alpha+1}\Big)^{n-\ell_T}  \Big(\frac{1}{\alpha}\Big) ^{\ell_T}  \prod_{i \in \agents} v_i(X_i^T) && (\text{by the above}) \\
       &=  \Big(\frac{1}{\alpha+1}\Big)^{n}  \Big(\frac{\alpha+1}{\alpha}\Big)^{\ell_T}  \prod_{i \in \agents} {v_i(X_i^T)} &&(\text{by simple algebra}) \\
       &\geq \Big(\frac{1}{\alpha+1}\Big)^{n}  \Big(\frac{\alpha+1}{\alpha}\Big)^{\ell_0}  \prod_{i \in \agents }v_i(X_i^0)  && (\text{by Lemma~\ref{lem:xi_lem}}) \\
       &= \Big(\frac{1}{\alpha+1}\Big)^{n}  \prod_{i \in \agents} v_i(X_i) && (\text{since } \ell_0 = 0)
    \end{align*}
    which proves the result.
\end{proof}

The main theorem follows immediately from {the two key lemmas}.

\begin{proof}[Proof of Theorem~\ref{thm:subadditive-partial-better}]
{Assume that the algorithm is given the Nash welfare maximizing allocation}
{as input. Then, }
the allocation $M^T$ returned by the algorithm is $\alpha$-EFX by Lemma~\ref{lem:feas} and $\frac{1}{\alpha+1}$-MNW by Lemma~\ref{lem:subadd_mnw}.
\end{proof}

\subsection{Complete Allocations}\label{sec:subadditive-complete}

The main result of this section is the following theorem.

\thmsubadditivecomplete*

To prove Theorem~\ref{thm:complete_subadditive} using Theorem~\ref{thm:subadditive-partial-better}, we use the swapping with singletons technique from \cite[Section 4.1]{GHLVV22}. 
{Using this technique, it can be shown} that any partial allocation can be turned into a $1$-separated allocation while weakly improving the values for all agents. 
This is captured in the following lemma.

\begin{restatable}
{lemma}{lemsubadditivesep}
    For any partial $\alpha$-EFX allocation $(Z_1, \ldots, Z_n)$, we can find another partial $\alpha$-EFX allocation $(Z_1', \ldots, Z_n')$ satisfying $1$-separation with $v_i(Z_i') \geq v_i(Z_i)$ for all agents $i$.
\label{lem:subadditive_sep}
\end{restatable}

{The proof of Lemma~\ref{lem:subadditive_sep} appears in Appendix~\ref{sec:sep_proofs}.
We now present the proof of Theorem~\ref{thm:complete_subadditive}.}

\begin{proof}[Proof of Theorem~\ref{thm:complete_subadditive}]
By Theorem~\ref{thm:subadditive-partial-better}, there exists a partial allocation $(Z_1, \ldots, Z_n)$ that is $\alpha$-EFX and $\frac{1}{\alpha+1}$-MNW. 
Then, by Lemma~\ref{lem:subadditive_sep}, the above existence implies the existence of a partial allocation $(Z_1', \ldots, Z_n')$ that is $\alpha$-EFX, $1$-separated, and $\frac{1}{1+\alpha}$-MNW, where the NW guarantee is implied by $\nw(Z') \geq \nw(Z)$. 
The result now follows by Lemma~\ref{lem:make_complete}.
\end{proof}

\section{Impossibility Results}\label{sec:impossibility}

In {the final} section, we provide a couple of impossibility results. 

First, we give an impossibility result for additive valuations. {The proof uses} an instance that generalizes the instance used {by} \citet*{CGH19}.

\thmlowerbound*

\begin{proof}
{{Consider an instance with} $m=2n-1$ items. Denote the first $n-1$ items {by} $a_1, \ldots, a_{n-1}$ and the next $n$ items {by} $b_1, \ldots, b_n$. Define the valuation functions in the following way:
\begin{align*}
v_i(a_j) &= 1/\alpha + \varepsilon && \textrm{for all } j \\
v_i(b_i) &= 1 \\
v_i(b_j) &= 0 && \textrm{for } j \neq i
\end{align*}
{for every agent $i$.}

Consider the allocation $X$ defined by $X_i = \{a_i, b_i\}$ for $i < n$ and $X_n = \{ b_n \}$. The Nash welfare of $X$ is $(1+1/\alpha+\varepsilon)^{(n-1)/n}$. Note that $X$ is not $\alpha$-EFX because agent $n$ envies agent $1$ even after removing $b_1$ from $X_1$, i.e., it holds that $\alpha \cdot v_n(X_1 - b_1) > v_n(X_n)$.

Now, consider any partial allocation $Z$ that is $\alpha$-EFX. Since there are only $n-1$ items denoted by $a$, there must be some agent $i$ who only gets items denoted by $b$ in $Z_i$. Consider any other agent $j$ and suppose that $Z_j$ includes both the item $b_j$ and some other item $a_k$. Note that if that is the case, then
\begin{align*}
v_i(Z_i) \leq 
v_i(\{b_1, \ldots, b_n\}) = 1 < 1 + \alpha \cdot \varepsilon = \alpha \cdot v_i(a_k) \leq \alpha \cdot v_i(\{a_k, b_j\} - b_j)    \leq  \alpha \cdot v_i(Z_j-b_j)
\end{align*}
which contradicts the assumption that $Z$ is $\alpha$-EFX. Hence, no $Z_j$ includes both $b_j$ and some $a_k$.
Therefore, $
    \nw(Z) = v_i(Z_i)^{1/n} \cdot \prod_{a \in \agents-i} v_a(Z_a)^{1/n} \leq \prod_{i \in \agents-i} (1/\alpha + \varepsilon)^{1/n} = (1/\alpha+\varepsilon)^{(n-1)/n}
$. The result follows by taking the limit as $\varepsilon \to 0$ and $n \to \infty$.}
\end{proof}

Now we give an impossibility result for general monotone valuations. 

\begin{restatable}[Monotone]{theorem}{thmmonotone}
\label{thm:monotone-upper}
For every $\alpha > 0$ and $\beta > 0$, there exists an instance with general monotone valuations that admits no allocation (even partial) that is $\alpha$-EFX and $\beta$-MNW.
\end{restatable}

\begin{proof}{ Let $N > 1$.
    Consider an instance with two identical agents and five identical items. Denote by $v(k)$ the value for either of the agents of getting $k$ items. Let $v(0) = 0$, $v(1) = 1$, $v(2) = 1$, $v(3) = \sqrt{N}$, $v(4)=N$, and $v(5) = N$. One of the MNW allocations is $X$ given by $X_1 = \{1,2,3,4\}$ and $X_2 = \{5\}$ with $\nw(X) = \sqrt{N \cdot 1} = \sqrt{N}$.

    Let $Z$ be any partial allocation that is $2/\sqrt{N}$-EFX.
    Note that there must be an agent who gets at most $2$ items in $Z$. Without loss of generality, assume that $|Z_1| \leq 2$. Then, $v_1(Z_1) \leq 1$. Suppose that $|Z_2| \geq 4$ and let $g \in Z_2$. Then, $(2/\sqrt{N} ) \cdot v_1(Z_2 - g) \geq (2/\sqrt{N}) \cdot \sqrt{N} > v_1(Z_1)$ which contradicts the assumption that $Z$ is $2/\sqrt{N}$-EFX. Therefore, $|Z_2| \leq 3$, and so $NW(Z) \leq \sqrt{1 \cdot \sqrt{N}}$ which implies $\nw(Z) / \nw(X) \leq 1/\sqrt[4]{N}$. The result follows by taking $N \to \infty$.}
\end{proof}

\ifdefined\compileAAAI\else

\bibliographystyle{unsrtnat}
\bibliography{bib}
\appendix

\section{Computational Remarks}\label{sec:comp}

In this section, we extend the algorithms used to prove our positive results (Algorithms~\ref{alg:alloc} and~\ref{alg:new_subadditive_alg}) to return allocations with the desired guarantees even when the input is not necessarily the optimal MNW allocation. Specifically, we prove the following theorems.

\begin{theorem}
     Assume that the valuations are additive. Then, there exists a polynomial time algorithm that, given a $\beta$-MNW allocation $X$ as input, finds a partial allocation that is $\alpha$-EFX and $\frac{\beta}{\alpha+1}$-MNW, for every $0 \leq \alpha \leq 1$ and $\beta \in (0,1)$.\label{thm:additive-partial-poly}
\end{theorem}
\begin{theorem}
    Assume that the valuations are subadditive. Then, there exists a polynomial time algorithm that, given a $\beta$-MNW allocation $X$ as input, finds a complete allocation that is $\alpha$-EFX and $\frac{\beta}{\alpha+1}$-MNW, for every $0 \leq \alpha \leq 1/2$ and $\beta \in (0,1)$.\label{thm:subadditive-complete-poly}
\end{theorem}

For subadditive valuations, Algorithm~\ref{alg:new_subadditive_alg} already satisfies the conditions of Theorem~\ref{thm:subadditive-complete-poly}.

\begin{proof}[Proof of Theorem~\ref{thm:subadditive-complete-poly}]
Given any $\beta$-MNW allocation $X$ as input, Algorithm~\ref{alg:new_subadditive_alg} produces an allocation $(M_1, \ldots, M_n)$ that is $\alpha$-EFX (by Lemma~\ref{lem:feas}) and $\frac{\beta}{\alpha+1}$-MNW (by Lemma~\ref{lem:subadd_mnw}).   
\end{proof}

For additive valuations, we first modify Algorithm~\ref{alg:alloc} in a similar way to \cite[Section 5]{CGH19}, see Algorithm~\ref{alg:add_poly}. Suppose that Algorithm~\ref{alg:add_poly} is given as input an allocation $X$ that is $\beta$-MNW for some $\beta \in (0,1)$. Observe that the statements of Claims~\ref{cla:add_basic} and \ref{cla:add_ineq} also hold for Algorithm~\ref{alg:add_poly} in this setting.

\begin{algorithm}
\caption{{Additive valuations.}}
\label{alg:add_poly}
\begin{flushleft}
\hspace*{\algorithmicindent} \textbf{Input} $(X_1, \ldots, X_n)$ is a complete $\rho$-MNW allocation. \\
\hspace*{\algorithmicindent} \textbf{Output} either $(M_1, \ldots, M_n)$ is a partial $\alpha$-EFX and $\frac{\rho}{\alpha+1}$-MNW allocation, \\
\hspace*{\algorithmicindent} 
\;\;\;\;\;\;\;\;\;\;\;\;\; or $(\widehat{X}_1, \ldots, \widehat{X}_n)$ is a complete $(1 + \frac{1}{\alpha+1} \cdot \frac{1}{n-1})^{1/n}\rho$-MNW allocation.
\end{flushleft}
\vspace{-0.15in}
\begin{algorithmic}
[1]\State match $i$ to $J$ means $M_i \gets J$
\State unmatch $Z_i$ means $M_u \gets \bot$ if there is $u$ matched to $Z_i$
\Procedure{alg}{} 
\State $Z\gets(X_1, \ldots, X_n)$
\State $M\gets (\bot, \ldots, \bot)$
\While{there is an agent $i$ with $M_i = \bot$}
\State $i^\star \gets $ any agent with $M_i = \bot$
\If{{($Z_{i^\star}=X_{i^\star}$ and }$v_{i^\star}(Z_{i^\star}) \geq \alpha \cdot v_{i^\star}(Z_j -g)$ for all $j$ and $g \in Z_j$) {or \\ \;\;\;\;\;\;\;\;\;\;\;\;\;\, ($Z_{i^\star} \neq X_{i^\star}$ and $v_{i^\star}(Z_{i^\star}) \geq v_{i^\star}(Z_j -g)$ for all $j$ and $g \in Z_j$)}}
    \State unmatch $Z_{i^\star}$\label{line:a_1}
    \State match $i^\star$ to $Z_{i^\star}$\label{line:a_2}
\Else  
    \State $j^\star,g \gets $ any $j^\star$ and $g \in Z_{j^\star}$ that maximize $v_{i^\star}(Z_{j^\star}-g)$
    \If{$Z_{j^\star}$ is unmatched}\label{line:if_un}
        \State match $i^\star$ to $Z_{j^\star}$\label{line:b}
    \Else
        \State unmatch $Z_{j^\star}$\label{line:else_1}
        \State match $i^\star$ to $Z_{j^\star}$
        \State $j_1, \ldots, j_\ell \gets$ the improving sequence (Definition~\ref{def:improving_seq})\label{line:else_2}
        \If{$j$ ends with condition (ii)}\label{line:if_ii}
            \State change $Z_{j^\star}$ to $Z_{j^\star} - g$\label{line:remove_1}
            \State match $i^\star$ to $Z_{j^\star}$\label{line:remove_2}
            \If{$v_{j^\star}(Z_{j^\star}) < (\frac{1}{\alpha+1})^{n/(n-1)} \cdot v_{j^\star}(X_{j^\star})$}\label{line:if_x_small}
                \State $\widehat{X} \gets (X_1, \ldots, X_n)$
                \State $\widehat{X}_{j^\star} \gets X_{j^\star} \setminus Z_{j^\star}$\label{line:hat_def}
                \State $\widehat{X}_{j_s} \gets (X_{j_s} \setminus Z_{j_s}) \cup Z_{j_{s-1}}$ for $2 \leq s < \ell$
                \State $\widehat{X}_{j_\ell} \gets X_{j_\ell} \cup Z_{j_{\ell-1}}$            
                \State \textbf{return} $(\widehat{X}_1, \ldots, \widehat{X}_n)$\label{line:return_x}
            \EndIf
        \EndIf
    \EndIf
\EndIf
\EndWhile
\State \textbf{return} $(M_1, \ldots, M_n)$\label{line:return_m}
\EndProcedure
\end{algorithmic}
\end{algorithm}

Recall that Lemma~\ref{lem:additive_bound} guarantees that if the input allocation $X$ is MNW, then the outcome of Algorithm~\ref{alg:alloc} satisfies $v_i(Z_i) \geq (1/(\alpha+1)) \cdot v_i(X_i)$ for all agents $i$. The proof of the lemma is by contradiction: if $v_i(Z_i)$ becomes less than $(1/(\alpha+1)) \cdot v_i(X_i)$ for some agent $i$, then there is a way to construct a new allocation $\widehat{X}$ with higher Nash welfare than $X$, contradicting optimality of $X$. Note that the same argument applies to Algorithm~\ref{alg:add_poly}, and so the statement of Lemma~\ref{lem:additive_bound} holds for the outcome of Algorithm~\ref{alg:add_poly} as well, assuming that $X$ is MNW.

In the general case, if $X$ is not necessarily a MNW allocation, Algorithm~\ref{alg:add_poly} essentially simulates Algorithm~\ref{alg:alloc}\footnote{The only difference is that Algorithm~\ref{alg:add_poly} does not remove $g$ from $Z_{j^\star}$ if either $Z_{j^\star}$ is unmatched or the improving sequence ends with condition (i), while Algorithm~\ref{alg:alloc} removes $g$ from $Z_{j^\star}$ in both cases.} as long as it holds that $v_i(Z_i) \geq (1/(\alpha+1))^{n/(n-1)} \cdot v_i(X_i)$ for all agents $i$. Note that this is a slightly weaker condition that $v_i(Z_i) \geq (1/(\alpha+1)) \cdot v_i(X_i)$ since $1/(\alpha+1) < 1$ and $n/(n-1) > 1$. If this condition is not violated throughout the execution of the algorithm, then the algorithm returns an allocation $(M_1, \ldots, M_n)$ in line~\ref{line:return_m}. In this case, the allocation $M$ is $\alpha$-EFX (by Claims~\ref{cla:add_basic} and \ref{cla:add_ineq}) and we show below that $M$ is $\frac{\beta}{\alpha+1}$-MNW. 

\begin{lemma}\label{lem:poly_lem_1}
    If Algorithm~\ref{alg:add_poly} returns an allocation $M$ in line~\ref{line:return_m}, then $M$ is $\frac{\beta}{\alpha+1}$-MNW.
\end{lemma}
\begin{proof}
    Note that at end of the run of the algorithm, it holds that $v_a(Z_a) \geq (\frac{1}{\alpha+1})^{n/(n-1)} \cdot v_a(X_a)$ for all agents $a$ by the if condition in line~\ref{line:if_x_small}.
    
    Moreover, there is at least one agent $i$ with $Z_i = X_i$. Indeed, suppose that this is not the case, and consider the iteration during which the algorithm executes line~\ref{line:remove_1} and it holds that $Z_{j} \subsetneq X_j$ for all $j \neq j^\star$. Then, by Claim~\ref{cla:add_ineq}(iii), at the beginning of that iteration, all bundles $Z_j$ with $j \neq j^\star$ must be matched to some agents. Moreover, by the if condition in line~\ref{line:if_un}, the bundle $Z_{j^\star}$ must also be matched to some agent. Therefore, all bundles are matched at the beginning of that iteration, which contradicts the assumption that $i^\star$ is an unmatched agent.
    
    It follows that
    \begin{align*}
        \nw(M) 
        &\geq \prod_{a \in \agents} v_a(Z_a)^{1/n} && (\text{since } v_a(M_a) \geq v_a(Z_a)) \\*
        &=  v_i(X_i)^{1/n}  \prod_{a \in \agents - i} v_a(Z_a)^{1/n} && (\text{since } Z_i = X_i) \\*
        &\geq v_i(X_i)^{1/n}   \prod_{a \in \agents - i} \big(\frac{1}{\alpha+1}\big)^{1/(n-1)}  v_a(X_a)^{1/n} && (\text{by the above}) \\*
        &= \frac{1}{\alpha+1} \cdot \nw(X) 
    \end{align*}
    which gives the result.
\end{proof}

Otherwise, if at some point $v_{j^\star}(Z_{j^\star})$ becomes less than $(1/(\alpha+1))^{n/(n-1)} \cdot v_{j^\star}(X_{j^\star})$ due to the change operation in line~\ref{line:remove_1}, Algorithm~\ref{alg:add_poly} uses the construction from the proof of Lemma~\ref{lem:additive_bound} to return an allocation $\widehat{X}$ for which, as we show below, Nash welfare increases at least by a factor of $(1 + \frac{1}{\alpha+1} \cdot \frac{1}{n-1})^{1/n}$.

\begin{lemma}\label{lem:poly_lem_2}
    If Algorithm~\ref{alg:add_poly} returns an allocation $\widehat{X}$ in line~\ref{line:return_x}, then it holds that $\nw(\widehat{X}) > (1 + \frac{1}{\alpha+1} \cdot \frac{1}{n-1})^{1/n} \cdot \nw(X)$.
\end{lemma}
\begin{proof}
Note that
\begin{align*}
v_{j^\star}(\widehat{X}_{j^\star}) &=
v_{j^\star}(X_{j^\star} \setminus Z_{j^\star}) && (\text{by line~\ref{line:hat_def}}) \\
&= v_{j^\star}(X_{j^\star}) - v_{j^\star}(Z_{j^\star}) && (\text{by additivity of }v_{j^\star}) \\
&> \Big(1 - \Big(\frac{1}{\alpha+1}\Big)^{n/(n-1)} \Big)  \cdot v_{j^\star}(X_{j^\star}) && (\text{by the if condition in line~\ref{line:if_x_small}}) \\
&= \Big(1 - \frac{1}{\alpha+1} \cdot \Big(1 - \frac{\alpha}{\alpha+1}\Big)^{1/(n-1)} \Big)  \cdot v_{j^\star}(X_{j^\star}) && (\text{by simple algebra}) \\
&\geq \Big( 1 - \frac{1}{\alpha+1} \cdot \Big(1 - \frac{\alpha}{\alpha+1} \cdot \frac{1}{n-1}\Big) \Big) \cdot v_{j^\star}(X_{j^\star}) && (\text{by Bernoulli's inequality}) \\
&= \frac{\alpha}{\alpha+1} \cdot \Big(1 + \frac{1}{\alpha+1} \cdot \frac{1}{n-1} \Big) \cdot v_{j^\star}(X_{j^\star})
\end{align*}
    By Claim~\ref{cla:add_ineq}{(\ref{list:1_inequality})}, it holds that
\EQ{
v_{j_s}(\widehat{X}_{j_s}) = v_{j_s}((X_{j_s}\setminus Z_{j_s}) \cup Z_{j_{s-1}}) = v_{j_s}(X_{j_s}) + v_{j_s}(Z_{j_{s-1}})- v_{j_s}(Z_{j_s}) > v_{j_s}(X_{j_s}).} 
It follows from Claim~\ref{cla:add_ineq}{(\ref{list:a_inequality})} and the fact that $X_{j_\ell}=Z_{j_\ell}$ (by Claim~\ref{cla:add_ineq}{(\ref{list:unmatched})}) that
\EQ{
v_{j_\ell}(\widehat{X}_{j_\ell}) = v_{j_\ell}(X_{j_\ell} \cup Z_{j_{\ell-1}}) = v_{j_\ell}(X_{j_\ell}) + v_{j_\ell}(Z_{j_{\ell-1}}) > v_{j_\ell}(X_{j_\ell}) + \frac{1}{\alpha} \cdot v_{j_\ell}(Z_{j_\ell}) = \frac{\alpha+1}{\alpha} \cdot v_{j_\ell}(X_{j_\ell}).}
Finally, combining the inequalities above gives
\begin{align*}
\frac{\nw(\widehat{X})^n}{\nw(X)^n} 
= \frac{v_{j^\star}(\widehat{X}_{j^\star}) \cdot v_{i^\star}(\widehat{X}_{i^\star}) \cdot v_{j_{3}}(\widehat{X}_{j_{3}}) \cdots v_{j_{\ell-1}}(\widehat{X}_{j_{\ell-1}}) \cdot v_{j_\ell}(\widehat{X}_{j_\ell})}{v_{j^\star}(X_{j^\star}) \cdot v_{i^\star}(X_{i^\star}) \cdot v_{j_{3}}(X_{j_{3}} ) \cdots v_{j_{\ell-1}}(X_{j_{\ell-1}})  \cdot v_{j_\ell}(X_{j_\ell})}
> 1 + \frac{1}{\alpha+1} \cdot \frac{1}{n-1}
\end{align*}
and the result follows.
\end{proof}

Finally, we show that Algorithm~\ref{alg:add_poly} runs in polynomial time.

\begin{lemma}
    Algorithm~\ref{alg:add_poly} terminates after polynomially many iterations.\label{lem:first_runtime}
\end{lemma}
\begin{proof}
    For any $t$, consider the state of the algorithm after exactly $t$ iterations of the while loop{, and} let $\mathsf{Removed}_t = \bigcup_{i\in \agents} X_i \setminus Z_i$ and $\mathsf{Matched}_t = \{ i \in \agents \bar M_i \neq \bot \}$ and $\mathsf{SelfMatched}_t = \{ i \in \agents \bar M_i = Z_i\}$ and $\mathsf{Cycles}_t = \{ (i_1, \ldots, i_k)  \bar i_1, \ldots, i_k \in \agents, M_{i_j} = Z_{i_{j+1}} \text{ for } 1 \leq j \leq k \text{ where } i_{k+1} = i_1\}$. Consider the tuple $$\phi_t = (|\mathsf{Removed}_t|, |\mathsf{SelfMatched}_t|, |\mathsf{Matched}_t|, |\mathsf{Cycles}_t|).$$ 
    Observe that the tuple $\phi_t$ strictly increases lexicographically with each iteration of the algorithm.
    Indeed, if the algorithm executes 
    lines~\ref{line:a_1}-\ref{line:a_2}, then $|\mathsf{SelfMatched}_t|$ increases and $|\mathsf{Removed}_t|$ does not change. 
    If the algorithm executes line~\ref{line:b}, then $|\mathsf{Matched}_t|$ increases and $|\mathsf{Removed}_t|$ and $|\mathsf{SelfMatched}_t|$ do not change. 
    If the algorithm executes lines~\ref{line:remove_1}-\ref{line:remove_2}, then $|\mathsf{Removed}_t|$ increases. 
    Finally, if the algorithm executes lines~\ref{line:else_1}-\ref{line:else_2} during the $t$-th iteration and the if condition in line~\ref{line:if_ii} is not satisfied, then $j$ ends with condition (ii), i.e., $M_{j^\star} = Z_{j_\ell}$. 
    Note that this implies that $(j_1, \ldots, j_\ell) \in \mathsf{Cycles}_t$ and that $j^\star \notin \mathsf{SelfMatched}_{t-1}$. 
    Moreover, since $i^\star$ is unmatched at the beginning of the $t$-th iteration, it must be the case that $j^\star$ does not belong to any cycle in $\mathsf{Cycles}_{t-1}$.
    It follows that  $|\mathsf{Cycles}_t|$ increases and $|\mathsf{Removed}_t|$, $|\mathsf{SelfMatched}_t|$, and $|\mathsf{Matched}_t|$ do not change.
    Therefore, since there are at most polynomially many values that $\phi_t$ can take, the algorithm terminates after polynomially many iterations.  
\end{proof}

Given these results, to obtain {a partial allocation with the desired guarantees}, it is enough to run Algorithm~\ref{alg:add_poly} repeatedly until it returns an allocation $M$ in line~\ref{line:return_m}, see Algorithm~\ref{alg:poly_repeat}.

\begin{algorithm}
\caption{Additive valuations.}
\label{alg:poly_repeat}
\begin{flushleft}
\hspace*{\algorithmicindent} \textbf{Input} $(X_1, \ldots, X_n)$ is a complete $\rho$-MNW allocation. \\
\hspace*{\algorithmicindent} \textbf{Output} $(X_1, \ldots, X_n)$ is a partial $\alpha$-EFX and $\frac{\rho}{\alpha+1}$-MNW allocation. 
\end{flushleft}
\vspace{-0.15in}
\begin{algorithmic}
[1]\Procedure{alg}{} 
\While{$X$ is not $\alpha$-EFX}
    \State $X \gets $ output of Algorithm~\ref{alg:add_poly} with $X$ as input
\EndWhile
\State \textbf{return} $(X_1, \ldots, X_n)$
\EndProcedure
\end{algorithmic}
\end{algorithm}

Let us bound the running time of Algorithm~\ref{alg:poly_repeat}. 
\begin{lemma}
    Algorithm~\ref{alg:poly_repeat} terminates after polynomially many iterations.\label{lem:poly_add_runtime}
\end{lemma}
\begin{proof}Let $T+1$ be the total number of iterations of the while loop, and let $X^0, X^1, \ldots, X^{T+1}$ be the consecutive allocations throughout the execution of the algorithm. Note that each of $X^1, \ldots, X^T$ are the allocations returned by Algorithm~\ref{alg:add_poly} in line~\ref{line:return_x}.
For any $0 \leq t \leq T-n$, applying Lemma~\ref{lem:poly_lem_2} exactly $n$ times gives
$$\nw(X^{t+n}) - \nw(X^t) > \frac{1}{\alpha+1} \cdot \frac{1}{n-1} \cdot \nw(X^t) \geq \frac{1}{\alpha+1} \cdot \frac{1}{n-1} \cdot \nw(X^0).$$ 
Thus, if $T \geq n (n-1)(\alpha+1) (1/\beta)$, then it must hold that 
\begin{align*}
    \nw(X^{T}) &\geq \sum_{i=1}^{\lfloor T/n \rfloor} \nw(X^{T-(i-1)n})-\nw(X^{T-in})
    > (1/\beta) \cdot \nw(X^0) 
\end{align*}
which contradicts the assumption that $X^0$ is $\beta$-MNW. Hence, Algorithm~\ref{alg:poly_repeat} terminates after at most $n (n-1)(\alpha+1) (1/\beta)$ iterations.   
\end{proof}

We are now ready to prove Theorem~\ref{thm:additive-partial-poly}.

\begin{proof}[Proof of Theorem~\ref{thm:additive-partial-poly}]
    To obtain a partial allocation with the stated guarantees, we use Algorithm~\ref{alg:poly_repeat}.
    By Lemmas~\ref{lem:poly_lem_1} and~\ref{lem:poly_lem_2}, the output of the algorithm is $\frac{\beta}{\alpha+1}$-MNW, and by Lemmas~\ref{lem:first_runtime} and~\ref{lem:poly_add_runtime}, the algorithm terminates in polynomial times. The output of Algorithm~\ref{alg:poly_repeat} can be turned to a $1$-separated allocation using Algorithm~\ref{alg:sing_swaps}, and that allocation can be turned into a complete one with the desired guarantees using Algorithm~\ref{alg:envy_cycles}. Since both algorithms run in polynomial time (see Appendix~\ref{sec:sep_proofs}), the result follows.
\end{proof}

\section{Separated Allocations}\label{sec:sep_proofs}

\lemmakecomplete*
\begin{proof}
    The envy graph $(V,E)$ induced by a partial allocation $(Y_1, \ldots, Y_n)$ is given by the set of vertices $V = \agents$ and the set of edges $E = \{ (i,j) \in \agents^2 \bar v_i(Y_j) > v_i(Y_i)\}$, i.e., an edge from $i$ to $j$ means that agent $i$ prefers $j$'s bundle $Y_j$ to her own bundle $Y_i$. Notice that any cycle in the envy graph can be eliminated by reallocating the bundles along the cycle, where each agent in the cycle gets a new bundle that she values more than her previous bundle.

    \begin{algorithm}
\caption{Envy cycles procedure.}\label{alg:envy_cycles}
\begin{flushleft}
\hspace*{\algorithmicindent} \textbf{Input} $(Z_1, \ldots, Z_n)$ is a partial allocation. \\
\hspace*{\algorithmicindent} \;\;\;\;\;\;\; \;\; $U$ is the set of unallocated items. \\
\hspace*{\algorithmicindent} \textbf{Output} $(Y_1, \ldots, Y_n)$ is a complete allocation with $NW(Y) \geq NW(Z)$.
\end{flushleft}
\vspace{-0.1in}
\begin{algorithmic}
[1]\Procedure{alg}{}
\State $Y\gets(Z_1, \ldots, Z_n)$
\For{$x \in U$}
\State eliminate any cycles in the envy graph of $(Y_1, \ldots, Y_n)$\label{line:eliminate_cycles}
\State $i^\star \gets$  any agent who is not envied by anyone
\State $Y_{i^\star} \gets Y_{i^\star} + x$  
\EndFor
\State \textbf{return} $(Y_1, \ldots, Y_n)$
\EndProcedure
\end{algorithmic}
\end{algorithm}

    Consider Algorithm \ref{alg:envy_cycles}.
    We will show that the following invariant is preserved throughout the execution of the algorithm. For any pair of agents $i$ and $j$, agent $i$ is $\min(\alpha, 1/(1+\gamma))$-EFX towards agent $j$, i.e., {it holds that}
    $$v_i(Y_i) \geq \min(\alpha, 1/(1+\gamma)) \cdot v_i(Y_j-x) \;\;\; \text{for any item } x \in Y_j.$$
    
    The invariant holds after initializing $Y$ to $Z$ since, as $Z$ is $\alpha$-EFX, it must hold that 
    $$v_i(Y_i) \geq \alpha \cdot v_i(Y_j-x) \geq \min(\alpha, 1/(1+\gamma)) \cdot v_i(Y_j-x)$$
    for any agents $i$ {and} $j$ and any item $x \in Y_j$.

    Moreover, the invariant is preserved by eliminating cycles in the envy graph since the set of bundles remains the same and every agent ends up with a new bundle with at least as high a valuation as the previous bundle.

    It remains to show that allocating item $x$ to $Y_{i^\star}$ does not violate the invariant. Indeed, for any agent $j \neq i^\star$ and any $y \in Y_{i^\star}+x$, we have that 
    \begin{align*}
    v_j(Y_{i^\star}+x-y) 
    &\leq v_j(Y_{i^\star}+x) && (\text{by monotonicity of }v_j) \\
    &\leq v_j(Y_{i^\star}) + v_j(x) && \text{(by subadditivity of } v_j)\\
    &\leq v_j(Y_j) + v_j(x) && (\text{since } i^\star \text{ is not envied by } j)\\
    &\leq v_j(Y_j) + \gamma \cdot v_j(Y_j) && (\text{by } \gamma\text{-separation})\\
    &= (1+\gamma)\cdot v_j(Y_j) \\
    &\leq \max(1/\alpha,1+\gamma) \cdot v_j(Y_j)
    \intertext{and it follows that $j$ is $\min(\alpha, 1/(1+\gamma))$-EFX towards $i^\star$ given the new bundle.
 Furthermore, note that for any agent $j$ and any item $y \in Y_j$ it holds that }
    v_{i^\star}(Y_{i^\star}+x) &\geq v_{i^\star}(Y_{i^\star}) && \text{(by monotonicity of } v_{i^\star})\\
    &\geq \min(\alpha, 1/(1+\gamma))\cdot v_{i^\star}(Y_j-y) && (\text{by the invariant})
 \end{align*}
 and hence $i^\star$ is $\min(\alpha, 1/(1+\gamma))$-EFX towards $j$ given the new bundle. 
 The invariant is also preserved for all pairs not involving $i^\star$ as their allocations have not changed. {It follows that the final allocation is $\min(\alpha, 1/(1+\gamma))$-EFX. 
 
 It remains to show that the EF1 guarantee is also preserved throughout the execution of the envy cycles procedure. As before, eliminating cycles in the envy graph cannot violate this condition. Moreover, if the procedure allocates an item $x$ to $Y_{i^\star}$, then it holds for any agent $j$ that $ v_j((Y_{i^\star}+x) - x) = v_j(Y_{i^\star}) \leq v_j(Y_j)$ since $i^\star$ is not envied by $j$. Therefore, if the initial allocation $Z$ is EF1, then the final allocation is also EF1.
 }
\end{proof}

\lemsubadditivesep*
\begin{proof}
    Consider Algorithm~\ref{alg:sing_swaps}. 
    \begin{algorithm}
\caption{Swapping with singletons procedure.}\label{alg:sing_swaps}
\begin{flushleft}
\hspace*{\algorithmicindent} \textbf{Input} $(Z_1, \ldots, Z_n)$ is a partial allocation, \\ 
\hspace*{\algorithmicindent} \;\;\;\;\;\;\; \;\; $U$ is the set of unallocated items. \\
\hspace*{\algorithmicindent} \textbf{Output} $(Z_1, \ldots, Z_n)$ is a partial $1$-separated allocation.
\end{flushleft}
\vspace{-0.1in}
\begin{algorithmic}
[1]\Procedure{alg}{}
\While{there is an agent $i$ so that $v_i(Z_i) < v_i(x)$ for some $x \in U$}\label{line:ww_cond}
\State $U \gets (U \cup Z_i) - x$
\State $Z_i \gets \{x \}$
\EndWhile
\State \textbf{return} $(Z_1, \ldots, Z_n)$
\EndProcedure
\end{algorithmic}
\end{algorithm}

    Note that Algorithm~\ref{alg:sing_swaps} terminates in a polynomial number of steps. To show this, let $Z_i^t$ denote the variable $Z_i$ after exactly $t$ iterations of the algorithm. For any fixed agent $i$, the number of iterations where $Z_i^t \neq Z_i^{t-1}$ is at most $m+1$ because after the first change $Z_i$ becomes a singleton and with every following change the value $v_i(Z_i)$ strictly increases hence it cannot happen that $Z_i$ was changed to the same singleton more than once.
    
    Now, {we show by induction that the} allocation $(Z_1^t, \ldots, Z_n^t)$ is $\alpha$-EFX. This holds for $t=0$ by the assumption. Assume that $t > 0$ and let $i$ and $j$ be any two distinct agents. Denote by $i^t$ and $x^t$ the agent and the item chosen during the $t$-th iteration of the algorithm.
    
    If $i,j \neq i^t$, then for all $g \in Z_{j}^{t-1} = Z_j^t$, it holds that 
    \begin{align*}
     v_i(Z_i^t) 
     &= v_i(Z_i^{t-1})  && (\text{since } i^t \neq i) \\
     &\geq v_i(Z_j^{t-1} - g) && (\text{by the inductive assumption}) \\
     &= v_i(Z_j^{t}-g)  && (\text{since } j^t \neq i) 
     \intertext{If $i = i^t$, then for all $g \in Z_{j}^{t-1} = Z_j^t$, it holds that}
      v_i(Z_i^t) 
      &= v_i(x^t) && (\text{since } i^t = i) \\
      &\geq v_i(Z_i^{t-1}) && (\text{by the condition in line~\ref{line:ww_cond}}) \\ 
      &\geq v_i(Z_j^{t-1} - g)  && (\text{by the inductive assumption}) \\
      &= v_i(Z_j^{t}-g)  && (\text{since } i^t \neq j) 
      \intertext{If $j = i^t$,  for any $g \in Z_{j}^t = \{x^t\}$, it holds that}
         v_i(Z_i^t) 
         &\geq 0 &&  (\text{by monotonicity of } v_i) \\
         &= v_i(\{x^t\} - g) && (\text{since } g = x^t) \\
         &= v_i(Z_j^{t-1} - g) && (\text{since }j=i^t)
     \end{align*}
    It follows that the allocation returned by the algorithm is $\alpha$-EFX and it is $1$-separated by the condition for the while loop in line~\ref{line:ww_cond}.
\end{proof}

\section{Maximin Share Guarantees}\label{sec:mms}

In this section, we analyze the maximin share guarantees of the allocations given in Theorem~\ref{thm:complete_additive}. 

Let us start with the definition of the maximin shares. Let $\Pi_k(S)$ denote the set of all partitions of the set $S$ into $k$ subsets. We define the {\em maximin share} of agent $i$ with respect to $S$ to be $\mu_i(n, S) = \max_{X \in \Pi_n(S)} \min_{j \in [n]} v_i(X_j)$. We say that an allocation $X = (X_1, \ldots, X_n)$ is $\alpha$-PMMS ($\alpha$-pairwise maximin share) if for every two agents $i$ and $j$ it holds that $v_i(X_i) \geq \alpha \cdot \mu_i(2, X_i \cup X_j)$. We say that an allocation $X$ is $\alpha$-GMMS ($\alpha$-groupwise maximin share) if for every subset of the agents $I \subseteq [n]$ it holds that $v_i(X_i) \geq \mu_i(|I|, \bigcup_{j \in I} X_j)$ for all $i \in I$.

The main result of this section is the following.

\begin{theorem}\label{thm:mms}
    Every instance with additive valuations admits
    a complete allocation that is $\alpha$-EFX, EF1, $\frac{\alpha}{\alpha^2+1}$-GMMS, $(\varphi-1)$-PMMS, and $\frac{1}{\alpha+1}$-MNW, for every $0 \leq \alpha \leq \varphi - 1 \approx 0.618$.
\end{theorem}

The techniques we use to prove the theorem are very similar to those used by~\citet*{ABM18}. Let us first state two useful observations made by~\citet*{BL16}.

\begin{lemma}\label{lem:mu_tech}
    For any set of goods $S$ and any number of agents $k$, it holds that 
    \begin{enumerate}[(i)]
        \item $\mu_i(k, S) \leq \mu_i(k-1, S-g)$ for any good $g \in S$, and\label{list:remove_one_mu}
        \item $\mu_i(k, S) \leq v_i(S) / k$\label{list:average}.
    \end{enumerate}
\end{lemma}

We also use two results of~\citet*{CKMPSW16}.

\begin{lemma}\label{lem:mnw}
    Let $X$ be a MNW allocation. For any two agents $i$ and $j$, it holds that 
    \begin{enumerate}[(i)]
         \item $v_i(X_i) \geq (\varphi-1) \cdot \mu_i(2, X_i \cup X_j)$ where the items in $X_i$ are treated as divisible items, and\label{list:mnw_pmms}
        \item $v_i(X_j-g^\star) \leq v_i(X_i) \cdot \min(1,1/v_i(g^\star))$ where $g^\star \in \argmax_{g \in X_j} v_i(g)$.\label{list:mnw_remove_one}
    \end{enumerate}
   
\end{lemma}

Before stating the third technical lemma, let us introduce some further notation. Let $X$ be any initial MNW allocation. Let $Z$ be the output of Algorithm~\ref{alg:alloc} with $X$ as input, and let $Y$ be the output of the envy cycles procedure (Algorithm~\ref{alg:envy_cycles}) with $Z$ as input. Note that $Y$ is exactly the allocation that we construct in the proof of Theorem~\ref{thm:complete_additive}; {thus,} $Y$ satisfies $\alpha$-EFX, EF1, and $\frac{1}{\alpha+1}$-MNW. We can obtain the following bounds on the value of any non-singleton bundle.

\begin{lemma}\label{lem:mms_bounds}
    For any two agents $i$ and $j$ such that $|Y_j| \geq 2$, the following hold:
    \begin{enumerate}[(i)]
        \item if $Y_j = Z_k$ for some agent $k$ and $Z_i \subsetneq X_i$, then $v_i(Y_j) \leq 2 \cdot v_i(X_i)$,\label{list:zk_and_subsetneq}
        \item if $Y_j = Z_k$ for some agent $k$ and $Z_i = X_i$, then $v_i(Y_j) \leq (\alpha + 1/\alpha) \cdot v_i(X_i)$, and\label{list:zk_and_equals}
        \item if $Y_j \neq Z_k$ for any agent $k$, then $v_i(Y_j) \leq 2 \cdot v_i(X_i)$.\label{list:modified_during_ec}
    \end{enumerate}
\end{lemma}
\begin{proof}
    \underline{Property (i).} By Claim~\ref{cla:add_ineq}(\ref{list:efx}), we get $$v_i(Y_j) = v_i(Z_k) \leq 2 \cdot v_i(Z_i) \leq 2 \cdot v_i(Y_i)$$ since $|Z_k| \geq 2$ by the assumption.

    \underline{Property (ii).} Let $g^\star \in \argmax_{g \in X_k} v_i(g)$.
    If $v_i(g^\star) \leq v_i(X_i)$, then $$v_i(Y_j) = v_i(Z_k) \leq v_i(X_k) = v_i(g^\star) + v_i(X_k-g^\star) \leq 2 \cdot v_i(X_i) \leq 2 \cdot v_i(Y_i)$$ 
    by Lemma~\ref{lem:mnw}(\ref{list:mnw_remove_one}).
    Otherwise, again by Lemma~\ref{lem:mnw}(\ref{list:mnw_remove_one}),
    $$v_i(Y_j) = v_i(Z_k) \leq v_i(X_k) = v_i(X_k - g^\star) + v_i(g^\star) \leq (v_i(X_i) /v_i(g^\star)) + v_i(g^\star)  \leq (\alpha+1/\alpha) \cdot v_i(X_i)$$  
    since $v_i(g^\star) \leq (1/\alpha) \cdot v_i(X_i)$ by Claim~\ref{cla:add_ineq}(\ref{list:efx}).

    \underline{Property (iii).} It must be the case that the bundle $Y_j$ was modified during the envy cycles phase, i.e., a good $g$ was added to the bundle at some point when the bundle $Y_j-g$ was not envied by any agent. Hence, it must hold that $$v_i(Y_j) = v_i(Y_j-g) + v_i(g) \leq v_i(Y_i) + \alpha \cdot v_i(Y_i) \leq 2 \cdot v_i(Y_i)$$ by Lemma~\ref{lem:additive_better_sep}.
\end{proof}

We are now ready to prove the main theorem.

\begin{proof}[Proof of Theorem~\ref{thm:mms}.]
    We first prove that $Y$ satisfies $(\varphi-1)$-PMMS. Fix any agent $j$. If $|Y_j| = 1$, then $\mu_i(2, Y_i \cup Y_j) \leq \mu_i(1, Y_i) \leq v_i(Y_i)$ by Lemma~\ref{lem:mu_tech}(\ref{list:remove_one_mu}). If it holds that $Z_i = X_i$ and $Y_j = Z_k$, then we get $v_i(X_i) \geq (\varphi-1) \cdot \mu_i(2, X_i \cup X_k)$ by Lemma~\ref{lem:mnw}(\ref{list:mnw_pmms}) which implies that $v_i(Y_i) \geq (\varphi-1) \cdot \mu_i(2, Y_i  \cup Z_k)$ since $v_i(Y_i) \geq v_i(X_i)$ and $Z_k \subseteq X_k$. Otherwise, Lemma~\ref{lem:mms_bounds}(\ref{list:zk_and_subsetneq},\ref{list:modified_during_ec}) implies that $v_i(Y_j) \leq 2 \cdot v_i(Y_i)$ which implies that $\mu_i(2, Y_i \cup Y_j) \leq (3/2) \cdot v_i(Y_i)$ by Lemma~\ref{lem:mu_tech}(\ref{list:remove_one_mu}).

    Finally, let $J = \{j \neq i : |Y_j| = 1\}$ and observe that 
\begin{align*}
    \mu_i(n, [m]) &\leq \mu_i(n-|J|, [m]-\bigcup_{j \in J} X_j) && (\text{by Lemma~\ref{lem:mu_tech}(\ref{list:remove_one_mu})})\\
    &\leq   \sum_{j \in [n]\setminus J} \frac{v_i(X_j)}{n - |J|} && (\text{by Lemma~\ref{lem:mu_tech}(\ref{list:average})})\\
    &= \sum_{j \in [n]\setminus J} \frac{(\alpha + 1/\alpha) \cdot v_i(X_i)}{n - |J|} && (\text{by Lemma~\ref{lem:mms_bounds}})\\
   &= (\alpha + 1/\alpha) \cdot v_i(X_i)
\end{align*}
which gives the result that $Y$ is $\frac{\alpha}{\alpha^2+1}$-GMMS. 
\end{proof}

\section{Proofs from Section~\ref{sec:additive}}\label{sec:proofs-additive}

\claaddbasic*
\begin{proof}
{
\underline{Property (i).} Each $Z_i$ is initialized to $X_i$, and the only time where we modify $Z_i$ is to remove one of its elements. Hence, the invariant $Z_i\subseteq X_i$ is preserved.

\underline{Property (ii).} Each $M_i$ is initialized to $\bot$. The only time we change $M_i$ is to set it to $Z_{i^\star}$ or $Z_{j^\star}$. Moreover, observe that before modifying any $Z_j$, we unmatch any agent $i$ who could be matched to $M_i = Z_j$ by setting $M_i = \bot$. Hence, the invariant $M_i\in \{\bot\}\cup\{Z_j\,|\,j\in \agents\}$ is preserved.

\underline{Property (iii).} Observe that before matching an agent $i$ to a bundle $Z_j$ we unmatch any agent $u$ who could be matched to $M_u = Z_j$. Hence, agents are always matched to distinct bundles.
}
\end{proof}

\claaddineq*
\begin{proof}
{
\underline{Properties {(i) and (ii)}.} Given a matched agent $i$, it is enough to show that these properties hold at the end of an iteration where $i$ becomes matched because decreasing other bundles does not create any new envy. At the end of each iteration, consider agent $i^\star$ who just became matched. If $i^\star$ is matched to $Z_{i^\star}$, then the if condition in lines \ref{line:first_if} and \ref{line:if_add} implies the {desired property depending on whether $Z_i = X_i$ or $Z_i \neq X_i$}. Otherwise, $i^\star$ is matched to $Z_{j^\star}$, and then by maximality of $(j^\star, g)$ we have that $v_{i^\star}(Z_{j^\star}) \geq v_{i^\star}(Z_j-g)$ for any other $j\neq j^\star$ and $g\in Z_j$.

\underline{Properties {(iii) and (iv)}.} Once again, it is enough to show that the properties hold at the end of an iteration where $i$ becomes matched because decreasing the bundle $Z_i$ will not violate the property. {Consider the state of the algorithm at the beginning of an iteration during which $i^\star$ is matched to $Z_{j^\star}$ for $j^\star \neq i^\star$. If $X_{i^\star} = Z_{i^\star}$, then it must be the case that the condition in line~\ref{line:first_if} is not satisfied, and thus for some $j \in [n]$ and $g' \in Z_j$, it holds that 
\begin{align*}
    v_{i^\star}(Z_{i^\star}) 
    &< \alpha \cdot v_{i^\star}(Z_j - g') && (\text{by the condition in line~\ref{line:first_if}}) \\
    &\leq \alpha\cdot v_{i^\star}(Z_{j^\star}-g) && (\text{by the choice of $j^\star$ and $g$})
\end{align*}
which gives {property (iv)} since $Z_{j^\star}$ is changed to $Z_{j^\star}-g$ in line~\ref{line:change_z_add}. {Similarly, if $X_{i^\star} \neq Z_{i^\star}$, then the condition in line~\ref{line:if_add} is not satisfied, and thus for some $j \in [n]$ and $g' \in Z_j$, it holds that 
\begin{align*}
    v_{i^\star}(Z_{i^\star}) 
    &< v_{i^\star}(Z_j - g') && (\text{by the condition in line~\ref{line:if_add}}) \\
    &\leq  v_{i^\star}(Z_{j^\star}-g) && (\text{by the choice of $j^\star$ and $g$})
\end{align*}
which gives property (iii) since $Z_{j^\star}$ is changed to $Z_{j^\star}-g$ in line~\ref{line:change_z_add}.
}

\underline{Property (v).} Observe that whenever we {change} or unmatch some bundle $Z_{i}$, it is immediately matched to an agent. Hence, any touched bundle $Z_{i}\subsetneq X_i$ is always matched to some agent $u$.
}}
\end{proof}

\section{Proofs from Section~\ref{sec:subadditive}}\label{sec:proofs-subadditive}

\clabasic*
\begin{proof}
    \underline{Properties (i) and (ii).} The only time (along Algorithm~\ref{alg:new_subadditive_alg} and Algorithm~\ref{alg:split_bundle}) where we modify $Z_a$ and $X_a$ is not shrunk to $Z_a$ for some agent $a$ is in line \ref{line:mod_z} in Algorithm \ref{alg:split_bundle}, in which case property (i) follows by the if condition in line \ref{line:l_perp_cond} of Algorithm~\ref{alg:split_bundle}. 
    Property (ii) follows directly from property (i), by subadditivity of $v_i$.

    \underline{Property (iii).} Observe first that the agent $i^t$ is defined in every iteration of the algorithm, the agent $j^t$ is defined in every iteration that went into Case 2, 3, or 4, and the agent $k^t$ is defined in every iteration that went into Case 2.4, 2.5, or 2.6.

    We have that $k^t \neq i^t$ since $M_{i^t}^{t-1} = \bot$ and $M_{k^t}^{t-1} = Z_{k^t}^{t-1}$. 

    To show that $k^t \neq j^t$ in Cases 2.4, 2.5, and 2.6, observe that 
   \begin{align*}
     v_{j^t}(Z_{j^t}^{t-1}) 
     &\geq (1/(\alpha+1)) \cdot v_{j^t}(X_{j^t}^{t-1}) && (\text{by Claim~\ref{cla:basic}(ii)}) \\
     &\geq (1/\alpha) \cdot v_{j^t}(X_{j^t}^{t-1} \setminus Z_{j^t}^{t-1})   && (\text{by Claim~\ref{cla:basic}(i)}) 
     \intertext{and that}
    v_{j^t}(Z_{j^t}^{t-1}) 
    &\geq (1/(\alpha+1)) \cdot v_{j^t}(X_{j^t}^{t-1}) && (\text{by Claim~\ref{cla:basic}(ii)}) \\
    &> (1/\alpha) \cdot v_{j^t}(g)  && (\text{by line~\ref{line:g_cond} of Algorithm~\ref{alg:split_bundle}}) 
    \intertext{and so it follows that}
    v_{j^t}(Z_{j^t}^{t-1}) 
    &> (1/\alpha) \cdot v_{j^t}(X_{j^t}^{t-1} \setminus Z_{j^t}^{t-1}) + (1/\alpha) \cdot v_{j^t}(g) && (\text{by the above})\\
    &\geq (1/\alpha) \cdot v_{j^t}((X_{j^t}^{t-1} \setminus Z_{j^t}^{t-1}) + g)   && (\text{by subadditivity of }v_{j^t})  \\
    &\geq (1/\alpha) \cdot v_{j^t}(S^t) && (\text{since } S^t \subseteq R^t) \\
    &\geq \alpha \cdot v_{j^t}(S^t) && (\text{since } \alpha \leq 1)
    \end{align*}
    and therefore $j^t$ does not satisfy the condition to be included in $\mathcal{K}_{S^t}^{t}$ which implies $k^t \neq j^t$. 
    
    Finally, to show $i^t \neq j^t$ in Cases 2, 3, and 4, note that by the if condition in line~\ref{line:self_cond} of Algorithm~\ref{alg:new_subadditive_alg}, there is some $J' \in \mathcal{B}$ so that $v_{i^t}(J') > v_{i^t}(Z_{i^t}^{t-1})$.
    Suppose that $J^t = Z_{i^t}^{t-1}-g^t$. Then, it holds that 
    \begin{align*}
      v_{i^t}(Z_{i^t}^{t-1}-g^t) 
      &= v_{i^t}(J^t)  && (\text{by the assumption}) \\
      &\geq v_{i^t}(J') && (\text{since }\textstyle{J^t \in \argmax_{J \in \mathcal{B}^t} v_i(J)}) \\
      &> v_{i^t}(Z_{i^t}^{t-1})  && (\text{by the choice of }J')
      \intertext{
    which contradicts monotonicity of $v_{i^t}$.
    Now, suppose that $J^t = X_{i^t}^{t-1} \setminus Z_{i^t}^{t-1}$. Then,}
    v_{i^t}(J') 
    &> v_{i^t}(Z_{i^{t}}^{t-1}) && (\text{by the choice of }J') \\
    &\geq (1/(\alpha+1)) \cdot v_{i^t}(X_{i^t}^{t-1}) && (\text{by Claim~\ref{cla:basic}(ii)}) \\
    &\geq (1/\alpha) \cdot v_{i^t}(X_{i^t}^{t-1} \setminus Z_{i^t}^{t-1}) && (\text{by Claim~\ref{cla:basic}(i)}) \\
    &\geq v_{i^t}(J^t)    && (\text{since }\textstyle{J^t \in \argmax_{J \in \mathcal{B}^t} v_i(J)})
    \end{align*}
    which contradicts $J^t \in \argmax_{J \in \mathcal{B}^t} v_i(J)$. Finally, it cannot be the case that $J^t = M_{i^t}^{t-1} - g^t$ because $M_{i^t}^{t-1} = \bot$. It follows that $i^t \neq j^t$. This proves property (iii).

    \underline{Property (iv).} Notice that if $J^t$ is defined, then by the if condition in line~\ref{line:self_cond} of Algorithm~\ref{alg:new_subadditive_alg}, there is some $J' \in \mathcal{B}^t$ so that $v_{i^t}(Z_{i^t}^{t-1}) < \alpha \cdot v_{i^t}(J')$. Thus, it holds that
    \begin{align*}
    v_{i^t}(\emptyset) 
    &\leq v_{i^t}(Z_{i^t}^{t-1}) && (\text{by monotonicity of } v_{i^t}) \\
    &< \alpha \cdot v_{i^t}(J') && (\text{by the choice of }J') \\
    &\leq \alpha \cdot v_{i^t}(J^t) && (\text{since } \textstyle J^t \in \argmax_{J \in \mathcal{B}^{t}} v_{i^t}(J)) 
    \intertext{which implies that $J^t \neq \emptyset$. Moreover, $S^t$ is non-empty because }
     v_{k^t}(\emptyset)
     &\leq v_{k^t}(Z_{k^t}^t) && (\text{by monotonicity of }v_{k^t}) \\
     &< \alpha \cdot v_{k^t}(S^t) && (\text{since } k^t \in \mathcal{K}^t_{S^t})
    \end{align*}
    which proves property (iii).
\end{proof}

\clazconst*
\begin{proof}
    \underline{Properties (i) and (ii).} This follows by the design of the algorithm.
    
    \underline{Property (iii).} This follows by the definition of the matching time.
    
    \underline{Property (iv).} Assume that $M_i^t = Z_a^r$ for some agent $a$ and that $Z_a$ is modified during the $w$-th iteration for some $r < w \leq t$. Since the change operation is only executed in Case 2, it must be that the algorithm went into Case 2 in the $w$-th iteration of the algorithm with $j^w = a$. Note that $i$ was unmatched in line~\ref{line:second_unmatch} (Algorithm~\ref{alg:split_bundle}). Moreover, $i \neq i^w$ because $M_{i^w}^{w-1} = \bot$, and so $i$ was not re-matched in {line~\ref{line:match}} (Algorithm~\ref{alg:split_bundle}).
    Also, note that $i \neq k^w$. Indeed, if $i = k^w$, then $i = a$ by the definition of $\mathcal{K}_{S^w}^w$, but this implies that $k^w = i = a = j^w$ which contradicts Claim~\ref{cla:basic}(iii).
    Thus, $i$ was not re-matched in lines \ref{line:k1} (Algorithm~\ref{alg:split_bundle}) and \ref{line:k2} (Algorithm~\ref{alg:split_bundle}).
    Therefore, it holds that $M_i^w = \bot$ which contradicts the choice of $r$.
\end{proof}

\clachains*
\begin{proof}
    Consider the graph with vertices $V = \{Z_a^t \colon a \in \agents\} \cup \{M_a^t \colon {a \in \agents}, M_a^t\text{ is blue}\}$ and with edges $E = \{(Z_a^t, M_a^t) \colon {a \in \agents}, M_a^t \neq \bot\}$. By properties (i) and (ii) of Claim~\ref{cla:z_const}, the in-degree of every vertex in $V$ is at most $1$. Moreover, the out-degree of every blue bundle is $0$.

    It is enough to show that the in-degree of every vertex $Z_a^t$ for some agent $a$ with $Z_a^t \subsetneq X_a^t$ is exactly $1$. If that is the case, the result follows. 
    To prove the claim, consider the last iteration $r$ during which the algorithm changed $Z_a$, i.e., it holds that $Z_a^{r-1} \neq Z_a^r = \ldots Z_a^t$. Note that there is such an iteration by the assumption that $Z_a^t \subsetneq X_a^t$.  We show by induction that for all $r \leq h \leq t$, there is some agent $b$ so that $M_b^h = Z_a^h$.
    
    First, note that during the $r$-th iteration, it must be that the algorithm went into Case 2 because there are no change operations in the other cases. In Case 2, the algorithm changes only $Z_{j^r}$ and hence it must hold that $j^r = a$. If the algorithm went into either of Cases 2.1, 2.3, 2.4, 2.5, or 2.6, then it also performed a shrink operation on $X_a$ and hence $X_{a}^r = Z_a^r$ which contradicts the assumption that $Z_a^r = Z_a^t \subsetneq X_a^t \subseteq X_a^r$. In Case 2.2, it holds that $M_{i^r}^r = Z_a^r$. Therefore, the desired property holds for $h = r$.

    Now, assume that there is some agent $b$ with $M_b^{h-1} = Z_a^{h-1}$ {for some $h > r$}. Observe that if $b$ was unmatched in line~\ref{line:first_unmatch} of Algorithm~\ref{alg:new_subadditive_alg}, then it holds that $a = i^h$ and $M_a^h = Z_a^h$. Otherwise,
    \begin{itemize}
        \item it cannot be the case that $b$ was unmatched in line~\ref{line:third_unmatch} (Algorithm~\ref{alg:new_subadditive_alg}) because $M_b^{h-1}$ is not blue,
        \item it cannot be the case that $b$ was unmatched in line~\ref{line:second_unmatch} (Algorithm~\ref{alg:split_bundle}) because then $a = j^t$ and in Case 2 the algorithm always changes $Z_{j^t}$ which cannot happen by the assumption that $Z_a$ was not changed after the $r$-th iteration,
        \item it cannot be the case that $b = i^h$ since $M_{i^h}^{h-1} = \bot$, 
        \item it cannot be the case that $b = k^h$ and $b$ was re-matched in line~\ref{line:k1} (Algorithm~\ref{alg:split_bundle}) or in line~\ref{line:k2} (Algorithm~\ref{alg:split_bundle}) because then since $k^h \in \mathcal{K}_{S^h}^h$, it must hold that $M_b^{h-1} = Z_b^{h-1}$ and so $a=b$; however, this {also} means that $X_a$ was shrunk during the $h$-th iteration which contradicts the assumption that $Z_a^h = Z_a^t \subsetneq X_a^t \subseteq X_a^h$,
    \end{itemize}
    and so it holds that $M_b^h = Z_a^h$. 
    Therefore, there is some $b$ so that $M_b^t = Z_a^t${,} and the result follows.
\end{proof}

\clanewbund*
\begin{proof} 
    In Case 1, it holds that $\mathcal{B}^t = \mathcal{B}^{t-1}$. In Case 3, any $H \in \mathcal{B}^t \setminus \mathcal{B}^{t-1}$ is a subset of $X_{j^{t-1}}^{t-2} \setminus Z_{j^{t-1}}^{t-2} \in \mathcal{B}^{t-1}$. In Case 4, any $H \in \mathcal{B}^t \setminus \mathcal{B}^{t-1}$ is a subset of $M_{j^{t-1}}^{t-2}-g^{t-1} \in \mathcal{B}^{t-1}$. 
    
    {In Case 2.1, observe that  $J^t \in \mathcal{B}^{t-1}$ and $\{g^{t-1}\}$ is a singleton. Hence, any $H \in \mathcal{B}^t \setminus \mathcal{B}^{t-1}$ is a subset of $J^t \in \mathcal{B}^{t-1}$.}
    
    In Cases 2.2 and 2.3, {observe that $J \in \mathcal{B}^{t-1}$ and $R^{t-1} = X_{i^{t-1}}^{t-2} \setminus J^{t-1}$}. Hence, it holds that any $H \in \mathcal{B}^t \setminus \mathcal{B}^{t-1}$ is a subset of $(X_{j^{t-1}}^{t-2} \setminus Z_{j^{t-1}}^{t-2}) + g^{t-1} = R^{t-1}$ {or a subset of $J^t \in \mathcal{B}^{t-1}$}. 
    
    In Case 2.4, 2.5, and 2.6, {observe that  $J^{t-1} \in \mathcal{B}^{t-1}$ and $R^{t-1} = X_{i^{t-1}}^{t-2} \setminus J^{t-1}$.} 
    Hence, it holds that any $H \in \mathcal{B}^t \setminus \mathcal{B}^{t-1}$ is a subset of $S^{t-1} \subseteq R^{t-1}$ { or a subset of $J^{t-1} \in \mathcal{B}^{t-1}$ or a subset of $R^{t-1} \setminus S^{t-1}$}. Moreover, {if $H \subseteq S^{t-1}$, then $H$ must in fact be} a strict subset of $S^{t-1}$ because in cases 2.4 and 2.6, {the bundle} $S^{t-1}$ becomes a blue bundle, and in case 2.5, {the bundle} $S^{t-1}$ becomes a white bundle, and by the definition of $\mathcal{B}^t$, the set of available bundles includes only strict subsets of white and blue bundles.
\end{proof}

\end{document}

\fi